\documentclass{amsart}
\usepackage{mathrsfs}
\usepackage{graphicx}
\usepackage{amsmath}
\usepackage{amscd}
\usepackage{amssymb}

\input xy
\xyoption{all}

\newcommand{\pd}{\partial}
\newcommand{\bC}{{\mathbb C}}
 \newcommand{\bH}{{\mathbb H}}

\newcommand{\bP}{{\mathbb P}}

\newcommand{\bR}{{\mathbb R}}
\newcommand{\bZ}{{\mathbb Z}}

\newcommand{\cO}{{\mathcal O}}

\newcommand{\half}{\frac{1}{2}}

\newcommand{\dbar}{\bar{\partial}}

\newtheorem{theorem}{Theorem}[section]
\newtheorem{theorem/definition}{Theorem/Definition}[section]

\newtheorem{prop}{Proposition}[section]
\newtheorem{lm}{Lemma}[section]

\theoremstyle{remark}

\theoremstyle{definition}

\newcommand{\be}{\begin{equation}}
\newcommand{\ee}{\end{equation}}
\newcommand{\bea}{\begin{eqnarray}}
\newcommand{\ben}{\begin{eqnarray*}}
\newcommand{\een}{\end{eqnarray*}}
\newcommand{\eea}{\end{eqnarray}}
\newcommand{\bet}{\begin{equation}
\begin{split}}
\newcommand{\eet}{\end{split}
\end{equation}}

\usepackage{color}

\definecolor{yellow}{rgb}{1,1,0}
\definecolor{orange}{rgb}{1,.7,0}
\definecolor{red}{rgb}{1,0,0}
\definecolor{white}{rgb}{1,1,1}

\definecolor{A}{rgb}{.75,1,.75}

\begin{document}

\title
{On Geometry and Symmetry of Kepler Systems. I.}

\author{Jian Zhou}
\address{Department of Mathematical Sciences\\
Tsinghua University\\Beijing, 100084, China}
\email{jzhou@math.tsinghua.edu.cn}

\begin{abstract}
We study the Kepler metrics on Kepler manifolds from the point of view of Sasakian geometry 
and Hessian geometry.
This establishes a link between the problem of classical gravity
and the modern geometric methods in the study of AdS/CFT correspondence
in string theory.
\end{abstract}
\maketitle


\section{Introduction}


This is the first of a series of papers in which we make a revisit to the geometry and symmetry of
Kepler problem with the vantage point of view of quantum gravity, supergravity,
and string theory.

Time and again in history,
geometry and symmetry have been used to study problems in gravity,
in particular,
the two-body problem, aka the Kepler problem.
This tradition goes back to the ancient Greeks
who regarded the circles and spheres as the geometric objects with perfect symmetries
and used them to build their models of the universe.
Geometric objects with discrete symmetries, such as the five Platonic solids,
were assigned by Plato to earth, air, water, fire, and a heavenly fifth element.
In {\em Mysterium Cosmographicum}, published in 1596,
Kepler proposed a model of the solar system by relating
the five extraterrestrial planets known at that time to the five Platonic solids.
This work attracted the attention of Tycho Brahe who offered Kepler
a job as an assistant and hence also access to his empirical data,
based on which Kepler formulated his second law in 1602, first law in 1605,
and third law in 1618.
With the introduction of elliptic orbits,
the Kepler system seemed to lose the connection with the circular symmetry.
In Newton's {\em Principia} published in 1687,
Euclidean geometry was used to establish the connection
between the inverse square law of the universal gravity and Kepler's laws.
With the advance of calculus,
the Kepler system became one of the first examples of integrable systems,
i.e., those solvable by quadrature,
and geometry and symmetry seemed to lose its relevance to the Kepler problem.

An interesting attempt in the nineteenth century to resurrect geometry and symmetry 
in the study of the Kepler system in the nineteenth century
is Hamilton's work \cite{Ham} on hodograph published in 1846.
He showed that the velocity vector of the Kepler system moves on a circle.
Only in the twentieth century,
geometry and symmetry gradually regained their positions in the study of gravity.
First of all, Einstein's general relativity in 1915 is a geometric
theory of gravity;
secondly,
N\"other's theorem published in 1917
revealed the connection between symmetry and conservation laws.
Progresses were made first in the quantum Kepler system,
e.g., the quantum mechanical system of the hydrogen atom,
where hidden symmetry was first discovered.
In 1926, one year before Schr\"odinger  equation was discovered,
Pauli \cite{Pauli} solved the problem of finding the spectrum of the hydrogen atom
using an $o(4) = o(3) \oplus o(3)$-symmetry generated by the angular-momentum
operators
and the Laplace-Runge-Lenz operators.
In 1935,
Fock \cite{Fock} gave a geometric interpretation of the appearance of $o(4)$
by  working in the momentum space by a Fourier transform
followed by an embedding of $\bR^3$ into $S^3$ by stereographic projection.
After thirty years, in the 1960s,
the $o(4)$-symmetry of the quantum Kepler system was extended
to the $o(4,1)$- and $o(4,2)$-symmetry,
see e.g.
Bander and Itzykson \cite{Ban-Itz} and
Gy\"orgyi \cite{Gyo} and the references therein.
In Bander-Itzykson \cite{Ban-Itz} the results were generalized to
quantum Kepler system in arbitrary dimension.
Such studies of the symmetry of quantum Kepler system
also leads to the study of classical Kepler system,
for example in a 1965 paper in which Bacry \cite{Bac} tried to
construct the $o(4,1)$-symmetry
of the quantum Kepler system he started with the classical
Kepler system.
In another direction,
attentions were shifted from
infinitesimal symmetry described by Lie algebras
to global symmetry described by Lie group action.
In the 1966 paper  Bacry-Ruegg-Souriau \cite{Bac-Rue-Sou},
the global study of the hidden symmetry of the Kepler system was initiated.
Quantum aspects of the Kepler problem was discussed in a
book by Englefield \cite{Eng} published in 1972.

Since the 1970s modern differential geometry was applied to Kepler problem.
First Riemannian geometry was applied by
Moser \cite{Moser} to give a classical mechanics version of Fock's result.
He applied the idea of regularization of many-body system developed by
e.g. Levi-Civita \cite{LC}, Sundman \cite{Sun}, Siegel \cite{Sie}, Kustaanheimo-Stiefel \cite{KS}
to relate the solutions of negative energy of the Kepler system
to geodesic flow on the 3-sphere, and
to identify the phase space of orbits of negative energy of the Kepler
system to $K=T^+S^3 = T^*S^3 - S^3$.
Moser's results were generalized to the case of nonnegative energy by Belbruno \cite{Bel}
and Osipov \cite{Osi}.
For an exposition,
see Milnor \cite{Mil}.
See also Ligon-Schaaf \cite{Lig-Sch} and Heckman-de Laat \cite{Hec-deL}.

The manifold $K$ was called the {\em Kepler manifold} by Souriau who as a pioneer
of symplectic geometry contributed to the introduction or
the development of many important concepts,
such as the moment map, the coadjoint action and the coadjoint orbits,
geometric quantization, and he gave a classification of the homogeneous symplectic manifolds,
known as the Kirillov-Kostant-Souriau theorem.
Many of these developments can be found in his book \cite{Souriau97} first published in 1970.
He applied these developments in symplectic geometry to the Kepler problem in two papers,
Souriau \cite{Souriau74} and Souriau \cite{Souriau83},
where $K$ was shown to be a coadjoint orbit of $SO(4, 2)$,
and a complex structure of $K$ was revealed: $K$ is diffeomorphic to the conifold in $\bC^4$
with the conifold point removed.
The work of Moser and Souriau were nicely presented in a section in the 1977 book
by Guillemin and Sternberg \cite{Gui-Ste77}.
In a remarkable book published in 1990 \cite{Gui-Ste90},
these authors revisited the Kepler problem by adding more insights.
They established the connections of the Kepler problem to other theories in
mathematics and physics,
such as Howe dual pairs, orbit method of representations, Penrose's twistor theory, etc.
More precisely,
they showed that the Kepler manifold
can be realized either as a coadjoint orbit of the first factor
or as a reduction with respect to the second one in the Howe pairs
$(SO(2,4),SL(2,\bR))$ in $Sp(4,\bR)$  and $(U(2,2), U(1))$ in $Sp(8,\bR)$,
The orbit approach leads to representation of the Poincar\'e group,
while Marsden-Weinstein symplectic reduction
in stages identifies $K$
with the space of forward null geodesics on the conformal completion  of the Minkowski space,
and hence it is natural to discuss different models of the universe
whose symmetries are subgroups of $SO(2,4)$
and also Kostant's model based on the group $SO(4,4)$.
For later work in this direction,
see e.g. Keane-Barrett \cite{Kea-Bar}, Keane-Barrett-Simons \cite{Kea-Bar-Sim}
and Keane \cite{Kea}.
In 2004 Cordani \cite{Cor}
published a comprehensive book on Kepler problem,
which treated both the classical and quantum aspects, together with
perturbation theory.
For more recent work on Kepler problem
see e.g. Cariglia \cite{Car}.

In the 1960s there have appeared some generalizations of the Kepler problem
called the MICZ-Kepler problem \cite{MIC, Z}.
Meng discovered the generalizations to higher dimensions
and their connections to Euclidean Jordan algebra, symmetric cones of tube type,
and minimal representations.
For a survey,
see Meng \cite{Meng}.
Even though we will not discuss them in this paper,
it is interesting to generalize our results in this direction. 

From the above brief sketch of some aspects in the development of the Kepler problem,
it is clear that most of the researches in the literature
focus on the classical mechanics of the Kepler system
or quantum mechanics of the quantum Kepler system.
The purpose of this series of papers is to study the geometry and symmetry
of the classical Kepler system
by some techniques from general relativity, quantum field theory,
and string theory.
Our motivation is to test some important ideas in string theory
in the classical setting of the Kepler system.
They include: connection with the theory of integrable hierarchy,
Calabi-Yau spaces and mirror symmetry, AdS/CFT correspondence.

Since its early days, string theory was 
proposed as a candidate theory for the unification of gravity 
and the standard model of other fundamental forces.
The explosive development of string theory in the past thirty years or so
has brought forward a remarkable supply of new ideas and new techniques.
Our modest goal of this work is to apply a sample of ideas and techniques
developed in string theory to revisit the Kepler problem
in the hope of while
testing the modern developments to a classical problem,
new insights can be gained for both of them.
As it turns out,
such consideration leads us to many connections of the Kepler system with many different
mathematical objects overlooked in the literature.
In this paper,
we will make connections to Sasakian geometry and Hessian geometry.
In later parts of the series,
we will relate to Lie sphere geometry and integrable systems.

Our point of departure from
the geometry of the Kepler problem established
by Moser and Souriau in the 1970s
is to shift attentions to the K\"ahler metrics on the Kepler manifolds.
It was observed by Rawnsley \cite{Rawnsley} in 1977 that the complex structure 
and the symplectic structure on the Kepler manifolds are compatible,
leading to K\"ahler metrics on the Kepler manifolds.
We will refer to such metrics as {\em Kepler metrics}.

In this paper we will present two types of results for Kepler metrics:
Relationship with Sasakian geometry and relationship with Hessian geometry.
First,
we will show that the Kepler metrics 
are Sasaki metric on the cotangent bundles of the round spheres,
and they are related to Sasakian metrics on the unit conormal bundle of the round
spheres.
Secondly,
we will focus on Kepler metrics on Kepler $n$-manifold with $n=2$ and $3$.
In these cases,
we will use the symmetry of the metrics to present
explicit constructions of these metrics by Calabi Ansatz.
This will lead us a generalization of 
some work of Guillemin \cite{Gui}, Abreu \cite{Abreu}, Donaldson \cite{Donaldson}
to the noncompact case.
With the introduction of symplectic coordinates
we will establish a connection of Kepler metrics 
for $n=2,3$ with Hessian geometry.
We also present a treatment of K\"ahler Ricci-flat metrics on $n$-conifolds,
resolved $n$-conifolds in the same fashion.
We will propose a connection between such metrics using some special K\"ahler Ricci flow.

We arrange the rest of this paper as follows.
In \S \ref{sec:Reg} we recall the Levi-Civita regularization 
and the Moser regularization of hte Kepler problem,
and introduce the Kepler manifold $K_n$.
In \S \ref{sec:Complex} we recall the complex structures on $K_n$ introduced by Souriau. 
Some detailed discussions in the cases of $n=2$ and $3$ are also presented.
In \S \ref{sec:Kepler-Metric} we first recall  the Kepler metrics and their K\"ahler potentials,
then we  show that they are Sasaki metrics on the conormal bundles of the round spheres,
and relate them to Sasaki metrics on  the unit  conormal bundles of the round spheres,
which we show  are Sasakian.
In \S \ref{sec:KRF-Con-Deformed} we recall some explicit
constructions of K\"ahler Ricci-flat metrics on the $n$-conifolds
and the deformed $n$-conifolds.
Starting in \S \ref{sec:Hessian}
we will relate some explicit constructions of K\"ahler metrics with symmetries
with Hessian geometry.
In \S \ref{sec:Hessian} we treat the case of $U(n)$-symmetric K\"ahler metrics
on $\bC^n-\{0\}$.
We specialize to the case of $U(n)$-symmetric K\"ahler Ricci-flat metrics 
in \S \ref{sec:U(n)-symmetric},
and we further specialize to the case of Kepler metric on $K_2$ and Eguchi-Hanson
spaces in \S \ref{sec:E-H}.
Kepler metric on $K_3$, K\"ahler Ricci-flat metrics on 3-conifold and resolved 
3-conifold are studied from the point of view of K\"ahler metrics with 
symmetries and Hessian geometry in \S \ref{sec:CA-1} and \S \ref{sec:CA-2}
in various ways.
Finally,
in \S \ref{sec:Summary},
we summarize our results
and propose some generalizations.

\section{Regularizations of Kepler System and Kepler Manifolds}

\label{sec:Reg}

A first key step towards the global study of the geometry and symmetry
of the Kepler system is {\em regularization}.
This was first introduced by Levi-Civita \cite{LC} more than a hundred years
ago in his research on restricted three-body problem.
The idea was used by Sundman \cite{Sun} a few years later to solve
the three-body problem by power series.
The idea of regularization was further developed by
Siegel \cite{Sie} and Kustaanheimo-Stiefel \cite{KS},
and it was used by Moser \cite{Moser} to study the Kepler problem
in arbitrary dimensions.

In this Section we recall the Levi-Civita regularization of the two-dimensional
Kepler problem
and the Moser regularization of $n$-dimensional Kepler problem.
For details,
see e.g. Moser \cite{Moser}, Souriau \cite{Souriau74},
Guillemin-Sternberg \cite{Gui-Ste90} and Cordani \cite{Cor},
and the references therein.

\subsection{Levi-Civita regularization of two-dimensional Kepler system}

Let us recall the Levi-Civita regularization of the two-dimensional Kepler system:
\begin{align}
\frac{d^2 x}{dt^2} & = - \frac{x}{(x^2+y^2)^{3/2}}, & \frac{d^2y}{dt^2} & =  -\frac{y}{(x^2+y^2)^{3/2}}.
\end{align}
One first rewrites it as a Hamiltonian system:
\begin{align}
\frac{dx}{dt} & = p = \{H, x\}, & \frac{dy}{dt} & = q = \{H, q\}, \\
\frac{dp}{dt} & = - \frac{x}{\rho^3} = \{H, p\}, &
\frac{dq}{dt} & = - \frac{y}{\rho^3} = \{H, q\},
\end{align}
where $H = \half (p^2+q^2) - \frac{1}{\rho}$,
$\rho = (x^2+y^2)^{1/2}$,
and the Poisson bracket is defined by
\be
\{f, g\} : = \frac{\pd f}{\pd p} \frac{\pd g}{\pd x} - \frac{\pd f}{\pd x}\frac{\pd g}{\pd p}
+ \frac{\pd f}{\pd q} \frac{\pd g}{\pd y} - \frac{\pd f}{\pd y}\frac{\pd g}{\pd q}.
\ee
In other words,
the phase space of the two-dimensional Kepler problem is
\be
\{(x,y, p, q) \in \bR^4\;\;|\;\; (x, y) \neq (0,0)\}
= (\bR^2 - \{0\}) \times \bR^2,
\ee
which is endowed with a symplectic form
\be
\omega = dp \wedge dx + dq \wedge dy.
\ee

To remove the singularity of the system,
one first introduces the fictitious time $\tau$ such that
\be
\frac{d\tau}{dt} = \frac{1}{\rho},
\ee
and make the following change of variables:
\begin{align} \label{eqn:LC1}
x + i y & = (\xi + i \eta)^2, & p - i q = \frac{\omega - i \chi}{2(\xi + i \eta)},
\end{align}
then the two-dimensional system becomes:
\begin{align}
\frac{d}{d\tau} (\xi + i \eta) & = \frac{1}{4} (\omega + i \chi),  \\
\frac{d}{d\tau} (\omega +i\chi) & = \frac{1}{4} \frac{\omega^2+\chi^2}{\xi - i \eta}
- \frac{2}{\xi - i \eta}.
\end{align}
Next, note that
\be
H = \half |p+iq|^2 - \frac{1}{|x+iy|}
= \frac{1}{4} \frac{|\omega + i \chi|^2}{|\xi + i \eta|^2}
- \frac{1}{|\xi+i\eta|^2}.
\ee
Hence by conservation of energy,
one can restrict to the energy surface $H = E$,
where one has
\be
|\omega+i\chi|^2 = 8 (1 + E |\xi+i\eta|^2),
\ee
and so the equation of motion becomes:
\begin{align}
\frac{d}{d\tau} (\xi + i \eta) & = \frac{1}{4} (\omega + i \xi),  \\
\frac{d}{d\tau} (\omega + i \chi) & = 2E \cdot (\xi + i \eta).
\end{align}
Now we restrict to the case of $E < 0$
and convert the above three equations into the following form:
\begin{align}
& \frac{d}{d\tilde{\tau}} (\tilde{\xi} + i \tilde{\eta})
= (\tilde{\omega} + i \tilde{\xi}),  \\
& \frac{d}{d\tilde{\tau}} (\tilde{\omega} + i \tilde{\chi})
 = 2E \cdot (\tilde{\xi} + i \tilde{\eta}), \\
& |\tilde{\xi}+i\tilde{\eta}|^2 + |\tilde{\omega}+i\tilde{\chi}|^2 = -E.
\end{align}
by a further change of variable:
\begin{align}
\tilde{\tau} & = \sqrt{\frac{-E}{2}}\tau, &
\tilde{\xi}+ i \tilde{\eta} & = -E (\xi + i \eta), &
\tilde{\omega} + i \tilde{\chi} & = \sqrt{\frac{-E}{8}} (\omega + i \chi).
\end{align}

\subsection{The Moser regularization}

Moser \cite{Moser} extended the inverse map of the stereographic projection
$\bR^n \to \bR^{n+1}$ defined by
\be \label{eqn:Moser1}
\vec{y} = (y_1, \dots, y_n) \mapsto (x_0 = \frac{|\vec{y}|^2-1}{|\vec{y}|^2+1},
\frac{2y_1}{|\vec{y}|^2+1}, \dots, \frac{2y_n}{|\vec{y}|^2+1}),
\ee
to a map $T^*\bR^n \to T^*\bR^{n+1}$
\be
(\vec{y}, \vec{\eta}) \mapsto (\vec{x}, \vec{\xi}),
\ee
called the Moser map, where $\vec{\xi}=(\xi_0, \dots, \xi_n)$ is defined by:
\begin{align} \label{eqn:Moser2}
\xi_0 & = \vec{\eta} \cdot \vec{y}, &
\xi_j & = \frac{|\vec{y}|^2+1}{2} \cdot \eta_j - (\vec{\eta} \cdot \vec{y}) \cdot y_j, \;\;\;\;
j=1, \dots, n.
\end{align}
One can check that
\begin{align}
\vec{x} \cdot \vec{x} & = 1, & \vec{x} \cdot \vec{\xi} & = 0.
\end{align}
Furthermore,
\bea
&& \vec{\xi} \cdot d \vec{x} = \vec{\eta} \cdot d \vec{y}, \label{eqn:Moser} \\
&& \eta_j = (1-x_0) \xi_j + \xi_0 x_j, \\
&& |\vec{\xi}| = \frac{|\vec{y}|^2+1}{2} \cdot |\vec{\eta}|.
\eea
The Moser map defines an inclusion of $(\bR^n -\{0\}) \times \bR^n$
into the Kepler manifold:
\be
K_n:=\{(\vec{x}, \vec{\xi}) \in \bR^{n+1} \times (\bR^{n+1}-\{0\})\;|\;
\vec{x} \cdot \vec{x}= 1, \;\; \vec{x} \cdot \vec{\xi} = 0 \}.
\ee
This space is naturally identified with $T^+S^n := T^*S^n - S^n$.

One can check that $\Phi = \half |\vec{\xi}|^2|\vec{x}|^2$,
restricted to $T^*S^n$,
generates the geodesic flow on $T^*S^n$.
Using the Moser map and the introduction of the fictitious time $\tau$,
the energy surface $\Phi = \half$ corresponds to $H = - \frac{1}{2}$.
Using the Lie scaling
\begin{align}
q & \mapsto \lambda^2 q, & p & \mapsto \lambda^{-1} p, & t & \mapsto \lambda^3 t,
\end{align}
the energy surface $\Phi = \half \lambda^2$ corresponds to $H = - \frac{1}{2\lambda^2}$.

\section{Complex Structures on Kepler Manifolds}

\label{sec:Complex}

The reason that in last Section we single out the two-dimensional Kepler problem
and Levi-Civita regularization before presenting the more general Moser regularization
is because it naturally leads to a construction of a complex structure on $K_2$.
Complex structures by a different construction
on general $K_n$ were discovered by Souriau \cite{Souriau74}.
We will recall both of these constructions in this Section
and show that they give a hypercomplex structure on $K_2$,
in particular,
it is doubly covered by $\bC^2-\{0\}$.
We also recall the complex structure on $K_3$.

\subsection{Complex structures on $K_2$}
\label{sec:K2}

Let us now point out some algebraic geometry implicit in the process of Levi-Civita regularization.
This leads us to the discussions of complex structures on $K_n$ in the next subsection.

The original phase space of the two-dimensional Kepler problem is $(\bR^2 - \{(0,0)\} \times \bR^2$,
with coordinates $(x,y, p,q)$.
Using complex coordinates $(x+iy, p-iq)$ on this space,
one can identify it with $\bC^* \times \bC$.
The Levi-Civita regularization suggests to introduce $z_1, z_2$ such that:
\begin{align} \label{eqn:LC2}
x+iy &= z_1^2, & p-iq &= \frac{z_2}{z_1}.
\end{align}
One can understand $(z_1,z_2)$ as linear coordinates on $\bC^2$
and $(x + iy, p-iq)$ as local coordinates on $\cO_{\bP^1}(-2)$.
Indeed, we have the following diagram:
\be \label{eqn:Diagram}
\xymatrix{
 \cO_{\bP^1}(-1) \ar[d] \ar[r]^{2:1} &  \cO_{\bP^1}(-2) \ar[d] \\
\bC^2 \ar[r] & \bC^2/\bZ_2 }
\ee
Let  $(\alpha_1, \beta_1)$, $(\alpha_2, \beta_2)$ be local coordinates on $\cO_{\bP^1}(-1)$
such that
\begin{align*}
\alpha_2 & = \frac{1}{\alpha_1}, & \beta_2 & = \alpha_1\beta_1,
\end{align*}
and $(\alpha_1, \gamma_1)$, $(\alpha_2, \gamma_2)$ be local coordinates on $\cO_{\bP^1}(-2)$
such that
\begin{align*}
\alpha_2 & = \frac{1}{\alpha_1}, & \gamma_2 & = \alpha_1^2\gamma_1,
\end{align*}
then the upper horizontal map is given by:
\begin{align}
\gamma_1 &= \beta_1^2, & \gamma_2 & = \beta_2^2,
\end{align}
and the left vertical map is given by
\begin{align}
z_1 & = \beta_1 = \alpha_2\beta_2, & z_2 & = \alpha_1\beta_1 = \beta_2.
\end{align}
One can then see that
\begin{align}
\alpha_1 & = \frac{z_2}{z_1}, & \gamma_1 & = z_1^2,
\end{align}
and so one can get
\begin{align}
x+iy & = \gamma_1, & p-iq &  = \alpha_1,
\end{align}
and with this identification one identifies
$\bC^* \times \bC$ as a coordinate chart on the Kepler manifold $K_2= \cO_{\bP^1}(-2) - \bP^1$.
For those points in $K_2$ not in the image of inclusion of $\bC^* \times \bC$,
$\alpha_2 = 0$ and so $\alpha_1 = \infty$,
and the regularization procedure enlarge the original phase space by including some
points with infinite momentum,
and for the solution space, add some orbits that correspond to collision solutions.

The Kepler manifold $K_2$ is diffeomorphic to $T^+S^2 = T^*S^2 - S^2$.
The above discussion suggests us to regard $S^2$ as $\bP^1$,
then one can identify $K_2$ as the holomorphic cotangent bundle of $\bP^1$
minus the zero section,
i.e. $\cO_{\bP^1}(-2) - \bP^1$.
Another way to understand
this space is to examine the vertical map
on the right in \eqref{eqn:Diagram}.
The quotient space $\bC^2/\bZ_2$ can be understood as an affine variety in $\bC^3$,
defined by the equation
\be
u v - w^2 = 0,
\ee
and the lower horizontal map in \eqref{eqn:Diagram} is given by
\be
(z_1, z_2) \mapsto (u = z_1^2, v= z_2^2, w = z_1z_2).
\ee
Define new coordinates $(w_1, w_2, w_3)$ on $\bC^3$ by:
\begin{align}
u & = w_0 + i w_1, & v &= w_0 - i w_1, & w_2 & = i w,
\end{align}
then one gets the equation of the two-dimensional conifold $C_2$
\be
w_0^2 + w_1^2 + w_2^2 = 0.
\ee
The vertical map on the right in \eqref{eqn:Diagram}
is defined by:
\begin{align}
u & = \gamma_1 = \alpha_2^2\gamma_2, &
v & = \alpha_1\gamma_1 = \alpha_2 \gamma_2, &
w & = \alpha_1^2\gamma_1 = \gamma_2.
\end{align}
Therefore, $K_2$ can also be identified with $C_2^*$, the two-dimensional conifold
with the conifold point removed.

\subsection{Complex structures on $K_n$} \label{sec:Kn}

The introduction of the complex structures on $K_2$ in last subsection
relies on the identification of $S^2$ with $\bP^1$
and makes use of the complex structure on $S^2$,
and hence it cannot be generalized to $K_n$.
Nevertheless,
the diffeomorphism from $K_2$ to $C_2^*:=C_2 - \{0\}$ suggests a
diffeomorphism of $K_n$ with $C_n^*:=C_n-\{0\}$ 
first discovered by Sourian \cite{Souriau74},
where  $C_n$ is the $n$-conifold in $\bC^{n+1}$ defined by the equation:
\be
w_0^2+ w_1^2 + \cdots + w_n^2 = 0.
\ee
For $(w_0, w_1, \dots, w_n) \in C_n -\{0\}$,
write $w_j = u_j + i v_j$, $u_j, v_j \in \bR$, $j=0, 1, \dots, n$.
Then we have
\ben
&& u_0^2+ u_1^2 + \cdots + u_n^2 = v_0^2+ v_1^2 + \cdots + v_n^2, \\
&& u_0v_0 + u_1v_1 + \cdots + u_nv_n = 0.
\een
Since $(w_0, w_1, \dots, w_n) \neq 0$,
none of $\vec{u} = (u_0, u_1, \dots, u_n)$ and $\vec{v} = (v_0, v_1, \dots, v_n)$ is zero.
A diffeomorphism  $\varphi: C_n -\{0\} \to K_n$ can be defined as follows:
\be
\vec{w} = (w_0, w_1, \dots, w_n) \mapsto (\vec{p}, \vec{q}) = (\frac{\vec{x}}{|\vec{x}|}, \vec{y}),
\ee
with its inverse map given by:
\be
(\vec{p}, \vec{q}) \mapsto \vec{w} = |\vec{q}| \cdot \vec{p} + i \vec{q}.
\ee

\subsection{Hypercomplex structure on $K_2$}

In the last two subsections we have constructed two different complex structures on $K_2$.
Let us write them down explicitly in the local coordinates $(x, y, p,q)$.
For $J_1$ defined in \S \ref{sec:K2}, it is clear that
\begin{align*}
J_1 \frac{\pd}{\pd x} & = \frac{\pd}{\pd y}, &
J_1 \frac{\pd}{\pd y} & = - \frac{\pd}{\pd x}, &
J_1 \frac{\pd}{\pd p} & = - \frac{\pd}{\pd q}, &
J_1 \frac{\pd}{\pd q} & = \frac{\pd}{\pd p}.
\end{align*}
Using the diffeomorphism $\varphi$ defined in \S \ref{sec:Kn},
we get a complex structure $J_2$ .
We first use the Moser map   to get
\begin{align*}
\xi_0 & = px+qy, & \xi_1 & = \frac{p^2+q^2+1}{2}x - (px+qx)p, &
\xi_2 & = \frac{p^2+q^2+1}{2}y- (px+qy)q, \\
x_0 & = \frac{p^2+q^2-1}{p^2+q^2+1}, &
x_1 & = \frac{2p}{p^2+q^2+1}, &
x_2 & = \frac{2q}{p^2+q^2+1}.
\end{align*}
and\emph{} we have
\be
|\xi| = \frac{p^2+q^2+1}{2} \sqrt{x^2+y^2},
\ee
so we get
\ben
&& w_0 = \frac{\sqrt{x^2+y^2}}{2}(p^2+q^2-1) + i(px+qy), \\
&& w_1 = \sqrt{x^2+y^2}p + i \biggl( \frac{p^2+q^2+1}{2}x- (px+qy)p \biggr), \\
&& w_2 = \sqrt{x^2+y^2}q + i \biggl( \frac{p^2+q^2+1}{2}y- (px+qy)q \biggr).
\een
We will take the open subset of $C_n-\{0\}$ on which $w_1, w_2$ can be used local coordinates.
Let $w_j = u_j + i v_j$, $u_j, v_j \in \bR$, $j=0,1, 2$.
Then we have
\begin{align*}
J_2 \frac{\pd}{\pd u_1} & = \frac{\pd}{\pd v_1}, &
J_2 \frac{\pd}{\pd v_1} & = - \frac{\pd}{\pd u_1}, \\
J_2 \frac{\pd}{\pd u_2} & = \frac{\pd}{\pd v_2}, &
J_2 \frac{\pd}{\pd v_2} & = - \frac{\pd}{\pd u_2}.
\end{align*}
One can check that
\begin{align*}
J_2 \frac{\pd}{\pd x} & = \frac{px-qy}{(x^2+y^2)^{1/2}}\frac{\pd}{\pd x}
+ \frac{yp+xq}{(x^2+y^2)^{1/2}}\frac{\pd}{\pd y}
-\frac{(p^2+q^2+1)}{2(x^2+y^2)^{1/2}} \frac{\pd}{\pd p}, \\
J_2 \frac{\pd}{\pd y} & = \frac{yp+xq}{(x^2+y^2)^{1/2}} \frac{\pd}{\pd x}
- \frac{px-qy}{(x^2+y^2)^{1/2}}\frac{\pd}{\pd y}
- \frac{p^2+q^2+1}{2(x^2+y^2)^{1/2}} \frac{\pd}{\pd q}, \\
J_2 \frac{\pd}{\pd p} & = 2(x^2+y^2)^{1/2} \frac{\pd}{\pd x}
-\frac{px-qy}{(x^2+y^2)^{1/2}} \frac{\pd}{\pd p}
- \frac{yp+xq}{(x^2+y^2)^{1/2}} \frac{\pd}{\pd q}, \\
J_2 \frac{\pd}{\pd q} & = 2(x^2+y^2)^{1/2} \frac{\pd}{\pd y}
-\frac{yp+xq}{(x^2+y^2)^{1/2}}\frac{\pd}{\pd p}
+ \frac{px-qy}{(x^2+y^2)^{1/2}}\frac{\pd}{\pd q},
\end{align*}
One can define a third almost complex structure $J_3$ by
\be
J_3 = J_1J_2.
\ee
Indeed one can check that
\begin{align*}
J_3 \frac{\pd}{\pd x} & = \frac{px-qy}{(x^2+y^2)^{1/2}}\frac{\pd}{\pd y}
- \frac{yp+xq}{(x^2+y^2)^{1/2}}\frac{\pd}{\pd x}
+ \frac{(p^2+q^2+1)}{2(x^2+y^2)^{1/2}} \frac{\pd}{\pd q}
= - J_2 \frac{\pd}{\pd y} = - J_2J_1 \frac{\pd}{\pd x}, \\
J_3 \frac{\pd}{\pd y} & = \frac{yp+xq}{(x^2+y^2)^{1/2}} \frac{\pd}{\pd y}
+ \frac{px-qy}{(x^2+y^2)^{1/2}}\frac{\pd}{\pd x}
- \frac{p^2+q^2+1}{2(x^2+y^2)^{1/2}} \frac{\pd}{\pd p}
= J_2 \frac{\pd}{\pd x} = - J_2J_1 \frac{\pd}{\pd y}, \\
J_3 \frac{\pd}{\pd p} & = 2(x^2+y^2)^{1/2} \frac{\pd}{\pd y}
+ \frac{px-qy}{(x^2+y^2)^{1/2}} \frac{\pd}{\pd q}
- \frac{yp+xq}{(x^2+y^2)^{1/2}} \frac{\pd}{\pd p} = J_2\frac{\pd}{\pd q}
= -J_2J_1\frac{\pd}{\pd p}, \\
J_3 \frac{\pd}{\pd q} & = - 2(x^2+y^2)^{1/2} \frac{\pd}{\pd x}
+ \frac{yp+xq}{(x^2+y^2)^{1/2}}\frac{\pd}{\pd q}
+ \frac{px-qy}{(x^2+y^2)^{1/2}}\frac{\pd}{\pd p} = - J_2 \frac{\pd}{\pd p}
= - J_2J_1\frac{\pd}{\pd q},
\end{align*}
and from these identities it is easy to see that
\begin{align}
J_3 &= -J_2J_1, & J_3^2 & = - 1.
\end{align}
By the following Lemma,
$J_1, J_2, J_3$ make $K_2$ a hypercomplex manifold.

\begin{lm}
Suppose that $J_1$ and $J_2$ are two integrable almost complex structures on a manifold $M$,
and if $J_3 = J_1J_2$ satisfies $J_3 = - J_2J_1$ and $J_3^2 = -1$,
then $J_3$ is also integrable.
\end{lm}

\begin{proof}
By Newlander-Nirenberg theorem,
an almost complex structure $J$ is integrable iff its Nijenhuis tensor vanishes,
i.e.,
\be
N_J(X, Y) = [X, Y] +J[JX, Y] +J[X, JY] - [JX, Jy] = 0
\ee
for any vector fields $X, Y$ on $M$.
By conditions we have
\ben
&& N_{J_1}(X, Y) = [X, Y] +J_1[J_1X,Y]+J_1[X,J_1Y] - [J_1X,J_1Y] = 0, \\
&& N_{J_2}(X, Y) = [X, Y] +J_2[J_2X,Y]+J_2[X,J_2Y] - [J_2X,J_2Y] = 0.
\een
One can derive that
\be \label{eqn:NJ1}
N_{J_3}(X, Y) = N_{J_3}(J_1X, J_2Y)
\ee
by the following computations:
\ben
N_{J_3}(X, Y)
& = & [X, Y] +J_3[J_3X,Y]+J_3[X,J_3Y] - [J_3X,J_3Y] \\
& = & [X, Y] +J_1J_2[J_1J_2X,Y]+J_1J_2[X,J_1J_2Y] - [J_1J_2X,J_1J_2Y]  \\
& = & [X, Y] +J_1J_2[J_1J_2X,Y]+J_1J_2[X,J_1J_2Y]  \\
& - & ([J_2X,J_2Y]+J_1[J_1J_2X,J_2Y]+J_1[J_2X, J_1J_2Y]) \\
& = & [X, Y] - [J_2X,J_2Y] +J_1J_2[J_1J_2X,Y]+J_1J_2[X,J_1J_2Y]  \\
& + & J_1[J_2J_1X,J_2Y]+J_1[J_2X, J_2J_1Y] \\
& = & [X, Y] - [J_2X,J_2Y] +J_1J_2[J_1J_2X,Y]+J_1J_2[X,J_1J_2Y]  \\
& + & J_1([J_1X,Y]+J_2[J_2J_1X,Y]+J_2[J_1X,J_2Y]) \\
& + & J_1([X, J_1Y] + J_2[J_2X,J_1Y]+J_2[X,J_2J_1Y]) \\
& = & [X, Y] - [J_2X,J_2Y]
+ J_1([J_1X,Y]+J_2[J_1X,J_2Y]) \\
& + &  J_1([X, J_1Y] + J_2[J_2X,J_1Y]) \\
& = & [J_1X, J_1Y] - [J_2X, J_2Y]
+ J_1J_2[J_1X,J_2Y] + J_1J_2[J_2X,J_1Y] \\
& = & N_{J_3}(J_1X, J_1Y).
\een
In the same fashion,
one gets
\be \label{eqn:NJ2}
N_{J_3}(X, Y) = N_{J_3}(J_2X, J_2Y).
\ee
Combining \eqref{eqn:NJ1} and \eqref{eqn:NJ2},
one gets:
\be
N_{J_3}(X, Y) = N_{J_3}(J_3, J_3Y).
\ee
Note $N_{J_3}(J_3X, J_3Y) = - N_{J_3}(X, Y)$,
so we get
\be
N_{J_3}(X, Y) = 0.
\ee
\end{proof}

Let us now understand the hypercomplex structure on $K_2$ from
\eqref{eqn:LC2}.
This formula gives a $2:1$ covering map
from $\bR^4-\{0\}$ to $K_2$.
We will call this map the {\em Levi-Civita map}.
Each identification of $\bR^4$ with the space
$\bH$ of quaternion numbers gives $\bR^4$ a hypercomplex structure.

\begin{prop}
Let $\xi, \eta, \omega, \chi$ be linear coordinates on $\bR^4$,
consider the $2:1$ map from $\bR^4-\{0\} \to K_2$
defined in the $(x, y, p, q)$-coordinate patch by:
\begin{align}
x + i y & = (\xi + i \eta)^2, &
p - i q = \frac{\omega - i \chi}{\xi + i \eta}.
\end{align}
then the hypercomplex structure $J_1, J_2, J_3$ corresponds to
the following hypercomplex structure on $\bR^4$:
\begin{align}
J_1 \pd_\xi & = \pd_\eta, & J_1\pd_\eta & = - \pd_\xi, &
J_1\pd_\omega & = - \pd_\chi, & J_1 \pd_\chi & = \pd_\omega, \label{eqn:J1} \\
J_2 \pd_\xi & = -\pd_\omega, & J_2\pd_\eta & = - \pd_\chi, &
J_2\pd_\omega & = \pd_\xi, & J_2 \pd_\chi & = \pd_\eta, \label{eqn:J2} \\
J_3 \pd_\xi & = \pd_\chi, & J_3\pd_\eta & = - \pd_\omega, &
J_3\pd_\omega & = \pd_\eta, & J_2 \pd_\chi & = -\pd_\xi. \label{eqn:J3}
\end{align}

\end{prop}

\begin{proof}
First we have
\begin{align*}
x & = \xi^2 - \eta^2, & y & = 2\xi \eta, &
p &= \frac{\xi \omega -\eta\chi}{\xi^2+\eta^2}, &
q & = \frac{\xi \chi + \eta \omega}{\xi^2+\eta^2}.
\end{align*}
It follows that
\ben
&& \pd_\xi = 2\xi \pd_x + 2\eta \pd_y
+ \frac{-\omega\xi^2+\omega \eta^2+2\xi\eta\chi}{(\xi^2+\eta^2)^2} \pd_p
-\frac{\chi \xi^2-\chi \eta^2+2\xi\eta \omega}{(\xi^2+\eta^2)^2} \pd_q, \\
&& \pd_\eta = -2\eta \pd_x + 2\xi \pd_y
- \frac{\chi \xi^2-\chi \eta^2+2\xi\eta \omega}{(\xi^2+\eta^2)^2} \pd_p
- \frac{-\omega\xi^2+\omega \eta^2+2\xi\eta\chi}{(\xi^2+\eta^2)^2} \pd_q, \\
&& \pd_\omega = \frac{\xi}{\xi^2+\eta^2} \pd_p +\frac{\eta}{\xi^2+\eta^2} \pd_q, \\
&& \pd_\chi = \frac{-\eta}{\xi^2+\eta^2} \pd_p +\frac{\xi}{\xi^2+\eta^2} \pd_q.
\een
It is then straightforward to check \eqref{eqn:J1}.
To check \eqref{eqn:J2}, it suffices to note:
\ben
&& J_2 \pd_\omega = \frac{\xi}{\xi^2+\eta^2}J_2 \pd_p
+\frac{\eta}{\xi^2+\eta^2} J_2\pd_q \\
& = & \frac{\xi}{\xi^2+\eta^2} \biggl(2(x^2+y^2)^{1/2} \frac{\pd}{\pd x}
-\frac{px-qy}{(x^2+y^2)^{1/2}} \frac{\pd}{\pd p}
- \frac{py+qx}{(x^2+y^2)^{1/2}} \frac{\pd}{\pd q}
\biggr) \\
& + & \frac{\eta}{\xi^2+\eta^2} \biggl(2(x^2+y^2)^{1/2} \frac{\pd}{\pd y}
- \frac{yp+xq}{(x^2+y^2)^{1/2}}\frac{\pd}{\pd p}
+ \frac{xp-yq}{(x^2+y^2)^{1/2}}\frac{\pd}{\pd q} \biggr) \\
& = & 2\xi\pd_x + 2\eta \pd_y
+ \frac{-\omega\xi^2+\omega\eta^2+2\xi\eta\chi}{(\xi^2+\eta^2)^2}\pd_p
- \frac{\chi\xi^2-\chi\eta^2+2\xi\eta\omega}{(\xi^2+\eta^2)^2} \pd_q \\
& = & \pd_\xi,
\een

\ben
&& J_2\pd_\chi = \frac{-\eta}{\xi^2+\eta^2} J_2\pd_p
+\frac{\xi}{\xi^2+\eta^2} J_2\pd_q \\
& = & \frac{-\eta}{\xi^2+\eta^2} \biggl(2(x^2+y^2)^{1/2} \frac{\pd}{\pd x}
- \frac{px-qy}{(x^2+y^2)^{1/2}} \frac{\pd}{\pd p}
- \frac{py+qx}{(x^2+y^2)^{1/2}} \frac{\pd}{\pd q}
\biggr) \\
& + & \frac{\xi}{\xi^2+\eta^2} \biggl(2(x^2+y^2)^{1/2} \frac{\pd}{\pd y}
- \frac{py+qx}{(x^2+y^2)^{1/2}}\frac{\pd}{\pd p}
+ \frac{px-qy}{(x^2+y^2)^{1/2}}\frac{\pd}{\pd q} \biggr) \\
& = & -2\eta \pd_x + 2\xi \pd_y
- \frac{wx^2-wy^2+2xyz}{2(x^2+y^2)^2} \pd_P
- \frac{-zx^2+zy^2+2xyw}{2(x^2+y^2)^2} \pd_Q \\
& = & \pd_\eta.
\een
Finally \eqref{eqn:J3} follows from \eqref{eqn:J1} and \eqref{eqn:J2}.
\end{proof}

\subsection{Kepler manifold $K_3$ as a complex manifold}

In this subsection we recall some well-known resolutions
of  the conifold singularity of $C_3$.
Since $K_3 \cong C_3 -\{0\}$,
one can use the resolutions to describe $K_3$.

Blow up $\bC^4$ at the origin, and consider the strict transform
of the quadric conifold $C_3$.
The singular point by the smooth quadric in $\bP^3$
\be
\{[w_0:w_1:w_2:w_3]\in \bP^3\;\;|\;\;w_0^2+w_1^2+w_2^2+w_3^2 = 0\},
\ee
which is a copy of $\bP^1 \times \bP^1$.
The total space of strict transform is isomorphic to the total space
$\cO_{\bP^1 \times \bP^1}(-1, -1)$,
therefore one gets an isomorphism
\be
K_3 \cong \cO_{\bP^1 \times \bP^1}(-1, -1) - \bP^1 \times \bP^1.
\ee
For later use, we need to explicit write down such an isomorphism.
First let
\begin{align} \label{eqn:z-in-w}
z_1 & = w_0 + i w_1, & z_2 & = w_0 - i w_1, &
z_3 & = i w_2+w_3, & z_4 & = i w_2 - w_3.
\end{align}
Then one has
\be \label{eqn:z-conifold}
z_1z_2 = z_3z_4.
\ee
On $K_3$, all of $z_1, \dots, z_n$ are not zero,
and one can define four local coordinated patches with local coordinates as follows:
\begin{align}
\alpha_{1,1} & = \frac{z_3}{z_1}, & \alpha_{1,2} & = \frac{z_4}{z_1},
& \beta_1 & = z_1 \neq 0, \\
\alpha_{2,1} & = \frac{z_4}{z_2}, & \alpha_{2,2} & = \frac{z_3}{z_2},
& \beta_2 & = z_2 \neq 0, \\
\alpha_{3,1} & = \frac{z_1}{z_3}, & \alpha_{3,2} & = \frac{z_2}{z_3},
& \beta_3 & = z_3 \neq 0, \\
\alpha_{4,1} & = \frac{z_2}{z_4}, & \alpha_{4,2} & = \frac{z_1}{z_4},
& \beta_4 & = z_4 \neq 0.
\end{align}
These exactly correspond to four local coordinate patches on $\cO_{\bP^1\times \bP^1}(-1, -1)$.
This resolution is not a minimal one.
There are two different minimal resolutions obtained from
this one blowing down along each copy of $\bP^1$ in $\bP^1 \times \bP^1$,
and they are said to be related to each other by flop.
Correspondingly,
there are two different ways to get an isomorphism of the form:
\be
K_3 \cong \cO_{\bP^1}(-1) \oplus \cO_{\bP^1}(-1) - \bP^1.
\ee
Indeed,
one can define the following local coordinate patches on $K_3$:
\begin{align}
\alpha_{1} & = \frac{z_3}{z_1}, & \beta_{1,1} & = z_1 \neq 0,
& \beta_{1,2} & = z_4, \\
\alpha_{2} & = \frac{z_4}{z_2}, & \beta_{2,1} & = z_3,
& \beta_{2,2} & = z_2 \neq 0,
\end{align}
they correspond to two local coordinate patches on $\cO_{\bP^1}(-1) \oplus \cO_{\bP^1}(-1)$,
or one can take alternatively:
\begin{align}
\alpha_{3} & = \frac{z_1}{z_3}, & \beta_{3,2} & = z_2,
& \beta_{3,2} & = z_3 \neq 0, \\
\alpha_{4} & = \frac{z_2}{z_4},
& \beta_{4,1} & = z_4 \neq 0, & \beta_{4,2} & = z_1.
\end{align}

For later use,
we need the following:

\begin{prop} \label{prop:Norm}
On $K_3$ the following identities hold:
\be
\sum_{j=0}^3 |w_j|^2= \half \sum_{j=1}^4 |z_j|^2,
\ee
for $k=1,\dots,4$,
\be
(1+|\alpha_k|^2) (|\beta_{k,1}|^2+|\beta_{k,2}|^2) =
(1+|\alpha_{k,1}|^2) (1+ |\alpha_{k,2}|^2) |\beta_k|^2
= \sum_{j=1}^4 |z_j|^2.
\ee
\end{prop}

\begin{proof}
Write $w_j$ and $z_j$ in terms of their real parts and imaginary parts:
$w_j = u_j + \sqrt{-1} v_j$,
$z_j = x_j + \sqrt{-1} y_j$.
By \eqref{eqn:z-in-w},
\begin{align*}
x_1 & = u_0 - v_1, & y_1 & = v_0 + u_1, & x_2 & = u_0 + v_1, & y_2 & = v_0 - u_1, \\
x_3 & = - v_2 + u_3, & y_3 & = u_2 + v_3, & x_4 & = - v_2 - u_3, & y_4 & = u_2 - v_3,
\end{align*}
therefore,
\ben
\sum_{j=1}^4 |z_j|^2 & = & \sum_{j=1}^4 (x_j^2+y_j^2)
= \sum_{j=0}^3 (2u_j^2+2v_j^2) = 2 \sum_{j=0}^3 |w_j|^2.
\een
This proves the first formula.
The second formula follows easily from \eqref{eqn:z-conifold}.
\end{proof}

\section{Kepler Metrics on Kepler Manifolds}

\label{sec:Kepler-Metric}

In this Section we first recall that the symplectic structures
on the Kepler manifolds are compatible with their complex structures,
hence one obtains K\"ahler structures on the Kepler manifolds.
We will call these metrics the {\em Kepler metrics}
and show that they are Sasaki metrics on the conormal bundles of the round spheres.
We will also relate the Kepler metrics to Sasaki metrics
on  the unit  conormal bundles of the round spheres,
which we show  are Sasakian.

\subsection{Symplectic structure on $K_2$}

The symplectic form on the original phase space $(\bR^2-\{(0,0)\}) \times \bR^2$
is
\be
\omega = d p \wedge dx + dq \wedge dy.
\ee
In the complex coordinate $(\alpha_1, \gamma_1)$
it is the real part of the holomorphic volume form
\be
\Omega = d \alpha_1 \wedge d \gamma_1.
\ee
This form is defined everywhere on $\cO_{\bP^1}(-2)$ and is nowhere vanishing,
indeed, in the other coordinate patch with local coordinate $(\alpha_2, \gamma_2)$,
one has
\be
\Omega = - d \alpha_2 \wedge d \gamma_2.
\ee
By taking the real part of $\Omega$,
one gets a symplectic form, still denoted by $\omega$,
on the whole space of $K_2$.
By an easy computation we get

\begin{prop}
The pullback to $\bR^4$ of symplectic form $\omega$ on $K_2$ is
\be
\varphi_{LC}^*\omega = 2d\omega \wedge d \xi + 2 d\chi \wedge d \eta.
\ee
\end{prop}

\subsection{Hyperk\"ahler metric on Kepler manifold $K_2$}

On $K_2$ define $g$ by
\be \label{def:g-from-omega}
g(X, Y) = \omega(X, J_2Y).
\ee
In the local coordinate patch with local coordinates $(x, y, p, q)$ we have:
\ben
g & = & \frac{p^2+q^2+1}{2(x^2+y^2)^{1/2}}dx^2+\frac{2(px-qy)}{(x^2+y^2)^{1/2}}dxdp
+ \frac{2(py+qx)}{(x^2+y^2)^{1/2}}dxdq \\
& + & \frac{p^2+q^2+1}{2(x^2+y^2)^{1/2}}dy^2
+ \frac{2(py+qx)}{(x^2+y^2)^{1/2}}dydp
- \frac{2(px-qy)}{(x^2+y^2)^{1/2}}dydq \\
& + & 2 (x^2+y^2)^{1/2} dp^2
+ 2 (x^2+y^2)^{1/2} dq^2.
\een
This is a Riemannian metric compatible with $J_2$,
hence it defines a K\"ahler metric in this patch,
hence a K\"ahler metric on the whole $K_2$ since this coordinate
patch is dense in $K_2$.

\begin{prop}
The Riemannian metric $g$ is a hyperk\"ahler metric with respect to $J_2, J_3, J_1$.
\end{prop}

\begin{proof}
For $a=1,2,3$, let
\be
\omega_{J_a}(X, Y) = g(J_a X, Y).
\ee
In the $(x,y,p,q)$ coordinate patch,
\ben
\omega_{J_1} & = & \frac{p^2+q^2+1}{2(x^2+y^2)^{1/2}} dx \wedge dy
+ \frac{py+qx}{(x^2+y^2)^{1/2}}dx \wedge dp
- \frac{px-qy}{(x^2+y^2)^{1/2}} dx \wedge dq \\
& - & \frac{px-qy}{(x^2+y^2)^{1/2}} dy \wedge dp
- \frac{py+qx}{(x^2+y^2)^{1/2}}dy \wedge dq
- 2 (x^2+y^2)^{1/2} dp \wedge dq, \\
\omega_{J_2} & = & \omega = dp \wedge dx + dq \wedge dy, \\
\omega_{J_3} & = & dx \wedge dq - d y \wedge dp.
\een
One can easily see that
\bea
&& d \omega_1 = 0, \\
&& \omega_{J_2} + i \omega_{J_3} = d(p-iq) \wedge d (x+iy) = \Omega.
\eea
This finishes the proof.
\end{proof}

It is easy to see that:
\ben
\omega_{J_1\emph{}} & = & \frac{i}{2} \biggl(\frac{|\alpha_1|^2+1}{2|\gamma_1|} d \gamma_1 \wedge d \bar{\gamma}_1
+ \frac{\alpha_1\bar{\gamma}_1}{|\gamma_1|} d \gamma_1 \wedge d \bar{\alpha}_1
- \frac{\bar{\alpha}_1\gamma_1}{|\gamma_1|} d \bar{\gamma}_1 \wedge d \alpha_1
+ 2|\gamma_1| d \alpha_1 \wedge d \bar{\alpha}_1 \biggr) \\
& = & i \pd \dbar \big(|\gamma_1|(|\alpha_1|^2+1) \big).
\een
Note we have
\be
|\gamma_1|(|\alpha_1|^2+1) = |\gamma_2|(|\alpha_2|^2+1),
\ee
this defines a global function on $K_2$,
which  is the K\"ahler potential for the K\"ahler form $\omega_{J_1}$.

We now pull back everything to $\bR^4-\{0\}$ by the Levi-Civita map.
We take the complex coordinates with respect to $J_2$ to be
\begin{align}
z_1 & = \omega + i \xi, & z_2 & = \chi + i \eta.
\end{align}
Then the K\"ahler form becomes
\be
i(dz_1 \wedge d \bar{z}_1 + dz_2 \wedge d \bar{z}_2),
\ee
and the K\"ahler potential becomes
\be
|z_1|^2+|z_2|^2.
\ee

\subsection{K\"ahler metrics and K\"ahler potentials on Kepler Manifolds $K_n$}

\label{sec:Kepler metrics}

Recall that the Moser map satisfies the identity \eqref{eqn:Moser1},
so after taking exterior differentials on both sides of this identity one has
\be
\sum_{j=0}^n d x_j \wedge d \xi_j = \sum_{j=1}^n d y_j \wedge d \eta_j.
\ee
This defines a symplectic structure on $K_n$.
As first noted by Rawnsley \cite{Rawnsley},
$\omega$ is the K\"ahler form of a Riemannnian metric
with respect to the complex structure $J=J_2$ using the diffeomorphism $K_n \cong C_n^*$,
and
\be \label{eqn:Rawnsley}
\omega = 2 i \pd \dbar |\vec{u}|,
\ee
where $|\vec{u}| = (u_0^2+\cdots + u_n^2)^{1/2}$.
These facts can be checked as follows.
We will first work on $\bC^{n+1}$ and then restrict to $C_n^*$.
On $\bC^{n+1}$ one has
\be
\omega = \sum_{j=0}^n d d x_j \wedge d \xi_j
= \sum_{j=0}^n d \frac{u_j}{(\sum_{k=0}^n u_k^2)^{1/2}} \wedge dv_j,
\ee
and so \eqref{eqn:Rawnsley} can be checked directly,
and using the fact that
\begin{align}
J \pd_{u_j} & = \pd_{v_j}, & J \pd_{v_j} & =- \pd_{u_j}, \;\;\; j=0, \dots, n,
\end{align}
one defines $g$ by \eqref{def:g-from-omega}:
\be
\begin{split}
g& = \frac{1}{(\sum_{k=0}^n u_k^2)^{1/2}} \sum_{j=0}^n du_j^2
- \frac{1}{(\sum_{k=0}^n u_k^2)^{3/2}} \biggl(\sum_{j=0}^n u_j du_j \biggr)^2 \\
& + \frac{1}{(\sum_{k=0}^n u_k^2)^{1/2}} \sum_{j=0}^n dv_j^2
- \frac{1}{(\sum_{k=0}^n u_k^2)^{3/2}} \biggl(\sum_{j=0}^n u_j dv_j\biggr)^2.
\end{split}
\ee
This is not positive definite on $\bC^{n+1}$,
and  we have to to restrict $g$ to $C_n^*$.
Let $u = |\vec{u}| = (u_0^2+ u_1^2 + \cdots + u_n^2)^{1/2}$,
and set
\begin{align}
u_j & = u \hat{u}_j, & v_j & = u \hat{v}_j, \;\; j=0, 1, \dots,n,
\end{align}
then we have:
\bea
&& \sum_{j=0}^n \hat{u}_j^2 = \sum_{j=0}^n \hat{v}_j^2 = 1, \\
&& \sum_{j=0}^n \hat{u}_j \hat{v}_j = 0.
\eea
and also
\ben
\sum_{j=0}^n d u_j^2
& = & \sum_{j=0}^n (d (u\hat{u}_j))^2
= \sum_{j=0}^n (\hat{u}_j du + u d \hat{u}_j)^2 \\
& = & \sum_{j=0}^n \hat{u}_j^2 du^2 + 2 u du \sum_{j=0}^n \hat{u}_j d\hat{u}_j
+ u^2 \sum_{j=0}^n d \hat{u}_j^2 \\
& = & du^2 +  u^2 \sum_{j=0}^n d \hat{u}_j^2,
\een
\ben
\sum_{j=0}^n dv_j^2
& = & \sum_{j=0}^n (d (u\hat{v}_j))^2
=  \sum_{j=0}^n (\hat{v}_j du + u d \hat{v}_j)^2 \\
& = & \sum_{j=0}^n \hat{v}_j^2 du^2
+ 2 u du \sum_{j=0}^n \hat{v}_j d \hat{v}_j
+ u^2 \sum_{j=0}^n d \hat{v}_j^2 \\
& = & du^2 +  u^2 \sum_{j=0}^n d \hat{v}_j^2,
\een

\ben
\sum_{j=0}^n u_j du_j & = & u du, \\
\sum_{j=0}^n u_j dv_j & = & \sum_{j=0}^n u\hat{u}_j d (u\hat{v}_j)
=  \sum_{j=0}^n u\hat{u}_j \hat{v_j} d u
+ \sum_{j=0}^n u^2 \hat{u}_j d \hat{v}_j  \\
& = & u^2 \sum_{j=0}^n \hat{u}_j d \hat{v}_j.
\een
It follows that
\ben
g &=& \frac{1}{u}( du^2 +  u^2 \sum_{j=0}^n d \hat{u}_j^2)-
\frac{1}{u^3} (u du)^2 \\
& + & \frac{1}{u} \sum_{j=0}^n ( du^2 +  u^2 \sum_{j=0}^n d \hat{v}_j^2)
- \frac{1}{u^3} \biggl(  u^2 \sum_{j=0}^n \hat{u}_j d \hat{v}_j \biggr)^2 \\
& = & \frac{1}{u} du^2 +  u \biggl( \sum_{j=0}^n d \hat{u}_j^2
+ \sum_{j=0}^n d \hat{v}_j^2
- \biggl( \sum_{j=0}^n \hat{u}_j d \hat{v}_j \biggr)^2 \biggr).
\een
In the following two subsections we will show that
\be
h=\sum_{j=0}^n d \hat{u}_j^2
+ \sum_{j=0}^n d \hat{v}_j^2
- \biggl( \sum_{j=0}^n \hat{u}_j d \hat{v}_j \biggr)^2
= \sum_{j=0}^n d \hat{u}_j^2
+ \sum_{j=0}^n d \hat{v}_j^2
- \biggl( \sum_{j=0}^n \hat{v}_j d \hat{u}_j \biggr)^2
\ee
is the Sasaki metric on the unit conormal bundle of $S^n$.

Let $u= \frac{r^2}{\sqrt{2}}$,
since we have
\be
g = \frac{1}{\sqrt{2}} \big( dr^2 + r^2h\big)
\ee
is a K\"ahler metric,
$h$ is Sasakian (cf. Sparks \cite{Sparks}).

We note that in the same way one can  define K\"ahler metrics on deformed conifolds.

\subsection{Reformulation of $h$ by the Moser map}

We now understand $h$
by the Moser map.
We take $x_j = \hat{u}_j$ and $\xi_j = \hat{v}_j$,
\be \label{eqn:h}
\begin{split}
h = & \biggl(d \frac{|\vec{y}|^2-1}{|\vec{y}|^2+1} \biggr)^2
+ \sum_{j=1}^n \biggl( d\frac{2y^j}{|\vec{y}|^2+1} \biggr)^2 \\
+ & (d(\vec{\eta} \cdot \vec{y}))^2
+\sum_{j=1}^n
\biggl(d(\frac{|\vec{y}|^2+1}{2} \cdot \eta_j
- (\vec{\eta} \cdot \vec{y}) \cdot y_j)\biggr)^2 \\
- & \biggl(\frac{|\vec{y}|^2-1}{|\vec{y}|^2+1} d(\vec{\eta} \cdot \vec{y})
+\sum_{j=1}^n \frac{2y^j}{|\vec{y}|^2+1}
d (\frac{|\vec{y}|^2+1}{2} \cdot \eta_j - (\vec{\eta} \cdot \vec{y}) \cdot y_j) \biggr)^2.
\end{split}
\ee
The first line on the right-hand side  of \eqref{eqn:h} can be simplified to be
\ben
I & = &  \frac{4 \sum_{j=1}^n dy_j^2}{(\sum_{j=1}^n y_j^2+1)^2},
\een
The third line on the right-hand side of \eqref{eqn:h} can be simplified as follows:
\ben
III & = & - \biggl(\frac{|\vec{y}|^2-1}{|\vec{y}|^2+1} d(\vec{\eta} \cdot \vec{y}) \\
& + &\sum_{j=1}^n \frac{2y^j}{|\vec{y}|^2+1}
\biggl(\frac{|\vec{y}|^2+1}{2} \cdot d \eta_j
+ \half \eta_j d|\vec{y}|^2
-y_j d(\vec{\eta} \cdot \vec{y})
- (\vec{\eta} \cdot \vec{y}) \cdot d y_j)\biggr)\biggr)^2\\
& = & - \biggl(\frac{|\vec{y}|^2-1}{|\vec{y}|^2+1} d(\vec{\eta} \cdot \vec{y}) \\
& + & \sum_{j=1}^n y_j d\eta_j
+ \frac{\vec{y} \cdot \vec{\eta}}{|\vec{y}|^2+1} d|\vec{y}|^2
-\frac{2|\vec{y}|^2}{|\vec{y}|^2+1} d(\vec{\eta} \cdot \vec{y})
- \frac{(\vec{\eta} \cdot \vec{y})}{|\vec{y}|^2+1} \cdot d |\vec{y}|^2)\biggr)\biggr)^2 \\
& = & - \biggl(
- d(\vec{\eta} \cdot \vec{y}) +\sum_{j=1}^n y_j d\eta_j  \biggr)^2
= - \biggl( \sum_{j=1}^n \eta_j dy_j  \biggr)^2.
\een
The second line on the right-hand side of \eqref{eqn:h} is more complicated
and will be treated in steps as follows:
First we expand the square in the second term:
\ben
II & = & (d(\vec{\eta} \cdot \vec{y}))^2
+\sum_{j=1}^n
\biggl(\frac{|\vec{y}|^2+1}{2} \cdot d \eta_j
+ \half \eta_j d|\vec{y}|^2
-y_j d(\vec{\eta} \cdot \vec{y})
- (\vec{\eta} \cdot \vec{y}) \cdot d y_j)\biggr)^2 \\
& = & (d(\vec{\eta} \cdot \vec{y}))^2
+ \frac{(1+|\vec{y}|^2)^2}{4}\sum_{j=1}^n d \eta_j^2
+ \frac{1}{4} |\eta|^2 (d|\vec{y}|^2)^2
+|\vec{y}|^2 ( d(\vec{\eta} \cdot \vec{y}))^2 \\
& + &  (\vec{\eta} \cdot \vec{y})^2 \cdot \sum_{j=1}^n d y_j^2
+ \frac{|\vec{y}|^2+1}{4} \cdot d |\eta|^2  d|\vec{y}|^2
- (|\vec{y}|^2+1) \cdot \sum_{j=1}^n y_j d \eta_j
\cdot d(\vec{\eta} \cdot \vec{y}) \\
& - & (|\vec{y}|^2+1) \cdot
 (\vec{\eta} \cdot \vec{y}) \cdot \sum_{j=1}^n d \eta_j d y_j
- (\vec{\eta} \cdot \vec{y}) d|\vec{y}|^2 d(\vec{\eta} \cdot \vec{y}) \\
& - & (\vec{\eta} \cdot \vec{y}) \cdot
\sum_{j=1}^n \eta_j d y_j \cdot d|\vec{y}|^2
+ (\vec{\eta} \cdot \vec{y}) \cdot
d(\vec{\eta} \cdot \vec{y}) d |\vec{y}|^2.
\een
Next we simplify the above expression as follows:
\ben
II & = & \frac{(1+|\vec{y}|^2)^2}{4}\sum_{j=1}^n d \eta_j^2
+ \frac{1}{4} |\eta|^2 (d|\vec{y}|^2)^2
+(1+|\vec{y}|^2) ( d(\vec{\eta} \cdot \vec{y}))^2 \\
& + &  (\vec{\eta} \cdot \vec{y})^2 \cdot \sum_{j=1}^n d y_j^2
+ \frac{|\vec{y}|^2+1}{4} \cdot d |\eta|^2  d|\vec{y}|^2
- (|\vec{y}|^2+1) \cdot \sum_{j=1}^n y_j d \eta_j
\cdot d(\vec{\eta} \cdot \vec{y}) \\
& - & (|\vec{y}|^2+1) \cdot
 (\vec{\eta} \cdot \vec{y}) \cdot \sum_{j=1}^n d \eta_j d y_j
- (\vec{\eta} \cdot \vec{y}) \cdot
\sum_{j=1}^n \eta_j d y_j \cdot d|\vec{y}|^2 \\
& = & A+ B + C,
\een
where we have split the terms according to their types:
\ben
A & = &  \frac{(1+|\vec{y}|^2)^2}{4}\sum_{j=1}^n d \eta_j^2,
\een
\ben
B & = &  \frac{|\vec{y}|^2+1}{4} \cdot d |\eta|^2  d|\vec{y}|^2
+ (|\vec{y}|^2+1) \cdot \sum_{j=1}^n \eta_j d y_j
\cdot \sum_{k=1}^n y_k d\eta_k  \\
& - & (|\vec{y}|^2+1) \cdot
 (\vec{\eta} \cdot \vec{y}) \cdot \sum_{j=1}^n d \eta_j d y_j \\
&= & (|\vec{y}|^2+1) \sum_{j=1}^n d\eta_j \cdot
\biggl(\sum_{k=1}^n (\eta_j y_k + \eta_k y_j) d y_k
- (\vec{\eta}\cdot \vec{\eta} ) d y_j \biggr).
\een
\ben
C & = & \frac{1}{4} |\eta|^2 (d|\vec{y}|^2)^2
+  (\vec{\eta} \cdot \vec{y})^2 \cdot \sum_{j=1}^n d y_j^2 \\
& + & (|\vec{y}|^2+1) \cdot \biggl( \sum_{j=1}^n \eta_j d y_j \biggr)^2
- (\vec{\eta} \cdot \vec{y}) \cdot
\sum_{j=1}^n \eta_j d y_j \cdot d|\vec{y}|^2  \\
& = & \sum_{j=1}^n \biggl( - (\vec{\eta} \cdot \vec{y}) dy_j
+ \sum_{k=1}^n (y_k \eta_j + y_j \eta_k) dy_k\biggr)^2
+  \biggl(\sum_{j=1}^n \eta_j d y_j \biggr)^2.
\een
Combining the above terms together,
we get:
\ben
h & = &  \frac{4 \sum_{j=1}^n dy_j^2}{(\sum_{j=1}^n y_j^2+1)^2}
+ \frac{(1+|\vec{y}|^2)^2}{4}\sum_{j=1}^n d \eta_j^2 \\
& + & (|\vec{y}|^2+1) \sum_{j=1}^n d\eta_j \cdot
\biggl(\sum_{k=1}^n (\eta_j y_k + \eta_k y_j) d y_k
- (\vec{\eta}\cdot \vec{\eta} ) d y_j \biggr) \\
& + & \sum_{j=1}^n \biggl( - (\vec{\eta} \cdot \vec{y}) dy_j
+ \sum_{k=1}^n (y_k \eta_j + y_j \eta_k) dy_k\biggr)^2 \\
& = & \frac{4 \sum_{j=1}^n dy_j^2}{(\sum_{j=1}^n y_j^2+1)^2} \\
& + & \frac{(1+|\vec{y}|^2)^2}{4}\sum_{j=1}^n
\biggl( d \eta_j
+ \frac{2}{(|\vec{y}|^2+1)}
\biggl(\sum_{k=1}^n (\eta_j y_k + \eta_k y_j) d y_k
- (\vec{\eta}\cdot \vec{\eta} ) d y_j \biggr) \biggr)^2.
\een
To summarize we have proved the following:

\begin{prop} \label{prop:h-D}
On $T^*S^n$,
the expression $h$ defined in \eqref{eqn:h} can be expressed in local coordinates
$\{y_i, \eta_i\}$ as follows:
\be
h = \frac{4 \sum_{j=1}^n dy_j^2}{(\sum_{j=1}^n y_j^2+1)^2}
+ \frac{(1+|\vec{y}|^2)^2}{4}\sum_{j=1}^n  D\eta_j,
\ee
where
\be
D\eta_j = d \eta_j
+ \frac{2}{(|\vec{y}|^2+1)}
\biggl(\sum_{k=1}^n (\eta_j y_k + \eta_k y_j) d y_k
- (\vec{\eta}\cdot \vec{\eta} ) d y_j \biggr).
\ee
\end{prop}

\subsection{Sasaki metric}

To under $h$ in last Proposition let us recall the {\em Sasaki metric} on the tangent bundle of a Riemannian manifold
$(M^n, g)$.
Suppose   in a local coordinate patch $U$ with local
coordinates $\{x^i\}$ the Riemannian metric of $M$ is given by the quadratic form
\be
ds^2 = g_{jk}(x)dx^jdx^k.
\ee
Denote by $\nabla$ the Levi-Civita connection
of $(M, g)$,
\be
\nabla \frac{\pd}{\pd x^i} = \Gamma_{ij}^k dx^j \otimes \frac{d}{dx^k}.
\ee
Then
\be
\nabla (v^i \frac{\pd}{\pd x^i} ) = d v^i \otimes \frac{\pd}{\pd x^i}
+ v^i \nabla \frac{\pd}{\pd x^i}
= Dv^i \otimes \frac{\pd}{\pd x^i},
\ee
where $Dv^i$ means the covariant differential of $v^j$, i.e.
\be
Dv^i = dv^i + \Gamma^i_{jk} v^j dx^k.
\ee
Similarly,
we have
\be
\nabla (p_idx^i ) = d p_i \otimes dx^i + p_i \nabla dx^i
= Dp_j \otimes d x^j,
\ee
where $Dp_i$ means the covariant differential of $p_j$, i.e.
\be
Dp_j = dp_j - \Gamma^i_{jk} p_i dx^k.
\ee

The Sasaki metric of $TM$ is given in $\pi^{-1}(U)$ by the quadratic form
\be
d\sigma^2 = g_{jk}(x)dx^jdx^k + g_{jk}(x)Dv^jDv^k.
\ee
It can also be written as:
\be
\begin{split}
d\sigma^2 & = g_{ij}(x)dx^idx^j + g_{ij}(x) \Gamma^i_{ab}\Gamma^j_{cd}v^a v^cdx^b dx^d \\
& + 2 g_{ij}(x) \Gamma^i_{ab} v^a dx^b dv^j
+ g_{ij}(x) dv^i dv^j.
\end{split}
\ee
Using the Riemannian metric one can define an isomorphism
$\sharp: T^*M \to TM$ by:
\be
p_j dx^j \mapsto v^i \frac{\pd}{\pd x^i} = p_j g^{ij} \frac{\pd}{\pd x^i}.
\ee
Using this isomorphism one can pull back the Sasaki metric
to $T^*M$.
Since we have
\be
v^i = g^{ij} p_j,
\ee
it follows that
\ben
d\sigma^2 & = &
g_{ij}dx^idx^j + g_{ij} \Gamma^i_{ab}\Gamma^j_{cd}g^{ak}p_k g^{cl}p_ldx^b dx^d \\
& + & 2 g_{ij} \Gamma^i_{ab} g^{ak}p_k dx^b d(g^{jl}p_l)
+ g_{ij} d(g^{ik}p_k) d(g^{jl}p_l) \\
& = &
g_{ij}dx^idx^j + g_{ij} \Gamma^i_{ab}\Gamma^j_{cd}g^{ak}p_k g^{cl}p_ldx^b dx^d \\
& + & 2 g_{ij} \Gamma^i_{ab} g^{ak}p_k dx^b g^{jl} d p_l
+ 2 g_{ij} \Gamma^i_{ab} g^{ak}p_k dx^b p_l dg^{jl} \\
& + & g_{ij} p_k p_l dg^{ik} dg^{jl}
+ 2 g_{ij} g^{jl} dg^{ik}p_k dp_l
+ g_{ij} g^{ik} g^{jl} d p_k  d p_l\\
& = & g_{ij}dx^idx^j
+ g_{ij} (\Gamma^i_{ab} g^{ak} dx^b + d g^{ik})
(\Gamma^j_{cd} g^{cl} dx^d + dg^{jl}) p_k p_l \\
& + &  2 (\Gamma^i_{ab} g^{ak}dx^b + dg^{ik} ) p_k dp_i
+ g^{kl} d p_k d p_l.
\een
Now note we have the following computations:
\ben
&& \Gamma^i_{ab} g^{ak} dx^b  + dg^{ik} \\
& = & \half g^{ij} (g_{aj,b} + g_{bj,a} - g_{ab,j}) g^{ak} dx^b
- g^{ij} \cdot g_{jl, b} dx^b\cdot g^{lk} \\
& = & \half g^{ij} g_{aj,b} g^{ak} dx^b
+ \half g^{ij} g_{bj,a} g^{ak} dx^b
- \half g^{ij} g_{ab,j} g^{ak} dx^b
- g^{ij} \cdot g_{jl, b} dx^b\cdot g^{lk} \\
& = & - \half g^{ij} g_{aj,b} g^{ak} dx^b
+ \half g^{ij} g_{bj,a} g^{ak} dx^b
- \half g^{ij} g_{ab,j} g^{ak} dx^b  \\
& = & - g^{ij} \Gamma^k_{jb} dx^b.
\een
It  follows that
\be
d\sigma^2 = g_{jk}dx^jdx^k + g^{kl} D p_k D p_l.
\ee
Indeed, we have:
\ben
d\sigma^2 & = & g_{ij}dx^idx^j
+ g_{ij} g^{is} \Gamma^k_{sb} dx^b
\cdot g^{jt} \Gamma^k_{tc} dx^c  p_k p_l \\
& + &  2 (\Gamma^i_{ab} g^{ak}dx^b + dg^{ik} ) p_k dp_i
+ g^{kl} d p_k d p_l \\
& = & g_{ij}dx^idx^j + g^{kl} \Gamma^i_{ka} p_i \Gamma^j_{lb} p_j dx^a dx^b
 2 g^{ij} \Gamma^k_{jb} dx^b p_k dp_i
+ g^{kl} d p_k d p_l \\
& = & g_{ij} dx^i dx^j + g^{ij} Dp_iD p_j.
\een

Now we let
\be
g_{ij} = \frac{4}{\big( 1+ |\vec{x}|^2)^2} \delta_{i, j},
\;\;\;\;\;\; i,j=1, \dots, n,
\ee
where $|\vec{x}| = \sum_{j=1}^n (x^j)^2$.
Then one has
\be
D p_j = d p_j  - \frac{ 2 \vec{x} \cdot \vec{p}}{1+|\vec{x}|^2} dx^j
+ \sum_{k } \frac{2(x^kp_j + x^j p_k)}{1+|\vec{x}|^2} dx^k.
\ee
where $\vec{x} \cdot \vec{p} =\sum_i x^ip_i$.

By comparing with Proposition \label{prop:h}
we have

\begin{prop} \label{prop:h2}
On $T^*S^n$,
the expression $h$ defined in \eqref{eqn:h}
is the Sasaki metric for the round metric on $S^n$.
\end{prop}

\subsection{Ricci curvature of the Kepler metric $g$ on $K_n$}

In the above we have shown that the $g$ on $K_n$ is K\"ahler,
now we show it is also Ricci-flat.
Recall
\be
\omega = 2 i \pd \dbar u,
\ee
where
$$u = (u_0^2 + \cdots + u_n^2)^{1/2}
= \frac{1}{\sqrt{2}} (|w_0|^2+|w_1|^2+ \cdots + |w_n|^2)^{1/2}
= \frac{1}{\sqrt{2}} |\vec{w}|.$$
Since we have
\begin{align}
\pd |\vec{w}| & = \frac{1}{2|\vec{w}|} \sum_{j=0}^n \bar{w}_j dw_j, &
\dbar |\vec{w}| & = \frac{1}{2|\vec{w}|} \sum_{k=0}^n w_k d\bar{w}_k,
\end{align}
after taking differential again,
one gets
\be
\omega = \frac{\sqrt{2}i}{2|\vec{w}|} \sum_{j=0}^n dw_j\wedge d\bar{w}_j
- \frac{\sqrt{2}i}{4|\vec{w}|^3}
\sum_{j=0}^n \bar{w}_j dw_j \wedge \sum_{k=0}^n w_k d\bar{w}_k.
\ee
On the local coordinate patch on $C_n^*$ with $w_1,\dots, w_n$ as local coordinates,
we have
\be
w_0 = \pm i (w_1^2+ \cdots + w_n^2)^{1/2},
\ee
and so we have
\begin{align}
\pd w_0 & = - \frac{\sum_{j=1}^n w_j dw_j}{w_0}, &
\dbar\bar{w}_0 & = - \frac{\sum_{k=1}^n \bar{w}_k d\bar{w}_k}{\bar{w}_0}.
\end{align}
Now the K\"ahler form is given by:
\ben
\omega
& = & \frac{\sqrt{2}i}{2|\vec{w}|} \sum_{j=1}^n dw_j\wedge d\bar{w}_j
+ \frac{\sqrt{2}i}{2|\vec{w}|\cdot|w_0|^2} \sum_{j=1}^n w_j dw_j \wedge
\sum_{k=1}^n \bar{w}_k d\bar{w}_k \\
& - & \frac{\sqrt{2}i}{4|\vec{w}|^3}
\biggl( \sum_{j=1}^n \bar{w}_j dw_j - \frac{\bar{w}_0\sum_{j=1}^n w_j dw_j}{z_0} \biggr)
\wedge \biggl( \sum_{k=1}^n z_k d\bar{w}_k
- \frac{w_0\sum_{k=1}^n \bar{w}_k d\bar{w}_k}{\bar{w}_0} \biggr) \\
& = & \frac{\sqrt{2}i}{2|\vec{w}|} \sum_{j=1}^n dw_j\wedge d\bar{w}_j
+ \frac{\sqrt{2}i}{2|\vec{w}|\cdot|w_0|^2}
\sum_{j=1}^n w_j dw_j \wedge \sum_{k=1}^n \bar{w}_k d\bar{w}_k \\
& - & \frac{\sqrt{2}i}{4|\vec{w}|^3}
\sum_{j=1}^n (w_0\bar{w}_j-\bar{w}_0w_j) dw_j
\wedge \sum_{k=1}^n (\bar{w}_0w_k-w_0\bar{w}_k) d\bar{w}_k.
\een
Write $\omega = \omega_{i\bar{j}} dw_i \wedge dw_j$,
one can get:
\ben
\det (\omega_{i\bar{j}}) & = & \biggl(\frac{\sqrt{2}i}{2} \biggr)^n
\frac{1}{2 |\sum_{j=1}^n w_j^2| \cdot (|\sum_{j=1}^n w_j^2|+\sum_{j=1}^n |w_j|^2)^{(n-2)/2}},
\een
using the following easily proved identity:
\be \label{eqn:Determinant}
\begin{split}
& \det ( \delta_{i,j}+w_i\bar{w}_j +¡¡x z_i\bar{z}_j )_{i,j=1, \dots£¬n} \\
= & (1+ \sum_{j=1}^n |w_j|^2) (1 + x \sum_{k=1}^n |z_k|^2)
- x \biggl| \sum_{j=1}^n z_j \bar{w}_j \biggr|^2.
\end{split}
\ee
The Ricci form of $\omega$ is
\be
\rho = - i\pd \dbar \log \det (\omega_{i\bar{j}})
= - \frac{n-2}{2} \pd \dbar \log \biggl(|\sum_{j=1}^n w_j^2|+\sum_{j=1}^n |w_j|^2\biggr).
\ee
In particular,
when $n=2$, $\rho = 0$.

\section{K\"ahler Ricci-Flat Metric on Conifolds and Deformed Conifolds}

\label{sec:KRF-Con-Deformed}

In last Section we have seen that
the Kepler metrics on the Kepler manifolds $K_n$ are in general not
K\"ahler Ricci-flat.
In this Section we show that it is possible to
construct natural K\"ahler Ricci-flat metrics on conifold
and the deformed conifold.

\subsection{Some K\"ahler Ricci-flat metrics on the conifold}
\label{sec:KRF-Con}

Let us consider a K\"ahler potential of the form:
\be
K = f(t)
\ee
on $C_n^*$, where $t$ is defined by:
\be
t = \sum_{j=0}^n |w_j|^2.
\ee
Since we have
\begin{align}
\pd t & = \sum_{j=1}^n \bar{w}_j d w_j - \frac{\bar{w}_0\sum_{j=1}^n w_j dw_j}{w_0}
= \frac{1}{w_0} \sum_{j=1}^n (w_0 \bar{w}_j - \bar{w}_0 w_j) d w_j, \\
\dbar t & = \sum_{j=1}^n w_j d \bar{w}_j  - \frac{w_0 \sum_{k=1}^n \bar{w}_k d\bar{w}_k}{\bar{w}_0}
= \frac{1}{\bar{w}_0} \sum_{j=1}^n (\bar{w}_0 w_j - w_0 \bar{w}_j) d \bar{w}_j.
\end{align}
and
\be
\pd \dbar t = \sum_{j=1}^n dw_j \wedge d \bar{w}_j
+  \frac{1}{|w_0|^2}  \sum_{j=1}^n w_j dw_j\wedge \sum_{k=1}^n \bar{w}_k d\bar{w}_k,
\ee
the K\"ahler form associated with $K$ is then
\ben
\omega_K : &=& i \pd \dbar K = i f'(t) \pd \dbar t + i f''(t) \pd t \wedge \dbar t \\
& = & i f'(t) \cdot \biggl(  \sum_{j=1}^n dw_j \wedge d \bar{w}_j
+  \frac{1}{|w_0|^2}  \sum_{j=1}^n w_j dw_j\wedge \sum_{k=1}^n \bar{w}_k d\bar{w}_k
\biggr) \\
& + & i f''(t) \cdot  \frac{1}{|w_0|^2} \sum_{j=1}^n (w_0 \bar{w}_j - \bar{w}_0 w_j) d w_j
\wedge \sum_{j=1}^n (\bar{w}_0 w_j - w_0 \bar{w}_j) d \bar{w}_j.
\een
By the formula \eqref{eqn:Determinant},
the determinant of the associated Hermitian matrix is:
\ben
&& i^n (f'(t))^n \cdot \biggl[ \biggl( 1 + \frac{1}{|w_0|^2} \sum_{j=1}^n |w_j|^2
\biggr) \cdot
\biggl(1 + \frac{f''(t)}{f'(t)}
\frac{1}{|w_0|^2} \sum_{j=1}^n |\bar{w}_0 w_j - w_0 \bar{w}_j|^2 \biggr) \\
& - & \frac{f''(t)}{f'(t)} \frac{1}{|w_0|^4}
\biggl|\sum_{j=1}^n w_j(\bar{w}_0 w_j - w_0 \bar{w}_j) \biggr|^2 \biggr] \\
& = & i^n (f'(t))^n \cdot  \frac{1}{|w_0|^2} \biggl[ \sum_{j=0}^n |w_j|^2 \cdot
\biggl(1 + \frac{f''(t)}{f'(t)} \cdot  2 \sum_{j=0}^n |w_j|^2 \biggr)
- \frac{f''(t)}{f'(t)} \biggl(\sum_{j=0}^n |w_j|^2\biggr)^2 \biggr] \\
& = & i^n   \cdot  \frac{1}{|w_0|^2} \biggl[ t(f'(t))^n + t^2 (f'(t))^{n-1}f''(t) \biggr].
\een
In the above we have used the following computations:
\bea
&& \sum_{j=1}^n |\bar{w}_0 w_j - w_0 \bar{w}_j|^2
= \sum_{j=1}^n (2|w_0|^2|w_j|^2- w_0^2\bar{w}_j^2 - \bar{w}_0^2 w_j^2)
= 2 |w_0|^2 \sum_{j=0}^n |w_j|^2, \label{eqn:Aux1} \\
&& \sum_{j=1}^n w_j(\bar{w}_0 w_j - w_0 \bar{w}_j)
= \bar{w}_0 \sum_{j=1}^n w_j^2 - w_0 \sum_{j=1}^n |w_j|^2
= - w_0 \sum_{j=0}^n |w_j|^2, \label{eqn:Aux2}
\eea
and in these equalities we use the equation
\be
\sum_{j=0}^n w_0^2 = 0.
\ee
The Ricci form of $\omega_K$ is
\be
\rho_K = -i \pd \dbar \log \biggl[ t(f'(t))^n + t^2 (f'(t))^{n-1}f''(t) \biggr].
\ee

To get K\"ahler Ricci-flat metric, it suffices to have
\be
t(f'(t))^n + t^2 (f'(t))^{n-1}f''(t) = c
\ee
Let $y(t) = t \phi'(t)$,
\be
y^{n-1} y' = ct^{n-2}
\ee
Integrating once,
\be
y^n = \frac{nc}{n-1} t^{n-1} + c_1.
\ee
Rewrite it as follows:
\be
f'(t) = \frac{1}{t}\biggl(\frac{nc}{n-1} t^{n-1} + c_1\biggr)^{1/n},
\ee
and so
\be
f(t) = \int
\frac{1}{t}
\biggl(\frac{nc}{n-1} t^{n-1} + c_1\biggr)^{1/n} dt.
\ee
In particular,
$f(t) = C t^{(n-1)/n}$ is a solution.

\subsection{Some K\"ahler Ricci-flat metrics on the deformed conifold}

It is remarkable that similar computations can be done on the deformed conifold
$\tilde{C}_n(a)$ defined by the equation
\be
\sum_{j=0}^n w_j^2 = a,
\ee
where $a \in \bC^* = \bC - \{0\}$.
We need to modify \eqref{eqn:Aux1} and \eqref{eqn:Aux2}:
\ben
&& \sum_{j=1}^n |\bar{w}_0 w_j - w_0 \bar{w}_j|^2
= \sum_{j=1}^n (2|w_0|^2|w_j|^2- w_0^2\bar{w}_j^2 - \bar{w}_0^2 w_j^2) \\
& = & 2 |w_0|^2 \sum_{j=0}^n |w_j|^2 - (a \bar{w}_0^2+\bar{a}w_0^2),\\
&& \sum_{j=1}^n w_j(\bar{w}_0 w_j - w_0 \bar{w}_j)
= \bar{w}_0 \sum_{j=1}^n w_j^2 - w_0 \sum_{j=1}^n |w_j|^2 \\
& = & a \bar{w}_0 - w_0 \sum_{j=0}^n |w_j|^2.
\een

\ben
&& i^n (f'(t))^n \cdot \biggl[ \biggl( 1 + \frac{1}{|w_0|^2} \sum_{j=1}^n |w_j|^2
\biggr) \cdot
\biggl(1 + \frac{f''(t)}{f'(t)}
\frac{1}{|w_0|^2} \sum_{j=1}^n |\bar{w}_0 w_j - w_0 \bar{w}_j|^2 \biggr) \\
& - & \frac{f''(t)}{f'(t)} \frac{1}{|w_0|^4}
\biggl|\sum_{j=1}^n w_j(\bar{w}_0 w_j - w_0 \bar{w}_j) \biggr|^2 \biggr] \\
& = & i^n (f'(t))^n \cdot  \frac{1}{|w_0|^2} \biggl[ \sum_{j=0}^n |w_j|^2 \cdot
\biggl(1 + \frac{f''(t)}{f'(t)} \cdot
\frac{1}{|w_0|^2} \biggl(2 |w_0|^2\sum_{j=0}^n |w_j|^2 - (a \bar{w}_0^2+\bar{a}w_0^2)\biggr) \biggr) \\
& - & \frac{f''(t)}{f'(t)} \frac{1}{|w_0|^4}
 \biggl|a \bar{w}_0 - w_0 \sum_{j=0}^n |w_j|^2\biggr|^2 \biggr] \\
& = & i^n (f'(t))^n \cdot  \frac{1}{|w_0|^2} \biggl[ t \cdot
\biggl(1 + \frac{f''(t)}{f'(t)} \cdot
\biggl(2 t - \frac{(a \bar{w}_0^2+\bar{a}w_0^2)}{|w_0|^2} \biggr) \biggr) \\
& - & \frac{f''(t)}{f'(t)} \frac{1}{|w_0|^4}
 \biggl(|a|^2 |w_0|^2 -t (a\bar{w}_0^2 + \bar{a}w_0^2)
 + t^2 |w_0|^2 \biggr) \biggr] \\
& = & i^n   \cdot  \frac{1}{|w_0|^2}
\biggl[ t(f'(t))^n + (t^2-|a|^2) (f'(t))^{n-1}f''(t) \biggr].
\een
Therefore,
to get a K\"ahler Ricci-flat metric, it suffices to have
\be
t(f'(t))^n + (t^2-|a|^2) (f'(t))^{n-1}f''(t) = c.
\ee
When $a=1$,
this was obtained by Stenzel \cite{Stenzel}.
Let us modify his solution for this case to general $a \in \bC^*$ as follows.
Let $t= |a| x$,
and $g(x) : = f(t)$.
Then since $g'(x) = |a| f'(t)$, $g''(x) = |a|^2 f''(t)$,
\be
x (g'(x))^n + (x^2-1) (g'(x))^{n-1} g''(x) = c|a|^{n-1}.
\ee
Now let $x = \cosh w$, and
$h(w):= g(x)$.
Then one has
\be
\frac{d}{dw} (h'(w))^n = n c|a|^{n-1} (\sinh w)^{n-1}.
\ee

\subsection{Sasakian-Einstein metric}

One can redo the calculations in Section \ref{sec:KRF-Con}
in the fashion of Section \ref{sec:Kepler metrics}.
I.e.,
one first works on $\bC^{n+1}$, then restricts to $C_n^*$.
Then
\ben
\omega_K
= \sqrt{-1} f'(t) \sum_{j=0}^n d w_j \wedge d \bar{w}_j
+ \sqrt{-1} f''(t) \sum_{j=0}^n \bar{w}_j dw_j \wedge \sum_{k=0}^n w_k d \bar{w}_k.
\een
The Riemannian metric is
\be
\begin{split}
g_K = & 2f'(t) \sum_{j=0}^n (du_j^2+dv_j^2)
+ 2 f''(t) \biggl(\sum_{j=0}^n(u_jdu_j+v_jdv_j)\biggr)^2 \\
+ & 2f''(t) \biggl(\sum_{j=0}^n(u_jdv_j-v_jdu_j)\biggr)^2.
\end{split}
\ee
By the same computations as in Section \ref{sec:Kepler metrics},
\be
\begin{split}
g_K = & 2 f'(t) \cdot \biggl(2du^2 + u^2 \sum_{j=0}^n (d\hat{u}_j^2 + d\hat{v}_j^2)\biggr) \\
+ & 8 f''(t) u^2 du^2 + 8 f''(t) u^4 \biggl( \sum_{j=0}^n \hat{u}_j d\hat{v}_j\biggr)^2.
\end{split}
\ee
Note $u^2= \half t$, so
\be
\begin{split}
g_K = & \biggl(\frac{f'(t)}{2t} + \frac{f''(t)}{2}\biggr) \cdot dt^2
+ t f'(t) \sum_{j=0}^n (d\hat{u}_j^2 + d\hat{v}_j^2)\\
+ & 2t^2 f''(t)  \biggl( \sum_{j=0}^n \hat{u}_j d\hat{v}_j\biggr)^2.
\end{split}
\ee
When $f(t) = C t^{(n-1)/n}$,
\be
\begin{split}
g_K = & C\frac{(n-1)^2}{2n^2} t^{-(n+1)/n} dt^2\\
+ & C\frac{n-1}{n} t^{(n-1)/n} \biggl( \sum_{j=0}^n (d\hat{u}_j^2 + d\hat{v}_j^2)
- \frac{2}{n} \biggl( \sum_{j=0}^n \hat{u}_j d\hat{v}_j\biggr)^2 \biggr).
\end{split}
\ee
Take $C= \frac{2n^2}{(n-1)^2}$ and $r = \sqrt{\frac{2n}{(n-1)}} t^{\frac{n-1}{2n}}$,
\be
g_K = dr^2 + r^2  \biggl( \sum_{j=0}^n (d\hat{u}_j^2 + d\hat{v}_j^2)
- \frac{2}{n} \biggl( \sum_{j=0}^n \hat{u}_j d\hat{v}_j\biggr)^2 \biggr).
\ee
As a corollary,
this shows that
$$\sum_{j=0}^n (d\hat{u}_j^2 + d\hat{v}_j^2)
- \frac{2}{n} \biggl( \sum_{j=0}^n \hat{u}_j d\hat{v}_j\biggr)^2$$
defines a Sasakian-Einstein metric \cite{Boy-Gal}
on the unit conormal bundle of the round sphere $S^n$.

\subsection{A special K\"ahler-Ricci flow}

\label{sec:KRFlow}

Consider K\"ahler-Ricci flow of the form
\be
\frac{d}{ds} \omega_K = \rho_K + \lambda \cdot \omega_K,
\ee
where $\lambda$ is some constant,
and
\ben
&& \omega_K = \sqrt{-1} \pd \dbar f_s(t), \\
&& \rho_K = -i \pd \dbar \log \biggl[ t(f_s'(t))^n + t^2 (f_s'(t))^{n-1}f_s''(t) \biggr],
\een
for some family $\{f_s(t)\}$ of functions in $t$,
parameterized by $s$.
One can consider the following flow in the space of K\"ahler potentials:
\be
\frac{d}{ds} f_s(t) = - \log \biggl[ t(f_s'(t))^n + t^2 (f_s'(t))^{n-1}f_s''(t) \biggr]
+\lambda f_s(t)+C_1.
\ee


\section{$U(n)$-Symmetric K\"ahler Metrics and Hessian Geometry}

\label{sec:Hessian}

Let us summarize what we have discussed so far in this paper:
(1) Regularizations of the Kepler systems lead to the symplectic manifolds $K_n$.
(2) $K_n$ have natural complex structures which can used to relate them to the $n$-conifolds
$C_n$.
(3)
The complex structures and symplectic structures on $K_n$ are compatible,
leading to the Kepler metrics.
(4) The Kepler metrics have K\"ahler potential $\frac{1}{\sqrt{2}} |\vec{w}|$,
and this is also the K\"ahler potential for some K\"ahler metrics the deformed conifolds.
(5) There are K\"ahler Ricci-flat metrics on the $n$-folds and the deformed conifolds
with K\"ahler potentials of the form $f(|\vec{w}|^2)$,
in particular $C|\vec{w}|^{2(n-1)/n}$ defines a K\"ahler Ricci-flat metric on $K_n$.
Furthermore, there are several ways to understand $K_3$:
(a) holomorphically equivalent to $\cO_{\bP^1}(-1) \oplus \cO_{\bP^1}(-1) - \bP^1$;
(b) holomorphically equivalent to $\cO_{\bP^1 \times \bP^1}(-1,-1) - \bP^1 \times \bP^1$;
(c) holmorphically equivalent to $C_3^*$.
Similarly, there are several ways to understand $K_2$:
(a) there is a $2:1$ covering by $\bC^2$;
(b) holomorphically equivalent to $\cO_{\bP^1}(-2) - \bP^1$;
(c) holmorphically equivalent to $C_2^*$.
In the preceding Sections we have
discussed the K\"ahler metrics on $K_2$ from
the third point of view,
and considered K\"ahler metrics on the two-dimensional deformed conifold.
In this Section we will take the first two points of view.
This will leads to a discussion of the K\"ahler metrics
on the two-dimensional resolved conifold.
We will take a more general point of view and follow the approach of construction of K\"ahler metrics
with $U(n)$-symmetry for $n \geq 2$
developed in Duan-Zhou \cite{Duan-Zhou1, Duan-Zhou2},
generalizing a construction by LeBrun \cite{LeBrun}.
This will lead us to the application of symplectic coordinates
and Hessian geometry
in the noncompact case.

\subsection{A simple example}

Consider the flat K\"ahler metric on $\bC$:
\be
\omega = \frac{\sqrt{-1}}{2} \pd\dbar u
= \frac{\sqrt{-1}}{2} dz \wedge d \bar{z},
\ee
where $u = |z|^2$.
Let $z = re^{i\theta}$,
then the Riemannian metric can be written in the following form:
\be
g = (dr)^2 + r^2 (d\theta)^2 = \frac{(dy)^2}{4y} + y \cdot (d\theta)^2,
\ee
where a new local coordinate $y$ is introduced instead of $r$:
\be
y = r^2
\ee
One also introduces:
\be
\psi = y (\ln y - 1),
\ee
and a dual coordinate $y^\vee$ by
\be
y^\vee = \frac{\pd \psi}{\pd y} = \ln y.
\ee
It is clear that:
\be
y = e^{y^\vee}.
\ee
The Legendre transform of $\psi$ is given by:
\be
\psi^\vee = yy^\vee - \psi = y \frac{\pd \psi}{\pd y} - \psi.
\ee
Then
\be
\psi^\vee = y = e^{y^\vee}.
\ee
Now the K\"ahler form and the Riemannian metric can be rewritten as:
\be
\omega = \half dy \wedge d\theta,
\ee
\be
g = \frac{1}{4} \frac{\pd^2\psi}{\pd y\pd y} (dy)^2
+  \frac{\pd^2 \psi^\vee}{\pd y^\vee\pd y^\vee} (d\theta)^2.
\ee
In the rest of this paper,
we will generalize this example in various ways
and use them to study the Kepler metrics.

\subsection{Symplectic coordinates}
On $\bC^n - \{0\}$ with linear coordinates $z_1, \dots, z_n$
we consider $U(n)$-symmetric K\"ahler metrics.
Take the K\"ahler potential to be a function of the form $\phi(u)$,
where $u=|z_1|^2 + \cdots + |z_n|^2$.
Then the K\"ahler form is the $(1,1)$-form given by:
\be \label{eqn:Kahler}
\omega = \sqrt{-1} \pd \dbar \phi(u)
= \sqrt{-1}
(\phi'(u) \sum_{i=1}^n dz_i \wedge d\bar{z}_i
+ \phi''(u) \sum_{i=1}^n \bar{z}_i dz_i \wedge \sum_{j=1}^n z_j d\bar{z}_j).
\ee
Let $z_i = r_i e^{\sqrt{-1} \theta_i}$,
one can rewrite $\omega$ as:
\be \label{eqn:Kahler-r-theta}
\omega
= 2 \phi'(u) \sum_{i=1}^n r_i dr_i\wedge d\theta_i
+ 2 \phi''(u) \sum_{i=1}^n r_i dr_i \wedge \sum_{j=1}^n r_j^2 d\theta_j.
\ee

Let $T^n$ be the torus subgroup of $U(n)$ consisting of diagonal unitary matrices.
Restricting the $U(n)$-action to $T^n$,
one gets a Hamiltonian action with moment map:
\be \label{eqn:Moment}
(y_1, \dots, y_n) = (|z_1|^2 \phi'(u), \dots, |z_n|^2\phi'(u)).
\ee
It is easy to see that
\be \label{eqn:Kahler-symplectic}
\omega = \sum_{i=1}^n d y_i \wedge d \theta_i.
\ee
Therefore,
$(y_i, \theta_i)$ are action-angle variables,
or following Abreu \cite{Abreu} and Donaldson
in the compact toric K\"ahler case,
they will be called the {\em symplectic coordinates}.

\subsection{K\"ahler potential in symplectic coordinates}
\label{sec:Kahler-potential}

Let
\be
y:= y_1 + \cdots + y_n.
\ee
It is clear that:
\be \label{eqn:y}
y = u \cdot \phi'(u).
\ee
This was introduced by LeBrun \cite{LeBrun}.
See also \cite{Duan-Zhou1, Duan-Zhou2}.
Integrating the above differential equation,
we get the K\"ahler potential $\phi$ as  a function of $u$:
\be
\phi = \int \frac{y(u)}{u} du.
\ee

Now under suitable conditions $u$ is a function of $y$,
and so is $\phi$.
We have:
\be
\frac{d \phi}{dy} = \frac{y}{u} \frac{du}{dy} = y \frac{d \log u}{dy},
\ee
and so after integration:
\be \label{eqn:phi-in-y}
\phi = \int y d \log u = y \log u - \int \log(u) dy.
\ee
This formula expresses the K\"ahler potential in terms of the symplectic coordinates.

\subsection{Riemannian metric in symplectic coordinates}
Now let us express the Riemannian metric also in terms of the symplectic coordinates.
The Hermitian metric is given by
\be
h = 2 \phi'(u) \sum_{i=1}^n dz_i \otimes d\bar{z}_i
+ 2 \phi''(u) \sum_{i=1}^n \bar{z}_i dz_i \otimes \sum_{j=1}^n z_j d\bar{z}_j.
\ee
and so the
Riemannian metric is
\be
\begin{split}
g =&  2\phi'(u) \sum_{i=1}^n ((dr_i)^2 + r_i^2 (d\theta_i)^2)
+ 2 \phi''(u) \sum_{i,j=1}^n r_i r_j (dr_i dr_j + r_ir_jd\theta_id\theta_j) \\
= & 2\phi'(u) \sum_{i=1}^n ((dr_i)^2 + r_i^2 (d\theta_i)^2)
+ 2 \phi''(u) \biggl\{\biggl(\sum_{i=1}^n r_i dr_i\biggr)^2
+ \biggl(\sum_{i=1}^n r_i^2 d\theta_i\biggr)^2 \biggr\}.
\end{split}
\ee

\begin{theorem} \label{thm:Metric}
In symplectic coordinates the Riemannian metric $g$ takes the following form:
\be
g = \sum_{i,j=1}^n (\frac{1}{2} G_{ij} dy_i dy_j +2 G^{ij}d\theta_id\theta_j),
\ee
where the coefficients $G_{ij}$ and $G^{ij}$ are given by:
\bea
G_{ij} & = & \frac{\delta_{ij}}{y_i} - \frac{1}{y} + \frac{1}{u} \frac{du}{dy}, \\
G^{ij} & = & y_i\delta_{ij}
+ \biggl(\frac{u}{y^2\frac{du}{dy}} - \frac{1}{y}\biggr) y_iy_j.
\eea
Furthermore, the matrices $(G_{ij})$ and $(G^{ij})$ are inverse to each other.
\end{theorem}

\begin{proof}
Note we have
\be
r_i = \sqrt{\frac{y_i}{\phi'(u)}}.
\ee
So the Riemannian metric can be written as
\ben
g & = & \frac{1}{2} \sum_{i=1}^n y_i
\biggl(\frac{dy_i}{y_i} - \frac{\phi''(u)}{\phi'(u)} du \biggr)^2
+ \frac{1}{2} \phi''(u) (d u )^2 \\
& + & 2\sum_{i=1}^n y_i (d\theta_i)^2
+ 2 \frac{\phi''(u)}{\phi'(u)^2} \biggl(\sum_{i=1}^n y_i d\theta_i\biggr)^2.
\een
We first rewrite $\frac{\phi''(u)}{\phi'(u)}du$ and $\frac{\phi''(u)}{\phi'(u)^2}$ as follows.
Taking logarithmic differential of \eqref{eqn:y} one can get:
\be
\frac{\phi''(u)}{\phi'(u)}du = \biggl(\frac{1}{y} - \frac{1}{u} \frac{du}{dy}\biggr)dy.
\ee
We also have:
\ben
&& \phi''(u) = \frac{1}{u} \cdot u\phi''(u) = \frac{1}{u} ((u\phi'(u))'-\phi'(u))
= \frac{1}{u} (\frac{dy}{du} - \frac{y}{u})
= \frac{1}{u \frac{du}{dy}} - \frac{y}{u^2},
\een
\ben
\frac{\phi''(u)}{\phi'(u)^2}
= \frac{u \cdot (u\phi''(u))}{(u\phi'(u))^2}
= \frac{u (\frac{dy}{du} - \frac{y}{u} )}{y^2} = \frac{u}{y^2\frac{du}{dy}} - \frac{1}{y}.
\een

\ben
g & = & \frac{1}{2} \sum_{i=1}^n y_i
\biggl(\frac{dy_i}{y_i} - \biggl(\frac{1}{y} - \frac{1}{u} \frac{du}{dy}\biggr)dy \biggr)^2
+ \frac{1}{2} \biggl(\frac{1}{u \frac{du}{dy}} - \frac{y}{u^2}\biggr)
\biggl( \frac{du}{dy}dy \biggr)^2 \\
& + & 2\sum_{i=1}^n y_i (d\theta_i)^2
+ 2 \biggl( \frac{u}{y^2\frac{du}{dy}} - \frac{1}{y} \biggr)
\biggl(\sum_{i=1}^n y_i d\theta_i\biggr)^2 \\
& = & \frac{1}{2} \sum_{i=1}^n \frac{(dy_i)^2}{y_i}
- \sum_{i=1}^n  dy_i \biggl(\frac{1}{y} - \frac{1}{u} \frac{du}{dy}\biggr)dy
+ \frac{1}{2} \sum_{i=1}^n y_i
\biggl(\frac{1}{y} - \frac{1}{u} \frac{du}{dy}\biggr)^2(dy)^2 \\
& + &  \frac{1}{2} \biggl(\frac{u\frac{du}{dy}}{y^2} - \frac{(\frac{du}{dy})^2}{y}\biggr)
(dy )^2 \\
& + & 2\sum_{i=1}^n y_i (d\theta_i)^2
+ 2 \biggl( \frac{u}{y^2\frac{du}{dy}} - \frac{1}{y} \biggr)
\biggl(\sum_{i=1}^n y_i d\theta_i\biggr)^2.
\een
Using the fact that $y=\sum_{i=1}^n y_i$,
one can make a further simplification:
\ben
g & = & \frac{1}{2} \sum_{i=1}^n \frac{(dy_i)^2}{y_i}
- \half \biggl(\frac{1}{y} - \frac{1}{u} \frac{du}{dy}\biggr) (dy)^2 \\
& + & 2\sum_{i=1}^n y_i (d\theta_i)^2
+ 2 \biggl( \frac{u}{y^2\frac{du}{dy}} - \frac{1}{y} \biggr)
\biggl(\sum_{i=1}^n y_i d\theta_i\biggr)^2.
\een
This proves the first statement.
The second statement can be proved as follows:
\ben
\sum_{j=1}^n G_{ij} G^{jk}
& = & \sum_{j=1}^n \biggl( \frac{\delta_{ij}}{y_i} - \frac{1}{y} + \frac{1}{u} \frac{du}{dy}\biggr)
\biggl( y_j\delta_{jk}
+ \biggl(\frac{u}{y^2\frac{du}{dy}} - \frac{1}{y}\biggr) y_jy_k
\biggr)  \\
& = & \delta_{ik} + \biggl(\frac{u}{y^2\frac{du}{dy}} - \frac{1}{y}\biggr) y_k
+ \biggl(- \frac{1}{y} + \frac{1}{u} \frac{du}{dy} \biggr) y_k \\
& + & \biggl(- \frac{1}{y} + \frac{1}{u} \frac{du}{dy}\biggr)
\cdot \biggl(\frac{u}{y^2\frac{du}{dy}} - \frac{1}{y}\biggr) \sum_{j=1}^n y_jy_k \\
& = & \delta_{ik}.
\een
\end{proof}

\subsection{The complex potential in symplectic coordinates}

It is clear that
\be
G_{ij} = \frac{\pd^2\psi}{\pd y_i\pd y_j}
\ee
for the function $\psi$ defined by:
\be \label{def:psi}
\psi = \sum_{i=1}^n y_i (\ln y_i - 1)
- y (\ln y -1) +  \int \log u(y) dy.
\ee
The function $\psi$ is called the {\em complex potential}
because
\be
\half \sum_{j=1}^n \frac{\pd^2\psi}{\pd y_i \pd y_j} dy_j
+ \sqrt{-1} d\theta_i = \frac{dz_i}{z_i}
\ee
is of type $(1, 0)$.
Indeed, we have
\ben
\frac{dz_i}{z_i}
& = &\frac{dr_i}{r_i} + \sqrt{-1} d\theta_i
= \half \biggl( \frac{dy_i}{y_i} - \frac{\phi''(u)}{\phi'(u)} du \biggr)
+ \sqrt{-1} d\theta_i \\
& = & \half \biggl( \frac{dy_i}{y_i} - \frac{dy}{y} + \frac{du}{u} \biggr)
+ \sqrt{-1} d\theta_i \\
& = & \half \sum_{j=1}^n \frac{\pd^2\psi}{\pd y_i \pd y_j} dy_j
+ \sqrt{-1} d\theta_i.
\een

\subsection{Legendre transform}
Introduce the dual local coordinates $y_i^\vee$ by
\be
y_i^\vee = \frac{\pd \psi}{\pd y_i},
\ee
and introduce a dual potential function:
\be
\psi^\vee = \sum_{i=1}^n y_i \frac{\pd \psi}{\pd y_i} - \psi
\ee
By \eqref{def:psi} we get:
\be
y_i^\vee
= \log y_i - \log y
+ \log u(y) = 2 \log r_i,
\ee
and
\be
\psi^\vee = y \ln u(y) -  \int \ln u(y) dy = \phi.
\ee

\begin{theorem} \label{thm:Metric2}
The Riemannian metric $g$ satisfies:
\be
g = \frac{1}{2} \sum_{i,j=1}^n \frac{\pd^2\psi}{\pd y_i\pd y_j} dy_i dy_j
+ 2 \sum_{i,j=1}^n \frac{\pd^2\psi^\vee}{\pd y_i^\vee\pd y_j^\vee} d\theta_i d\theta_j.
\ee
In particular,
\be \label{eqn:G-upper}
G^{ij} = \frac{\pd^2\psi^\vee}{\pd y_i^\vee \pd y_j^\vee}
= \frac{\pd^2\phi}{\pd y_i^\vee \pd y_j^\vee}.
\ee
\end{theorem}

\begin{proof}
We start with:
\ben
\psi = \sum_{j=1}^n y_j y_j^\vee - \psi^\vee
\een
and take $\frac{\pd}{\pd y_i}$ on both sides to get:
\ben
\frac{\pd \psi}{\pd y_i}
&  =  & y_i^\vee + \sum_{j=1}^n y_j \frac{\pd y_j^\vee}{\pd y_i}
- \sum_{j=1}^n \frac{\pd \psi^\vee}{\pd y_j^\vee} \frac{\pd y_j^\vee}{\pd y_i} \\
& = & \frac{\pd \psi}{\pd y_i} + \sum_{j=1}^n \biggl( y_j - \frac{\pd \psi^\vee}{\pd y_j^\vee}
\biggr) \cdot \frac{\pd^2 \psi}{\pd y_i\pd y_j}.
\een
Hence
\ben
\sum_{j=1}^n \biggl( y_j - \frac{\pd \psi^\vee}{\pd y_j^\vee}
\biggr) \cdot \frac{\pd^2 \psi}{\pd y_i\pd y_j} = 0.
\een
Since the matrix $(\frac{\pd^2\psi}{\pd y_i\pd y_j})_{i,j=1, \dots, n}$ is invertible,
we have:
\be
y_j = \frac{\pd \psi^\vee}{\pd y_j^\vee}.
\ee
Take $\frac{\pd}{\pd y_i^\vee}$ on both sides:
\ben
\frac{\pd y_j}{\pd y_i^\vee} = \frac{\pd^2 \psi^\vee}{\pd y_i^\vee\pd y_j^\vee}.
\een
In other words,
the matrix $(\frac{\pd^2\psi^\vee}{\pd y_i^\vee\pd y_j^\vee})_{i,j=1, \dots, n}$
is the matrix $(\frac{\pd y_j}{\pd y_i^\vee})_{i,j=1, \dots, n}$,
hence it is the inverse matrix of $(\frac{\pd y_i^\vee}{\pd y_j})_{i,j=1, \dots, n}
= (\frac{\pd^2\psi}{\pd y_i\pd y_j})_{i,j=1, \dots, n}$,
\end{proof}

\subsection{SYZ mirror construction}

Inspired by the Strominger-Yau-Zaslow \cite{SYZ} construction,
introduce:
\be
 g^\vee = \sum_{i,j=1}^n \biggl(
\frac{\pd^2 \psi^\vee}{\pd y_i^\vee\pd y_j^\vee}dy_i^\vee dy_j^\vee
+  \frac{\pd^2\psi}{\pd y_i\pd y_j} d\theta_i^\vee d\theta_j^\vee \biggr).
\ee
Since we have:
\be
\sum_{i,j=1}^n \frac{\pd^2\psi}{\pd y_i\pd y_j} dy_idy_j
= \sum_{i,j=1}^n \frac{\pd^2 \psi^\vee}{\pd y_i^\vee\pd y_j^\vee}dy_i^\vee dy_j^\vee,
\ee
therefore,
\be
 g^\vee = \sum_{i,j=1}^n \biggl(
\frac{\pd^2 \psi}{\pd y_i\pd y_j}dy_i dy_j
+  \frac{\pd^2\psi}{\pd y_i\pd y_j} d\theta_i^\vee d\theta_j^\vee \biggr).
\ee
This has K\"ahler potential $\psi$ and complex potential $\psi^\vee$.
And one can check that
\be
\omega^\vee = \sum_{i=1}^n d y_i^\vee \wedge d \theta^\vee_i.
\ee

\subsection{Ricci form in symplectic coordinates}

\label{sec:Ricci}

It is not hard to see that
\be \label{eqn:Volume}
\frac{\omega^n}{n!}
= \sqrt{-1}^n y'(u) \big(\frac{y(u)}{u}\big)^{n-1}
\prod_{i=1}^n dz_i \wedge d\bar{z}_i.
\ee
Indeed,
because of the $U(n)$-symmetry
one can restrict to the $z_1$-axis,
where the K\"ahler form can be written as:
\be \label{eqn:KahlerForm}
\omega = \sqrt{-1} (y'(u) dz_1 \wedge d\bar{z}_1
+ \frac{y(u)}{u} \sum_{i=2}^n dz_i \wedge d\bar{z}_i ).
\ee
Now if $\omega$ is nondegenerate,
by (\ref{eqn:Volume}),
the Ricci form is given by
\be \label{eqn:Ricci}
\rho = -\sqrt{-1} \pd \dbar \Phi(u),
\ee
where
\be \label{eqn:Phi}
\Phi = \log \big[ y'(u) \big(\frac{y(u)}{u}\big)^{n-1} \big].
\ee

Similar to \eqref{eqn:Kahler-r-theta},
one has:
\be \label{eqn:Ricci-r-theta}
\rho
= -2 \Phi'(u) \sum_{i=1}^n r_i dr_i\wedge d\theta_i
- 2 \Phi''(u) \sum_{i=1}^n r_i dr_i \wedge \sum_{j=1}^n r_j^2 d\theta_j.
\ee
Comparing \eqref{eqn:Ricci-r-theta} with \eqref{eqn:Kahler-r-theta}
it is clear that $\rho = \lambda \omega$
for some constant $\lambda$ if and only if
\be
\Phi'(u) = -\lambda \phi'(u).
\ee
This equation can be solved by quadrature \cite{Duan-Zhou1}:
\be \label{eqn:Integral}
\int \frac{y^{n-1} dy}{\frac{-\lambda}{n+1} y^{n+1} + y^n + C_1} = \ln u.
\ee

\begin{prop}
The Ricci form $\rho$ can be expressed in symplectic coordinates by the following formula:
\be \label{eqn:Ricci-Symplectic}
\rho = - \sum_{i=1}^n d \biggl( \frac{u \Phi'(u)}{y} \cdot y_i \biggr) \wedge d \theta_i.
\ee
\end{prop}

\begin{proof}
Similar to the computations for the Riemannian metric we have:
\ben
\rho
& = & -\Phi'(u) \sum_{i=1}^n \frac{y_i}{\phi'(u)}
\biggl( \frac{dy_i}{y_i} - \frac{\phi''(u)}{\phi'(u)} du
\biggr)\wedge d\theta_i \\
& - & \Phi''(u) \frac{y}{\phi'(u)} \biggl( \frac{dy}{y} - \frac{\phi''(u)}{\phi'(u)} du
\biggr)  \wedge \sum_{j=1}^n \frac{y_j}{\phi'(u)} d\theta_j \\
& = & -\Phi'(u) \sum_{i=1}^n \frac{uy_i}{y}
\biggl( \frac{dy_i}{y_i} - \biggl(\frac{1}{y} - \frac{1}{u} \frac{du}{dy}\biggr)dy
\biggr)\wedge d\theta_i \\
& - & \Phi''(u) u \biggl( \frac{dy}{y} - \biggl(\frac{1}{y} - \frac{1}{u} \frac{du}{dy}\biggr)dy
\biggr)  \wedge \sum_{j=1}^n \frac{u y_j}{y} d\theta_j \\
& =& - \frac{u\Phi'(u)}{y} \sum_{i=1}^n dy_i \wedge d \theta_i
+ \biggl( \frac{u\Phi'(u)}{y^2}- \frac{\Phi'(u) + u \Phi''(u)}{y} \frac{du}{dy}\biggr) dy
\wedge \sum_{i=1}^n d\theta_i \\
& = & - \sum_{i=1}^n d \biggl( \frac{u \Phi'(u)}{y} \cdot y_i \biggr) \wedge d \theta_i.
\een
\end{proof}

Note we have:
\ben
\frac{u \Phi'(u)}{y}
& = & \frac{u}{y} \frac{\frac{d\Phi(u)}{dy}}{\frac{du}{dy}}
= \frac{u}{y \frac{du}{dy}} \frac{d}{dy}
\log \big[ \frac{1}{\frac{du}{dy}} \big(\frac{y(u)}{u}\big)^{n-1} \big] \\
& = & - \frac{u\frac{du^2}{dy^2}}{y \big(\frac{du}{dy}\big)^2}
+ \frac{(n-1)u}{y^2 \frac{du}{dy}} - \frac{n-1}{y}.
\een

\subsection{A special K\"ahler-Ricci flow}

Consider K\"ahler-Ricci flow of the form
\be \label{eqn:KRF}
\frac{d}{dt} \omega = \rho + \lambda \cdot \omega,
\ee
where $\lambda$ is some constant.
Since in our case we have
\be
\omega = \sqrt{-1} \pd \dbar \phi(u),
\ee
\be
\rho = -\sqrt{-1} \pd \dbar \Phi(u),
\ee
one can consider the following flow in the space of K\"ahler potentials:
\emph{}\be
\frac{d}{dt} \phi = - (\Phi+\lambda\phi+C_1).
\ee
By \eqref{eqn:Phi} this is just:
\be
\frac{d}{dt} \phi
= - (\log \big[ (\phi'+u\phi'') \big(\phi'\big)^{n-1} \big] + \lambda \phi + C_1).
\ee

\subsection{Scalar curvature}

\label{sec:Scalar}

The scalar curvature $R$ of $\omega$ is given by:
\be
\rho \wedge \omega^{n-1} = R \omega^n.
\ee
By \eqref{eqn:Kahler-symplectic} and \eqref{eqn:Ricci-Symplectic},
one gets:
\be
R = - \frac{1}{n} \sum_{i=1}^n \frac{\pd}{\pd y_i} \biggl( \frac{u \Phi'(u)}{y} \cdot y_i \biggr).
\ee
From this one can derive other formula for the scalar curvature
in the literature.
First of all,
the right-hand side of this formula can be rewritten as follows:
\be
n R = - y \frac{d}{d y}  \biggl(\frac{u \Phi'(u)}{y} \biggr)
 - n \frac{u \Phi'(u)}{y}.
\ee
This was derived in \cite{Duan-Zhou2} in the following way.
Along the $z_1$-axis we have
\be \label{eqn:Ricci2}
\rho =
-\sqrt{-1} ((u\Phi'(u))' dz_1 \wedge d\bar{z}_1 + \Phi'(u)
\sum_{i=2}^n dz_i \wedge d\bar{z}_i).
\ee
By (\ref{eqn:KahlerForm}), (\ref{eqn:Volume}) and (\ref{eqn:Ricci2}),
we get
\ben (u \Phi')'
\big(\frac{y}{u}\big)^{n-1} + (n-1)\Phi'
\big(\frac{y}{u}\big)^{n-2}y' =-n R y' \big(\frac{y}{u}\big)^{n-1}.
\een
This equation can be integrated to get the following:

\begin{prop} (\cite{Duan-Zhou2})
A pseudo-K\"ahler form as in (\ref{eqn:Kahler}) has a constant
scalar curvature $R$ if and only if
\be \label{eqn:CScalar}
\int \frac{y^{n-1} dy}{- \frac{R}{n+1} y^{n+1} + y^n + C_1y + C_2}
= \ln u.
\ee
\end{prop}

By the definition of $\Phi$,
\be
\Phi = \log \big[ y'(u) \big(\frac{y(u)}{u}\big)^{n-1} \big].
\ee
we get
\be
\Phi'(u) = \frac{y''(u)}{y'(u)} + (n-1) \frac{y'(u)}{y(u)} - \frac{n-1}{u}.
\ee
The right-hand side can be rewritten as a function in $y$:
\be
\Phi'(u) = - \frac{\frac{d^2u}{dy^2}}{\big(\frac{du}{dy}\bigr)^2}
+  \frac{n-1}{y \frac{du}{dy}} - \frac{n-1}{u}.
\ee
It follows that:
\be \label{eqn:R}
\begin{split}
R = &  - \frac{d u}{d y} \Phi'(u)
 - u \frac{d}{d y}  \Phi'(u)
 - (n-1) \frac{u \Phi'(u)}{y} \\
= & \frac{\frac{\pd^2u}{\pd y^2}}{\frac{du}{dy}}
+ \frac{(n-1)(n-2)}{y} + \frac{2(n-1)u\frac{\pd^2u}{\pd y^2}}{(\frac{du}{dy})^2y}
- \frac{(n-2)(n-1) u}{y^2 \frac{du}{dy}} \\
+ & \frac{u\frac{d^3u}{dy^3}}{\big(\frac{du}{dy}\big)^2}
- \frac{2u \big(\frac{d^2u}{dy^2}\big)^2}{\big(\frac{du}{dy}\big)^3}.
\end{split}
\ee

Now we generalize Abreu's formula \cite{Abreu,Donaldson}
in the compact case to
the situation of this Section:

\begin{prop}
For the K\"ahler metric considered in this Section,
the formulas for scalar curvature hold:
\be \label{eqn:Abreu}
R = - \sum_{i,j=1}^n \frac{\pd^2 G^{ij}}{\pd y_i\pd y_j}
= - \sum_{i,j=1}^n \frac{\pd^4\phi}{\pd y_i^2\pd y_j^2} .
\ee
\end{prop}

\begin{proof}
To prove the first equality, recall that
\ben
G^{ij} & = & y_i\delta_{ij}
+ \biggl(\frac{u}{y^2\frac{du}{dy}} - \frac{1}{y}\biggr) y_iy_j.
\een
Hence we have:
\ben
&& \sum_{i=1}^n \frac{\pd G^{ij}}{\pd y_i}
= \sum_{i=1}^n \biggl[ \delta_{ij}
+ \biggl(\frac{u}{y^2\frac{du}{dy}} - \frac{1}{y}\biggr) (y_j + \delta_{ij} y_i) \\
& + & \biggl(\frac{\frac{du}{dy}}{y^2\frac{du}{dy}}
-\frac{2u}{y^3\frac{du}{dy}} -\frac{u}{y^2(\frac{du}{dy})^2} \frac{d^2u}{d y^2}
+ \frac{1}{y^2}\biggr) y_iy_j \biggr] \\
& = & 1 + (n+1) \biggl(\frac{u}{y^2\frac{du}{dy}} - \frac{1}{y}\biggr) y_j
+ \biggl(\frac{2}{y} -\frac{2u}{y^2\frac{du}{dy}}
-\frac{u}{y(\frac{du}{dy})^2} \frac{d^2u}{d y^2} \biggr) y_j \\
& = & 1
+ \biggl(-\frac{n-1}{y}+\frac{(n-1)u}{y^2\frac{du}{dy}}
-\frac{u}{y(\frac{du}{dy})^2} \frac{d^2u}{d y^2} \biggr) y_j.
\een
Furthermore,
\ben
&&  - \sum_{i,\emph{}j=1}^n \frac{\pd^2 G^{ij}}{\pd y_i\pd y_j}
= - \sum_{j=1}^n \biggl[
\biggl(-\frac{n-1}{y}+\frac{(n-1)u}{y^2\frac{du}{dy}}
-\frac{u}{y(\frac{du}{dy})^2} \frac{d^2u}{d y^2} \biggr)  \\
& + & \biggl(\frac{n-1}{y^2}
+\frac{(n-1)\frac{du}{dy}}{y^2\frac{du}{dy}}
-\frac{2(n-1)u}{y^3\frac{du}{dy}}
-\frac{(n-1)u}{y^2(\frac{du}{dy})^2} \frac{d^2u}{d y^2}  \\
&- & \frac{\frac{du}{dy}}{y(\frac{du}{dy})^2} \frac{d^2u}{d y^2}
+\frac{u}{y^2(\frac{du}{dy})^2} \frac{d^2u}{d y^2}
+\frac{2u}{y(\frac{du}{dy})^3} \biggl(\frac{d^2u}{d y^2}\biggr)^2
- \frac{u}{y(\frac{du}{dy})^2} \frac{d^3u}{d y^3}
\biggr) y_j \biggr]\\
& = & \frac{(n-1)(n-2)}{y} - \frac{(n-1)(n-2)u}{y^2\frac{du}{dy}}
+\frac{2(n-1)u}{y(\frac{du}{dy})^2} \frac{d^2u}{d y^2} \\
& + & \frac{\frac{d^2u}{dy^2} }{(\frac{du}{dy})^2}
- \frac{2u}{(\frac{du}{dy})^3} \biggl(\frac{d^2u}{d y^2}\biggr)^2
+ \frac{u}{(\frac{du}{dy})^2} \frac{d^3u}{d y^3}.
\een
The right-hand side of the last equality equals to $R$
by \eqref{eqn:R}.
This proves the first equation in \eqref{eqn:Abreu},
and the second equation follows from \eqref{eqn:G-upper}.
\end{proof}

\subsection{Derivations of Ricci curvature and scalar curvature in Hessian geometry}

In \S \ref{sec:Ricci} and \S \ref{sec:Scalar}
we have derived formulas for Ricci curvature and the scalar curvature, respectively,
in the context of $U(n)$-symmetric K\"ahler metrics.
In this Subsection,
we will derive the corresponding formulas in the context of Hessian geometry.

Suppose that $\psi$ is a convex function
on a domain $\Omega \subset \bR^n$ endowed
with linear coordinates $\{y_i\}$.
Introduce the dual local coordinates $y_i^\vee$ by
\be
y_i^\vee = \frac{\pd \psi}{\pd y_i},
\ee
and introduce a dual potential function:
\be
\psi^\vee = \sum_{i=1}^n y_i \frac{\pd \psi}{\pd y_i} - \psi.
\ee
Define a Riemannian metric $g$ by:
\be
g = \frac{1}{2} \sum_{i,j=1}^n G_{ij} dy_i dy_j
+ 2 \sum_{i,j=1}^n G^{ij} d\theta_i d\theta_j.
\ee
where the coefficients are defined by:
\begin{align}
G_{ij} & = \frac{\pd^2\psi}{\pd y_i \pd y_j}, &
G^{ij} & = \frac{\pd^2\psi^\vee}{\pd y_i^\vee \pd y_j^\vee}.
\end{align}
One can check that the matrices $(G_{ij})$ and $(G^{ij})$ are inverse to each other,
and
\be
g= 2 \sum_{i,j=1}^n G^{ij} (\frac{1}{4} dy_i^\vee dy_j^\vee+d\theta_i d\theta_j).
\ee
Introduce complex coordinates $w_i$:
\be
w_i = \half y_i^\vee + \sqrt{-1} \theta_i.
\ee
Then $g$ is a K\"ahler metric with K\"ahler form:
\be \label{eqn:omega-Hessian}
\omega = \sqrt{-1} \sum_{i,j=1}^n G^{ij} d w_i \wedge d\bar{w}_j
= \sqrt{-1} \sum_{i,j=1}^n \frac{\pd^2\psi^\vee}{\pd y_i^\vee \pd y_j^\vee}
 d w_i \wedge d\bar{w}_j.
\ee
Since
\begin{align}
\frac{\pd}{\pd w_i} & = \half
\biggl(2\frac{\pd}{\pd y_i^\vee} - \sqrt{-1} \frac{\pd}{\pd \theta_i}\biggr), &
\frac{\pd}{\pd \bar{w}_i} &
= \half\biggl(2\frac{\pd}{\pd y_i^\vee} + \sqrt{-1} \frac{\pd}{\pd \theta_i}\biggr),
\end{align}
we have
\be
\omega
= \sqrt{-1} \sum_{i,j=1}^n \frac{\pd^2\psi^\vee}{\pd w_i \pd \bar{w}_j}
 d w_i \wedge d\bar{w}_j.
\ee
Therefore,
its Ricci form
\be \label{eqn:Rho-Hessian}
\begin{split}
\rho = & - \sqrt{-1} \sum_{i,j=1}^n \frac{\pd^2\log\det (G^{ij})}{\pd w_i \pd \bar{w}_j}
 d w_i \wedge d\bar{w}_j \\
 = & - \sqrt{-1} \sum_{i,j=1}^n \frac{\pd^2\log\det (G^{ij})}{\pd y^\vee_i \pd y^\vee_j}
 d w_i \wedge d\bar{w}_j.
 \end{split}
\ee
The scalar curvature $R$ is then
\be
R = - \sqrt{-1} \sum_{i,j=1}^n G_{ij}
\frac{\pd^2\log\det (G^{ij})}{\pd y^\vee_i \pd y^\vee_j}.
\ee
By the computations in Abreu \cite{Abreu},
\be
R = - \sum_{i,j=1}^n \frac{\pd^2G^{ij}}{\pd y_i\pd y_j}.
\ee

\subsection{K\"ahler-Ricci flow in Hessian geometry}

\label{sec:KRF-Hessian}\emph{}

One can consider K\"ahler-Ricci flow of the form \eqref{eqn:KRF}
with $\omega$ and $\rho$ given by \eqref{eqn:omega-Hessian} and \eqref{eqn:Rho-Hessian}.
This leads us to the following flow:
\be \label{eqn:KRF-Hessian}
\frac{d}{dt} \psi^\vee = -\log \det \biggl( \frac{\pd^2\psi^\vee}{\pd y_i^\vee \pd y_j^\vee}
\biggr) + \lambda \psi^\vee.
\ee
Dually,
one can consider the flow for the complex potential:
\be \label{eqn:CF-Hessian}
\frac{d}{dt} \psi = -\log \det \biggl( \frac{\pd^2\psi}{\pd y_i \pd y_j}
\biggr) + \lambda \psi.
\ee

\section{The Case of $U(n)$-Symmetric K\"aler Ricci-Flat Metrics}

\label{sec:U(n)-symmetric}

In this Section we apply the results developed in last Section to the case
of some $U(n)$-symmetric K\"ahler Ricci-flat metrics on $\cO_{\bP^{n-1}}(-n)$
(see e.g Duan-Zhou \cite{Duan-Zhou1}).
We will first work on $\bC^n - \{0\}$,
then consider the extension from $(\bC^n - \{0\})/\bZ_n$.
When $n=2$,
one recovers the Kepler metric on $K_2$ and the Eguchi-Hanson metric.

\subsection{$U(n)$-symmetric K\"ahler Ricci-flat metrics}

If $\lambda = 0$ i.e. $\rho = 0$ in (\ref{eqn:Integral}) then we get:
\be y(u)^n = C_3 u^n + C_4 \ee
for some constants $C_3, C_4$.
When $C_3 = 1$ and $C_4 > 0$, write $C_4 = b^n$ for
some $b > 0$, then we have $y(u) = (u^n+b^n)^{1/n}$ and so the K\"ahler
potential is given by:
\ben \phi(u) = \int \frac{(u^n + b^n)^{1/n}}{u} du. \een
It satisfies
\begin{align} \label{eqn:phi'}
\phi'(u) & = \frac{(u^n + b^n)^{1/n}}{u}, & \phi''(u) & =  -
\frac{b^n}{u^2 (u^n+b^n)^{(n-1)/n}},
\end{align}
and so the K\"ahler form is given by:
\ben
\omega & = & \sqrt{-1} (\phi'(u) \pd \dbar u + \phi''(u) \pd u \wedge \dbar u) \\
& = & \sqrt{-1}\biggl(
\frac{(u^n + b^n)^{1/n}}{u} (dz^1 \wedge d\bar{z}^1 + \cdots + dz^n\wedge d \bar{z}^n) \\
& - &  \frac{b^n}{u^2 (u^n+b^n)^{(n-1)/n}} (\bar{z}^1 dz^1 + \cdots
+ \bar{z}^n dz^n) \wedge (z^1d\bar{z}^1 + \cdots + z^n d \bar{z}^n)
\biggr).
\een
When $b \to 0$,
we recover the flat metric on $\bC^n$:
\ben
\omega = \sqrt{-1} \sum_{i=1}^n dz^i \wedge d\bar{z}^i.
\een

By \eqref{eqn:Moment} and \eqref{eqn:phi'},
the moment map of the diagonal torus subgroup of $U(n)$ is given by
\be
(y_1, \dots, y_n) =
\biggl(|z_1|^2 \frac{(u^n + b^n)^{1/n}}{u}, \dots,
|z_n|^2\frac{(u^n + b^n)^{1/n}}{u}\biggr).
\ee
The image of the moment map is the convex body:
\be
\{(y_1, \dots, y_n) \in \bR^n \;|\; y_j \geq 0, \;\;j=1, \dots, n, \;\;
y_1 + \cdots + y_n \geq b \}.
\ee
When $b> 0$,
the convex body has $n$ vertex points.
When $b=0$,
the convex body is a simplex.

By \eqref{eqn:y} and \eqref{eqn:phi'},
we now have
\be \label{eqn:y-}
y = (u^n+b^n)^{1/n},
\ee
and so
\be
u = (y^n-b^n)^{1/n}.
\ee
Hence by \eqref{eqn:phi-in-y},
\be
\phi = \int \frac{y^n}{y^n-b^n} d y\emph{}
= \frac{b}{n} \sum_{j=0}^{n-1} \xi_n^j \ln (y- \xi_n^jb) + C.
\ee
And by \eqref{def:psi},
the complex potential is:
\be
\psi = \sum_{i=1}^n y_i (\ln y_i - 1)
- y (\ln y -1) + \frac{1}{n} \sum_{j=0}^{n-1} (y-b\xi_n^j) (\log (y- b\xi_n^j) - 1)-C.
\ee
It is interesting to compare this formula 
with the formula of Guillemin \cite{Gui}
in the case of compact toric manifolds.

The dual local coordinates $y_i^\vee$ are then
\be
y_i^\vee = \frac{\pd \psi}{\pd y_i}
= \ln y_i - \ln y + \frac{1}{n} \sum_{j=0}^{n-1} \log (y- b\xi_n^j)
= \ln \frac{y_i (y^n-b^n)^{1/n}}{y},
\ee
and from this we find:
\be
y =
\biggl(\big(\sum_{j=1}^n e^{y_j^\vee}\big)^n + b^n \biggr)^{1/n}
\ee
and
\be
y_i = \frac{e^{y_i^\vee}}{\sum_{j=1}^n e^{y_j^\vee}}
\biggl(\big(\sum_{j=1}^n e^{y_j^\vee}\big)^n + b^n \biggr)^{1/n}.
\ee
It follows that $\psi^\vee = \phi$ can be written in terms of $y_i^\vee$ as follows:
\be
\psi^\vee
= \frac{b}{n} \sum_{j=0}^{n-1} \xi_n^j
\ln \biggl(\biggl(\big(\sum_{i=1}^n e^{y_i^\vee}\big)^n + b^n \biggr)^{1/n}- \xi_n^jb\biggr) + C.
\ee
One checks that
\ben
\frac{\pd \psi^\vee}{\pd y_k^\vee}
& = & \frac{b}{n} \sum_{j=0}^{n-1} \xi_n^j
\frac{\biggl(\big(\sum_{i=1}^n e^{y_i^\vee}\big)^n + b^n \biggr)^{1/n-1}
\big(\sum_{i=1}^n e^{y_i^\vee}\big)^{n-1} e^{y_k^\vee}}
{\biggl(\big(\sum_{i=1}^n e^{y_i^\vee}\big)^n + b^n \biggr)^{1/n}- \xi_n^jb} \\
& = & \frac{b}{n} \sum_{j=0}^{n-1} \xi_n^j
\frac{y^{-(n-1)}
\big(\sum_{i=1}^n e^{y_i^\vee}\big)^{n-1} e^{y_k^\vee}}
{y- \xi_n^jb} \\
& = & y^{-(n-1)} \big(\sum_{i=1}^n e^{y_i^\vee}\big)^{n-1} e^{y_k^\vee} \cdot \frac{y^n}{y^n-b^n} \\
& = & \frac{e^{y_k^\vee}}{\sum_{j=1}^n e^{y_j^\vee}}
\biggl(\big(\sum_{j=1}^n e^{y_j^\vee}\big)^n + b^n \biggr)^{1/n} \\
& = & y_k.
\een

By Theorem \ref{thm:Metric} and Theorem \ref{thm:Metric2},
the Riemannian metric $g$ takes the following form:
\ben
g & = & \sum_{i,j=1}^n (\frac{1}{2} G_{ij} dy_i dy_j +2 G^{ij}d\theta_id\theta_j) \\
& = & \frac{1}{2} \sum_{i,j=1}^n \frac{\pd^2\psi}{\pd y_i\pd y_j} dy_i dy_j
+ 2 \sum_{i,j=1}^n \frac{\pd^2\psi^\vee}{\pd y_i^\vee\pd y_j^\vee} d\theta_i d\theta_j,
\een
where the coefficients $G_{ij}$ and $G^{ij}$ are now given by:
\ben
G_{ij} & = & \frac{\delta_{ij}}{y_i} - \frac{1}{y} + \frac{y^{n-1}}{y^n-b^n}
= \frac{\delta_{ij}}{y_i} + \frac{b^n}{y(y^n-b^n)}, \\
G^{ij} & = & y_i\delta_{ij}
+ \biggl(\frac{y^n-b^n}{y^{n+1}} - \frac{1}{y}\biggr) y_iy_j
= y_i\delta_{ij} - \frac{b^n}{y^{n+1}}y_iy_j.
\een
Their determinants are
\ben
&& \det (G_{ij}) = \frac{y^n}{y_1 \cdots y_n (y^n-b^n)}, \\
&& \det (G^{ij}) = \frac{y_1 \cdots y_n (y^n-b^n)}{y^n}
= \exp (y_1^\vee + \cdots + y_n^\vee).
\een
By \eqref{eqn:Rho-Hessian},
the Ricci form for $g$ is
\ben
\rho = -\sqrt{-1} \sum_{i,j=1}^n \frac{\pd^2}{\pd y_i^\vee\pd y_j^\vee}
(y_1^\vee + \cdots + y_n^\vee) \cdot d w_i \wedge d \bar{w}_j = 0.
\een
Let $g^\vee$ be the metric defined by:
\ben
&& g^\vee = \sum_{i,j=1}^n \biggl(
\frac{\pd^2 \psi^\vee}{\pd y_i^\vee\pd y_j^\vee}dy_i^\vee dy_j^\vee
+  \frac{\pd^2\psi}{\pd y_i\pd y_j} d\theta_i^\vee d\theta_j^\vee \biggr).
\een
Then its Ricci form is
\ben
\rho^\vee & = & - \sqrt{-1} \sum_{i,j=1}^n \frac{\pd^2}{\pd y_i\pd y_j}
\log \frac{y^n}{y_1 \cdots y_n (y^n-b^n)}
\cdot d w_i \wedge d \bar{w}_j.
\een
Clearly,
$g^\vee$ is not Ricci-flat.
To make $g^\vee$ Ricci-flat in the literature of SYZ Conjecture
it was proposed to quantum correct the complex  structure $J^\vee$
(see e.g. \cite{Auroux, Chan-Lau-Leung}).

\subsection{Quotient by $\bZ_n$}
\label{sec:Quotient}

The above family of metrics on $\bC^n - \{0\}$ is invariant under the
$\bZ/n\bZ$-action:
$$(z^1, \dots, z^n) \mapsto ( e^{2\pi i/n}z^1,
\dots, e^{2\pi i/n} z^n).$$
Since the quotient space
$(\bC^n - \{0\})/(\bZ/n\bZ)$ can be identified with
$\cO_{\bP^{n-1}}(-n) - \bP^{n-1}$,
one obtains a family of K\"ahler Ricci-flat metrics on
the latter space.
To get explicit expressions,
make the following change of variables:
\begin{align} \label{eqn:Changes}
z^1 &= v^{1/n}, & z^2 & = v^{1/n}w^2, & \dots && z^n & = v^{1/n}w^n.
\end{align}
Equivalently,
\begin{align}
v & = (z^1)^n, & w^2 & = \frac{z^2}{z^1}, & \cdots &&
w^n & = \frac{z^n}{z^1}.
\end{align}
Then $\omega$ becomes
\ben
\hat{\omega}_b & = & \sqrt{-1} \biggl( \frac{(1+|w|^2)^n}{n^2(|v|^2(1+|w|^2)^n+b^n)^{(n-1)/n}} dv \wedge d\bar{v} \\
& + & \sum_{i=2}^n \frac{\bar{v} w^i (1+|w|^2)^{n-1}}{n(|v|^2(1+|w|^2)^n+b^n)^{(n-1)/n}} dv \wedge d\bar{w}^i \\
& - & \sum_{i=2}^n \frac{v \bar{w}^i (1+|w|^2)^{n-1}}{n(|v|^2(1+|w|^2)^n+b^n)^{(n-1)/n}} d\bar{v} \wedge dw^i \\
& + & \frac{(|v|^2(1+|w|^2)^n+b^n)^{1/n}}{(1+|w|^2)} \sum_{i=2}^n dw^i \wedge d\bar{w}^i \\
& - & \frac{b^n}{(1+|w|^2)^2(|v|^2(1+|w|^2)^n+b^n)^{(n-1)/n}}
\sum_{i=2}^n \bar{w}^idw^i \wedge \sum_{j=2}^nw^j d\bar{w}^j\biggr).
\een
When $b > 0$,
$\hat{\omega}_b$ defines a K\"ahler metric on the total space
of the bundle $\cO_{\bP^{n-1}}(-n)$.
This family of metrics is the Calabi metrics \cite{Calabi} on $\cO_{\bP^{n-1}}(-n)$.
It is well-known that when $n=2$, this is the Eguchi-Hanson metric \cite{Egu-Han}.
We will also explicitly check this in \S \ref{sec:E-H}.

When $b\to 0$, $\hat{\omega}_b$ becomes:
\ben
\hat{\omega}_0 & = & \sqrt{-1} \biggl(
\frac{(1+|w|^2)}{n^2|v|^{2(n-1)/n}} dv \wedge d\bar{v}
+ \sum_{i=2}^n \frac{\bar{v} w^i}{n|v|^{2(n-1)/n}} dv \wedge d\bar{w}^i \\
& - & \sum_{i=2}^n \frac{v \bar{w}^i }{n|v|^{2(n-1)/n}} d\bar{v}
\wedge dw^i + |v|^{2/n}\sum_{i=2}^n dw^i \wedge d\bar{w}^i \biggr).
\een
When $\cO_{\bP^{n-1}}(-n)$ is blown down to $\bC^n/\bZ_n$,
$\hat{\omega}_0$ is transformed to the orbifold flat K\"ahler metric.

Let $T \subset U(n)$ be the group of diagonal $n\times n$ unitary matrices.
It acts naturally on $\bC^n$:
\ben
(e^{i\theta_1}, \dots, e^{i\theta_n}) \cdot (z_1, \dots, z_n)
= (e^{i\theta_1}z_1, \dots, e^{i\theta_n}z_n).
\een
This action induces an action on $\cO_{\bP^{n-1}}(-n)$,
in local coordinates $(v, w^2, \dots, w^n)$,
it is given by:
\ben
(e^{i\theta_1}, \dots, e^{i\theta_n}) \cdot (v, w^2 \dots, w^n)
= (e^{n i\theta_1}v,e^{i(\theta_2-\theta_1)}w^2,\dots, e^{i(\theta_n-\theta_1)}w^n).
\een
Make the following change of coordinates:
\ben
\alpha_1 & = & n \theta_1, \\
\alpha_j & = & \theta_j - \theta_1, \;\;\;\; j=2, \dots, n.
\een
This suggests a different torus action,
defined in local coordinates $(v, w^2, \dots, w^n)$ by:
\ben
(e^{i\alpha_1}, \dots, e^{i\alpha_n}) \cdot (v, w^2 \dots, w^n)
= (e^{i\alpha_1}v,e^{i\alpha_2}w^2,\dots, e^{i\alpha_n}w^n).
\een
It has the following moment map:
\ben
(x_1, \dots, x_n) = (ny_1, y_2-y_1, \dots, y_n - y_n),
\een
whose image is
\be
x_1 \geq 0, \;\;\; x_j + \frac{1}{n} x_1 \geq 0, \;\;\; j =2, \dots, n, \;\;\;
x_1 + \cdots + x_n  \geq a.
\ee
Denote by $x_i^\vee$ the dual coordinates:
\be
x_i^\vee = \frac{\pd \psi}{\pd x_i}.
\ee
Then we have
\be
x_1^\vee = \frac{1}{n}\sum_{i=1}^n y_i^\vee, \;\;\;\;\;
x_j^\vee = y_j^\vee, \;\;\; j=2, \dots, n.
\ee


\section{Kepler Metric on $K_2$ and Eguchi-Hanson Metrics by Calabi Ansatz}

\label{sec:E-H}

In this Section we will rederive the Kepler metric on $K_2$
and the Eguchi-Hanson metric by applying Calabi Ansatz on $\cO_{\bP^1}(-2)$.
This will lead us to a discussion of the Kepler metric on $K_3$ in the next two Sections.

\subsection{Calabi Ansatz on $\cO_{\bP^1}(-1)$}
Consider local coordinates $(z, w)$, $(\tilde{z}, \tilde{w})$ on $\cO_{\bP^1}(-1)$ related to each other by:
\begin{align*}
\tilde{z} & = \frac{1}{z}, & \tilde{w} & = zw, \\
z & = \frac{1}{\tilde{z}}, & w & = \tilde{z}\tilde{w}.
\end{align*}
Define an Hermitian metric on $\cO_{\bP^1}(-1)$ by
\be
u = r^2=  |w|^2 (1+|z|^2).
\ee
Consider a K\"ahler metric of the form:
\be \label{eqn:CA1}
\omega = a \pi^*\omega_{\bP^1} + \sqrt{-1} \pd \dbar \phi(u)
= \sqrt{-1} a \pd \dbar \log (1+|z|^2) + \sqrt{-1} \pd \dbar \phi(u).
\ee
By a computation similar to that in \S \ref{sec:Hessian},
we have
\ben
\omega & = & \sqrt{-1} \biggl\{ \frac{1}{(1+|z|^2)^2}  (a + y + u|z|^2 y') dz \wedge d\bar{z}  \\
& + & y' \bar{z} w  \cdot dz  \wedge d\bar{w}
  + y' z \bar{w}  \cdot dw  \wedge d \bar{z}
+  y' \cdot (1+|z |^2)dw  \wedge d \bar{w}
\biggr\},
\een
where  $y(u) = u \phi'(u)$.
It follows that
\ben
\Phi & = & \frac{\omega^2/2!}{(\sqrt{-1})^2 dz \wedge d\bar{z} \wedge dw \wedge d\bar{w}} \\
& = & \frac{1}{1+|z|^2} \biggl(a +  y + u|z |^2 y'\biggr) \cdot y'
-  \frac{u|z|^2}{1+|z|^2} \big( y' \big)^2
= \frac{(a+y)y'}{1+|z|^2}.
\een
Since we have
\ben
\pd\dbar \log \Phi & = & - \frac{dz \wedge d\bar{z}}{(1+|z|^2)^2}
+ \frac{y''}{y'} \pd \dbar u + \frac{y'y'''-(y'')^2}{(y')^2} \pd u \wedge \dbar u\\
& + & \frac{y'}{a+y} \pd\dbar u + \frac{y''(a+y)- (y')^2}{(a+y)^2} \pd u \wedge \dbar u \\
& = & - \frac{dz \wedge d\bar{z}}{(1+|z|^2)^2} \\
& + & (\log((a+y)y'))'\cdot (|w|^2dz \wedge d \bar{z} + z\bar{w} dw \wedge d\bar{z}
+ \bar{z} w dz \wedge d \bar{w} + (1+|z|^2) dw \wedge d \bar{w}) \\
& + & (\log((a+y)y'))''\cdot (|w|^4|z|^2 d z \wedge d \bar{z} + u \bar{z}w dz \wedge d\bar{w}
+ u z \bar{w} dw \wedge d\bar{z} + u (1+|z|^2) dw \wedge d \bar{w}),
\een
to get a K\"ahler Ricci-flat metric,
we need all the coefficients of $\pd\dbar \Phi$ to vanish.
It suffices  to solve the following system of two equations:
\ben
&&  - \frac{1}{(1+|z|^2)^2} + (\log ((a+y)y'))' \cdot |w|^2
+ (\log ((a+y)y'))'' \cdot |w|^4 |z|^2 = 0, \\
&& (\log ((a+y)y'))'   + (\log ((a+y)y'))'' \cdot u =0.
\een
From the second equation
\be
(\log ((a+y)y'))'  \cdot u = c,
\ee
plug this into the first equation:
\ben
&&  - \frac{1}{(1+|z|^2)^2} + \frac{c|w|^2}{u} - \frac{c \cdot |w|^4 |z|^2}{u^2} = 0,  .
\een
From this one gets $c=1$,
and so
\be
(\log (y'(a+y)))' = \frac{1}{u},
\ee
hence we can follow the following steps to get the solutions:
\be
(a+y)y' = C_1 u,
\ee
\be
y = - a \pm \sqrt{C_1u^2 + C_2},
\ee

\be
\phi(u) = \int \frac{-a \pm \sqrt{C_1u^2+C_2}}{u} du.
\ee
We will take the plus sign and set $C_1 = 1$, $C_2 = b^2$.
Then the K\"ahler form becomes:
\ben
\omega
& = & \sqrt{-1} \biggl\{ \frac{1}{(1+|z|^2)^2}
\biggl( \sqrt{u^2+b^2} + u|z|^2 \cdot \frac{u}{\sqrt{u^2+b^2}}
\biggr) dz \wedge d \bar{z} \\
& + &  \frac{u}{\sqrt{u^2 + b^2}} \bar{z} w  \cdot dz  \wedge d\bar{w}
  +  \frac{u}{\sqrt{u^2 + b^2}} z \bar{w}  \cdot dw  \wedge d \bar{z} \\
& + &  \frac{u}{\sqrt{u^2 + b^2}} \cdot (1+|z |^2)dw  \wedge d \bar{w}
\biggr\}.
\een
This metric degenerate along the zero section on $\cO_{\bP^1}(-1)$,
so we shift to $\cO_{\bP^1}(-1)$.
This can be achieved by taking
\begin{align*}
v & = w^2, & \tilde{v} & = \tilde{w}^2.
\end{align*}
Then the K\"ahler form becomes:
\ben
\omega
& = & \sqrt{-1} \biggl\{ \frac{1}{(1+|z|^2)^2}
\biggl( \sqrt{|v|^2(1+|z|^2)^2+b^2} + \frac{|v|^2(1+|z|^2)^2|z|^2}{\sqrt{|v|^2(1+|z|^2)^2+b^2}}
\biggr) dz \wedge d \bar{z} \\
& + &  \frac{(1+|z|^2)}{\sqrt{|v|^2(1+|z|^2)^2 + b^2}} \bar{z} v  \cdot dz  \wedge d\bar{v}
  +  \frac{(1+|z|^2)}{\sqrt{|v|^2(1+|z|^2)^2 + b\emph{}^2}} z \bar{v}  \cdot dv  \wedge d \bar{z} \\
& + &  \frac{(1+|z |^2)^2}{4\sqrt{|v|^2(1+|z|^2)^2 + b^2}} \cdot dv  \wedge d \bar{v}
\biggr\}.
\een
This matches with the metric $\hat{\omega}_b$ for $n=2$ in \S \ref{sec:Quotient}.

\subsection{Polar coordinates} \label{sec:Polar}
Use the stereographic projection to get:
\begin{align}
x^0 & = \frac{1-|z|^2}{|z|^2+1}, & x^1+\sqrt{-1} x^2 = \frac{2z}{|z|^2+1},
\end{align}
and
\be
z= \frac{x^1+\sqrt{-1}x^2}{x^0+1}.
\ee
In terms of the Euler angles,
\begin{align}
x^0 & = \cos \theta, & x^1 & = \sin \theta \cos \varphi, & x^2 & = \sin \theta \sin \varphi.
\end{align}
It follows that
\be
z = e^{i \varphi} \tan (\frac{\theta}{2}),
\ee
and so we can find from
\be
1+|z|^2= \frac{1}{\cos^2(\frac{\theta}{2})},
\ee
and
\be
dz = i e^{i\phi} \tan(\frac{\theta}{2}) d \phi
+\half \frac{e^{i\phi}}{\cos^2 (\frac{\theta}{2})} d\theta
\ee
the following formula:
\be
\frac{|dz|^2}{(|z|^2+1)^2} = \frac{1}{4} (d\theta^2 + \sin^2\theta d\phi^2).
\ee
Let
\be
w = r \cos (\frac{\theta}{2})  \cdot e^{i\beta}.
\ee
Its differential is given by:
\be
dw = r \cos (\frac{\theta_1}{2}) \cdot e^{i\beta}
(\frac{dr}{r} - \frac{1}{2} \tan(\frac{\theta}{2})d\theta+ i d\beta).
\ee

Now the Riemannian metric associated to \eqref{eqn:CA1} is given by:
\ben
g & = & \frac{1}{(1+|z|^2)^2}
\biggl(a + y(u) +|z|^2uy'(u)\biggr) |dz|^2   \\
& + & 2(1+|z|^2) y'(u) \cdot \Re (\bar{w} z d w d\bar{z}) \\
& + & y'(u) \cdot \frac{u}{|w|^2} \cdot |dw|^2.
\een
It can be rewritten as:
\be
g = \frac{a + y(u)}{(1+|z|^2)^2} |dz|^2
+ u y'(u) \cdot |\gamma|^2,
\ee
where
\be
\gamma= \frac{\bar{z}}{1+|z|^2} dz
+ \frac{\bar{w}}{|w|^2} d w.
\ee
It is easy to find
\be
\gamma = \frac{dr}{r} + i (d\beta + \sin^2(\frac{\theta}{2}) d\varphi),
\ee
so we get:

\ben
g & = &  y'(u) dr^2 + \frac{1}{4} \biggl(a + y(u)\biggr)
(d\theta^2 + \sin^2\theta d\varphi^2)
+ r^2y'(u) \cdot (d\beta + \sin^2(\frac{\theta}{2}) d\varphi)^2.
\een
When
\be
y = - a + \sqrt{u^2 + b^2},
\ee
the metric is
\ben
g = \frac{r^2}{\sqrt{r^4+b^2}} dr^2
+ \frac{\sqrt{r^4+b^2}}{4}
(d\theta^2 + \sin^2\theta d\varphi^2)
+ \frac{r^4}{\sqrt{r^4+b^2}} \cdot (d\beta + \sin^2(\frac{\theta}{2}) d\varphi)^2.
\een
Let $\psi = 2\beta - \varphi$,
one can rewrite it as
\ben
g = \frac{r^2}{\sqrt{r^4+b^2}} dr^2
+ \frac{\sqrt{r^4+b^2}}{4}(d\theta^2 + \sin^2\theta d\varphi^2)
+ \frac{r^4}{4\sqrt{r^4+b^2}} \cdot (d\psi + \cos\theta d\varphi)^2.
\een
Finally,
letting $s= (r^4+b^2)^{1/4}$,
\be
g = \biggl(1- \frac{b^2}{s^4}\biggr)^{-1}
+ \frac{s^2}{4} \biggl(1- \frac{b^2}{s^4}\biggr)
(d\psi + \cos\theta d\varphi)^2
+ \frac{s^2}{4}(d\theta^2 + \sin^2\theta d\varphi^2).
\ee
This is the standard form of the Euguchi-Hanson metric \cite{Egu-Han}.

When $b=0$,
\be
g= dr^2 + \frac{r^2}{4} ((d\theta^2 + \sin^2\theta d\varphi^2)
+ (d\psi + \cos\theta d\varphi)^2).
\ee
This recovers the Kepler metric on $K_2$,
and its shows that the standard metric on $\bR\bP^3$
up to a constant is a Sasaki-Einstein metric.

\section{Kepler Metric on 3-Conifold and Related Metrics by Calabi Ansatz on  $\cO_{\bP^1}(-1) \oplus \cO_{\bP^1}(-1)$}

\label{sec:CA-1}

In this Section we will study the Kepler metric on $K_3$
and some related metrics from the point of view of both \S \ref{sec:E-H}
and \S \ref{sec:Hessian}.

\subsection{Calabi Ansatz on  $\cO_{\bP^1}(-1) \oplus \cO_{\bP^1}(-1)$}

Consider local coordinates $(z, w_1, w_2)$
and  $(\tilde{z}, \tilde{w}_1, \tilde{w}_2)$ on
the total space  of the rank two vector bundle $\cO_{\bP^1}(-1) \oplus \cO_{\bP^1}(-1)$
related to each other by the following  formulas:
\begin{align*}
\tilde{z} & = \frac{1}{z}, & \tilde{w}_j & = zw_j, \\
z & = \frac{1}{\tilde{z}}, & w_j & = \tilde{z}\tilde{w}_j.
\end{align*}
Define an Hermitian metric by setting the Hermitian norm square
of the element in the vector bundle with local coordinate $(z, w_1, w_2)$ to be:
\be
u = r^2=  |w|^2 (1+|z|^2),
\ee
where $|w|^2= \sum_{j=1}^2 |w_j|^2$.
As a special case of a construction of Calabi \cite{Calabi},
we consider K\"ahler metrics of the form:
\be \label{eqn:CA-Conifold}
\begin{split}
\omega_a & = a \pi^*\omega_{\bP^1} + \sqrt{-1} \pd \dbar \phi(u)\\
& = \sqrt{-1} \frac{a dz  \wedge d\bar{z} }{(1+|z |^2)^2}
+ \sqrt{-1} \phi'(u) \cdot  \pd\dbar u
+ \sqrt{-1} \phi''(u) \cdot \pd u \wedge \dbar u,
\end{split}
\ee
where $\phi$ is a suitable function in $u$.
Since we have
\ben
&& \pd u = |w|^2 \bar{z} dz + (1+|z|^2)\sum_{j=1}^2 \bar{w}_j dw_j, \\
&& \dbar u = |w|^2 z d\bar{z} + (1+|z|^2) \sum_{j=1}^2 w_j d\bar{w}_j, \\
\pd u \wedge \dbar u 
& = & |w|^4|z|^2 d z \wedge d \bar{z} + u \sum_{j=1}^2 \bar{z}w_j dz \wedge d\bar{w}_j
+ u \sum_{j=1}^2 z \bar{w}_j dw_j \wedge d\bar{z} \\
& + & (1+|z|^2)^2 \sum_{j,k} \bar{w}_j w_k dw_j \wedge d \bar{w}_k, \\
\pd \dbar u & = & |w|^2dz \wedge d \bar{z}
+ \sum_{j=1}^2 z\bar{w}_j dw_j \wedge d\bar{z}
+ \sum_{j=1}^2 \bar{z} w_j dz \wedge d \bar{w}_j  \\
& + & \sum_{j=1}^2 (1+|z|^2) dw_j \wedge d \bar{w}_j,
\een
one has:
\be
\begin{split}
\omega_a = & \sqrt{-1} \biggl\{ \biggl(\frac{a}{(1+|z |^2)^2}
+ \phi'(u) \cdot |w |^2 + \phi''(u) \cdot |w |^4|z|^2\biggr) dz  \wedge d\bar{z}  \\
+ & \sum_{j=1}^2 \biggl( \phi'(u) + u \phi''(u) \biggr)
\biggl( \bar{z} w_j  \cdot dz  \wedge d\bar{w}_j
+ z \bar{w}_j  \cdot dw_j  \wedge d \bar{z}  \biggr) \\
+ & \sum_{j,k=1}^2\biggl(\phi'(u) \cdot (1+|z |^2)\delta_{jk}
 + \phi''(u) \cdot (1+|z |^2)^2\bar{w}_jw_k\biggr)
dw_j  \wedge d \bar{w}_k \biggr\}.
\end{split}
\ee

\subsection{Symplectic coordinates}

Similar to \eqref{eqn:Kahler-symplectic} we have:

\begin{prop}
Let $z = r_0 e^{\sqrt{-1} \theta_0}$, $w_i = r_i e^{\sqrt{-1} \theta_i}$,
then
\be
\omega_a = \sum_{j=0}^2 dy_j \wedge d \theta_j,
\ee
where $y_0, y_1, y_2$ are defined by:
\bea
&& y_0 = -\frac{a}{1+r_0^2} + r_0^2 (r_1^2+r_2^2) \cdot \phi'((1+r_0^2)(r_1^2+r_2^2)),
\label{eqn:y0} \\
&&
y_j 
= r_j^2 (1+r_0^2) \cdot \phi'((1+r_0^2)(r_1^2+r_2^2)),
\;\;\; j=1,2. \label{eqn:yj}
\eea
\end{prop}

\begin{proof}
We first check that:
\ben
\omega_a & = & \biggl[\frac{2a r_0 dr_0 }{(1+r_0^2)^2}
+  2 \phi'(u) \cdot \biggl( (r_1^2+r_2^2) r_0 dr_0
+ \sum_{j=1}^2 r_0^2 r_j dr_j  \biggr) \\
& + & 2 \phi''(u) \cdot \biggl( (r_1^2+r_2^2)^2 r_0^3 dr_0
+ (1+r_0^2)(r_1^2+r_2^2) r_0^2 \sum_{j=1}^2 r_j dr_j \biggr)  \biggr] \wedge d \theta_0 \\
& + & \sum_{j=1}^2  \biggl[ 2 \phi'(u) \cdot \biggl(
r_j^2 r_0 dr_0 + (1+r_0^2) r_j dr_j \biggr) \\
& + & 2 \phi''(u) \cdot \biggl(
(1+r_0^2)(r_1^2+r_2^2) r_j^2r_0 dr_0
+ (1+r_0^2)^2\sum_{k=1}^2 r_j^2 r_k d r_k \biggr) \biggr]
\wedge d \theta_j.
\een
From this the Proposition can be easily checked.
\end{proof}

Define a $T^3$-action by
\ben
&& (t_0, t_1, t_2) \cdot (z, w_1, w_2) = (t_0 z, t_1w_1, t_2w_2), \\
&& (t_0, t_1, t_2) \cdot (\tilde{z}, \tilde{w}_1, \tilde{w}_2)
= (t_0^{-1} \tilde{z}, t_0t_1\tilde{w}_1, t_0t_2\tilde{w}_2).
\een
They are generated by the following holomorphic vector fields:
\ben
&& X_0 = z \frac{\pd}{\pd z}
= -\tilde{z} \frac{\pd}{\pd \tilde{z}}
+ w_1 \frac{\pd}{\pd w_1} + w_2 \frac{\pd}{\pd w_2}, \\
&& X_1 = w_1 \frac{\pd}{\pd w_1}
= \tilde{w}_1 \frac{\pd}{\pd \tilde{w}_1}, \\
&& X_2 =  w_2 \frac{\pd}{\pd w_2}
= \tilde{w}_2 \frac{\pd}{\pd \tilde{w}_2}.
\een
One can check that
\be
i_{X_j} \omega_a = \sqrt{-1} \dbar  y_j,
\ee
for $j=0,1,2$.

\subsection{K\"ahler potential in symplectic coordinates}

Let
\be
y:= y_1 + y_2.
\ee
By \eqref{eqn:yj} it is clear that:
\be \label{eqn:y-3}
y = u \cdot \phi'(u).
\ee
Then one can generalize \S \ref{sec:Kahler-potential}
almost verbatim.

\subsection{Riemannian metric in symplectic coordinates}

The Riemannian metric $g_a$ associated with $\omega_a$
can be explicitly written down as follows:
\ben
g_a & = & \frac{2a (dr_0^2 +r_0^2d \theta_0^2) }{(1+r_0^2)^2}
+  2 \phi'(u) \cdot \biggl( (r_1^2+r_2^2) (dr_0^2+r_0^2d\theta_0^2) \\
& + & 2 \sum_{j=1}^2 r_0 r_j (dr_0dr_j + r_0d\theta_0 r_j d\theta_j)
+ \sum_{j=1}^2 (1+r_0^2) (dr_j^2 + r_j^2 d \theta_j^2) \biggr) \\
& + & 2 \phi''(u) \cdot \biggl( (r_1^2+r_2^2)^2 r_0^2 (dr_0^2+ r_0^2 d \theta_0^2) \\
& + & 2 (1+r_0^2)(r_1^2+r_2^2) \sum_{j=1}^2 r_0 r_j(dr_j dr_0 + r_jd\theta_j r_0 d \theta_0) \\
& + & (1+r_0^2)^2\sum_{j,k=1}^2 r_j r_k
(dr_j d r_k + r_jr_k d \theta_j d \theta_k) \biggr).
\een

\begin{theorem}
In symplectic coordinates the Riemannian metric $g_a$ takes the following form:
\be
g_a = \sum_{i,j=1}^n (\frac{1}{2} G_{ij} dy_i dy_j +2 G^{ij}d\theta_id\theta_j),
\ee
where the coefficients $G_{ij}$ and $G^{ij}$ will be given in the proof.
Furthermore, the matrices $(G_{ij})$ and $(G^{ij})$ are inverse to each other.
\end{theorem}

\begin{proof}
We separate the terms with $dr_i$'s and those with $d\theta_j$'s
in the above expression for $g_a$ to get:
\be
g_a = g_a^r + g_a^\theta,
\ee
where $g_a^r$ and $g_a^\theta$ are given by:
\ben
g_a^r & = & \frac{2a dr_0^2}{(1+r_0^2)^2}
+  2 \phi'(u) \cdot \biggl( (r_1^2+r_2^2) dr_0^2
+ 2 \sum_{j=1}^2 r_0 r_j dr_0dr_j
+ \sum_{j=1}^2 (1+r_0^2) dr_j^2 \biggr) \\
& + & 2 \phi''(u) \cdot \biggl( (r_1^2+r_2^2) r_0 dr_0
+ (1+r_0^2) \sum_{j=1}^2 r_j dr_j \biggr)^2, \\
g_a^\theta & = & \frac{2a r_0^2d \theta_0^2 }{(1+r_0^2)^2}
+  2 \phi'(u) \cdot \biggl( (r_1^2+r_2^2)r_0^2d\theta_0^2
+ 2 \sum_{j=1}^2 r_0^2d\theta_0 r_j^2 d\theta_j
+ \sum_{j=1}^2 (1+r_0^2) r_j^2 d \theta_j^2 \biggr) \\
& + & 2 \phi''(u) \cdot \biggl( (r_1^2+r_2^2)r_0^2 d \theta_0
+ (1+r_0^2)\sum_{j=1}^2 r_j^2d\theta_j  \biggr)^2.
\een
From \eqref{eqn:y0} and \eqref{eqn:yj} one can get:
\bea
&& r_0 = \sqrt{\frac{y_0+a}{y_1+y_2-y_0}}, \\
&& r_j = \sqrt{\frac{y_j (y_1+y_2-y_0)}{\phi'(u)\cdot (y_1+y_2+a)}}, \;\;\;\; j=1,2.
\eea
With these identities one can compute $g_a^r$ and $g_a^\theta$.
For $g_a^\theta$ we first get:
\ben
g_a^\theta & = & 2 \biggl(\frac{a (y-y_0)(y_0+a)}{(y+a)^2}
+ \frac{(y_0+a)y}{y+a}
+ \frac{\phi''(u)}{\phi'(u)^2} \cdot \biggl( \frac{(y_0+a)y}{y+a}
\biggr)^2 \biggr) d \theta_0^2 \\
& + &  4 \sum_{j=1}^2 \biggl( \frac{(y_0+a)y_j}{y+a}
+ \frac{\phi''(u)}{\phi'(u)^2} \cdot  \frac{(y_0+a)yy_j}{y+a}
\biggr)d\theta_0 d\theta_j   \\
& + & 2 \sum_{j=1}^2 y_j d \theta_j^2
+ 2 \frac{\phi''(u)}{\phi'(u)^2} \cdot \biggl( \sum_{j=1}^2 y_j d\theta_j  \biggr)^2.
\een
We also have
\bea
&& \phi'(u) = \frac{y}{u}, \\
&& \frac{\phi''(u)}{\phi'(u)^2} = \frac{u}{y^2\frac{du}{dy}} - \frac{1}{y},
\eea
so we get:
\ben
g_a^\theta & = & 2 \biggl(\frac{(y-y_0)(y_0+a)}{y+a}
+ \frac{(y_0+a)^2}{(y+a)^2} \frac{u}{\frac{du}{dy}}\biggr) d \theta_0^2
+ 4 \sum_{j=1}^2 \frac{(y_0+a)y_j}{y(y+a)} \frac{u}{\frac{du}{dy}}
d\theta_0 d\theta_j   \\
& + & 2 \sum_{j,k=1}^2 (-1)^{j+k} \frac{y_1y_2}{y}  d \theta_jd\theta_k
+ 2 \frac{u}{y^2\frac{du}{dy}} \cdot \biggl( \sum_{j=1}^2 y_j d\theta_j  \biggr)^2 \\
& = & 2 \sum_{j,k} G^{jk} d\theta_jd\theta_k.
\een
Denote by $(G_{jk})$ the inverse matrix of $(G^{jk})$.
By a computation we have found the following
remarkably simple expressions for $G_{jk}$:
\ben
&& G_{00} = \frac{1}{y_0+a} +\frac{1}{y-y_0}
= \frac{y+a}{(y_0+a)(y-y_0)}, \\
&& G_{01} = G_{10} = G_{02} = G_{20} = - \frac{1}{y-y_0}, \\
&& G_{11} = \frac{\frac{du}{dy}}{u}
+ \frac{1}{y-y_0}-\frac{1}{y} -\frac{1}{y+a}+\frac{1}{y_1}, \\
&& G_{22} = \frac{\frac{du}{dy}}{u}
+ \frac{1}{y-y_0}-\frac{1}{y} -\frac{1}{y+a}+\frac{1}{y_2}, \\
&& G_{12} = G_{21} =\frac{\frac{du}{dy}}{u}
+ \frac{1}{y-y_0}-\frac{1}{y}-\frac{1}{y+a}.
\een
It is clear that
\be
g+a^r = 2 \sum_{i,j=0}^2 G^{ij} d\log r_i \cdot d \log r_j.
\ee
By a long and tedious computation we
check that:
\be
g_a^r = \half \sum_{j,k=0}^2 G_{jk} dy_j dy_k.
\ee
This completes the proof.
\end{proof}

\subsection{Hessian geometry}

Now as in \S \ref{sec:Hessian}
we are in the regime of Hessian geometry
and many results in that Section can be applied.
It is easy to see that
\be
G_{ij} = \frac{\pd^2\psi}{\pd y_i\pd y_j}
\ee
for the function $\psi$ defined by:
\be \label{def:Psi-3}
\begin{split}
\psi = & (y-y_0) (\ln (y-y_0) -1) + (y_0+a) (\ln (y_0+a)-1)
+ \sum_{i=1}^2 y_i (\ln y_i - 1) \\
- & y (\ln y -1) - (y+a) (\ln (y+a)-1) +  \int \log u(y) dy.
\end{split}
\ee
One can check that for $j=0,1,2$,
\ben
\frac{dw_j}{w_j}
& = &\frac{dr_j}{r_j} + \sqrt{-1} d\theta_j
= \half \sum_{k=0}^2 \frac{\pd^2\psi}{\pd y_j \pd y_k} dy_k
+ \sqrt{-1} d\theta_j,
\een
where $w_0 = z$.

The dual local coordinates $y_i^\vee$ are given by
\bea
&& y_0^\vee = \ln(y_0+a)-\ln(y-y_0) = 2\ln r_0, \\
&& y_i^\vee
= \ln(y-y_0)+\ln(y_i)-\ln(y)-\ln(y+a)+\ln(u(y))
= 2 \log r_i,
\eea
for $i=1,2$, and dual potential function:
\be
\psi^\vee = y \ln u(y)  -  \int \ln u(y) dy  - a \ln \frac{y_0+a}{y+a}
= \phi - a \ln \frac{r_0^2}{1+r_0^2}.
\ee
And the proof of Theorem \ref{thm:Metric2} works
in this case too.

\subsection{Kepler metric on $K_3$ lifted to $\cO_{\bP^1}(-1) \oplus \cO_{\bP^1}(-1)$}
By \eqref{eqn:Rawnsley} and Proposition \ref{prop:Norm},
the Kepler metric on $K_3$ is of the type \eqref{eqn:CA-Conifold} with
$a=0$ and $\phi(u) = u^{1/2}$.
Since $y = \half u^{1/2}$, we have $u = \frac{1}{4}y^2$.
In this case the moment map is given by:
\ben
&& y_0 = \frac{r_0^2 (r_1^2+r_2^2)}{2\sqrt{(1+r_0^2)(r_1^2+r_2^2)}},
 \\
&&
y_j 
= \frac{r_j^2 (1+r_0^2)}{2\sqrt{(1+r_0^2)(r_1^2+r_2^2)}},
\;\;\; j=1,2.
\een
It has the following inverse map:
\ben
&& r_0 = \sqrt{\frac{y_0}{y_1+y_2-y_0}}, \\
&& r_j = 2\sqrt{y_j (y_1+y_2-y_0)}, \;\;\;\; j=1,2.
\een
The image of the moment map is the convex cone
generated by $(1,1,0)$, $(1,0,1)$, $(0,1,0)$, $(0,0,1)$.
This cone is bounded by four planes:
\be
y_j = 0, \;\;\; j=0,1,2, \;\;\;\; \;\;\;
y_1 + y_2 - y_0 = 0.
\ee
One can check that
\be
\psi = (y_1+y_2-y_0) (\ln (y_1+y_2-y_0) -1) + y_0 (\ln y_0-1)
+ \sum_{i=1}^2 y_i (\ln y_i - 1).
\ee
The dual local coordinates $y_i^\vee$ are given by
\bea
&& y_0^\vee = \ln(y_0)-\ln(y_1+y_2-y_0), \\
&& y_i^\vee
= \ln(y_1+y_2-y_0)+\ln(y_i),\;\;\;\;\;i=1,2,
\eea
and the dual potential function is:
\be
\psi^\vee = \phi =2y.
\ee

\subsection{K\"ahler Ricci-flat metric on $K_3$}

In this case, $a=0$, $\phi(u) = \frac{3}{2} u^{2/3}$.
Since $y =   u^{2/3}$, we have $u = y^{3/2}$.
In this case the moment map is given by:
\ben
&& y_0 = \frac{r_0^2 (r_1^2+r_2^2)}{\sqrt[3]{(1+r_0^2)(r_1^2+r_2^2)}}, \\
&&
y_j  = \frac{r_j^2 (1+r_0^2)}{\sqrt[3]{(1+r_0^2)(r_1^2+r_2^2)}},
\;\;\; j=1,2.
\een
It has the following inverse map:
\ben
&& r_0 = \sqrt{\frac{y_0}{y_1+y_2-y_0}}, \\
&& r_j = \sqrt{\frac{y_j (y_1+y_2-y_0)}{\sqrt{(y_1+y_2)}}}, \;\;\;\; j=1,2.
\een
The image of the moment map is the convex cone
generated by $(y_0,y_1,y_2)= (1,1,0)$, 
$(1,0,1)$, $(0,1,0)$, $(0,0,1)$.
This cone is bounded by four planes:
\be
y_j \geq 0, \;\;\; j=0,1,2, \;\;\;\; \;\;\;
y_1 + y_2 - y_0 \geq 0.
\ee
One can check that
\be
\begin{split}
\psi = & (y_1+y_2-y_0) (\ln (y_1+y_2-y_0) -1) + y_0 (\ln (y_0)-1)
+ \sum_{i=1}^2 y_i (\ln y_i - 1) \\
- & \half (y_1+y_2) (\ln (y_1+y_2) -1).
\end{split}
\ee\emph{}
The dual local coordinates $y_i^\vee$ are given by
\bea
&& y_0^\vee = \ln(y_0)-\ln(y_1+y_2-y_0), \\
&& y_i^\vee
= \ln(y_1+y_2-y_0)+\ln(y_i)-\half \ln (y_1+y_2),\;\;\;\;\;i=1,2,
\eea
and dual potential function:
\be
\psi^\vee = \phi =\frac{3}{2} y.
\ee

\subsection{K\"ahler Ricci-flat metric on the resolved $3$-conifold}

\label{sec:Resolved}

Write $\omega_a = \sqrt{-1} \sum_{j,k=0}^2 \omega_{j\bar{k}} dw_j \wedge d \bar{w}_k$,
where $w_0 = z$.
Let $\Phi = \det (\omega_{j\bar{k}})_{j,k=0..2}$.
It is easy to find:
\ben
\Phi
& = & \frac{(a + u \phi'(u))(u\phi'(u))'}{1+|z|^2}
 \cdot \phi'(u) \cdot (1+|z|^2) \\
& = & (a+y) y' y (1+|z|^2)^{-1}|w|^{-2}.
\een
It is easy to see that
\ben
\pd\dbar \log \Phi & = & -\frac{dz \wedge d\bar{z}}{(1+|z|^2)^2}
- (n-1) \biggl( \frac{dw_j\wedge d \bar{w}_j}{|w|^2}
- \frac{\sum_{j,k} \bar{w}_jw_k dw_j \wedge d \bar{w}_k}{|w|^4} \biggr) \\
& + & (\log((a+y)y'y^{n-1}))'\cdot \pd \dbar u
+ (\log((a+y)y'y^{n-1}))''\cdot \pd u \wedge \dbar u.
\een
when $\rho = -\sqrt{-1} \pd\dbar \log \Phi = 0$,
the coefficients of $v$ all vanish:
\ben
&&  \frac{-1}{(1+|z|^2)^2} + (\log ((a+y)y'y))' \cdot |w|^2
+ (\log ((a+y)y'y))'' \cdot |w|^4 |z|^2 = 0, \\
&& (\log ((a+y)y'y))'   + (\log ((a+y)y'))'' \cdot u =0, \\
&& (\log ((a+y)y'y))' (1+|z|^2) \delta_{jk}
+ (\log ((a+y)y'y))'' \cdot (1+|z|^2)^2 \bar{w_j} w_k \\
& = & \frac{\delta_{jk}}{|w|^2} - \frac{\bar{w_j}w_k}{|w|^4}.
\een
In the third equation we take $j \neq k$,
then we get:
\be
(\log ((a+y)y'y))'' = -\frac{1}{u^2}.
\ee
and so from the second equation
\be
(\log ((a+y)y'y))'  = \frac{1}{u},
\ee
It follows that
\be \label{eqn:Resolved-Con1}
(a+y)yy' = cu,
\ee
After integration:
\be
\frac{a}{2} y^2 + \frac{y^3}{3} = \frac{c}{2} u^2 + c_1.
\ee
Take $c=\frac{2}{3}$ and $c_1 = 0$,
and one can find by solving a cubic equation:
\be \label{eqn:Resolved-Con}
\frac{3a}{2} y^2 + y^3 = u^2 .
\ee
Since $a>0$,
by the Inverse Function Theorem,
when $y > 0$,
$y$ is a smooth function in $u$.
Let us now analyse the behavior of $y$ and $y'(u)$ as $u \to 0$.
Clearly from \eqref{eqn:Resolved-Con},
\be
y(u) \sim \sqrt{\frac{2}{3a}} u -\frac{2u^2}{9a^2},
\ee
and so from \eqref{eqn:Resolved-Con1},
\be
y'(u) \sim \sqrt{\frac{2}{3a}}-\frac{4u}{9a^2}.
\ee
It follows that
\begin{align}
\phi'(u) & \sim \sqrt{\frac{2}{3a}}-\frac{2u}{9a^2}, & 
\phi''(u) & \sim -\frac{2}{9a^2}.
\end{align}
Therefore, as $u \to 0$,
\be
\begin{split}
\omega_a \sim & \sqrt{-1} \biggl\{ \biggl(\frac{a}{(1+|z |^2)^2}
+ \sqrt{\frac{2}{3a}} \cdot |w |^2 
-\frac{2}{9a^2} \cdot |w |^4|z|^2\biggr) dz  \wedge d\bar{z}  \\
+ & \sum_{j=1}^2 \sqrt{\frac{2}{3a}}
\biggl( \bar{z} w_j  \cdot dz  \wedge d\bar{w}_j
+ z \bar{w}_j  \cdot dw_j  \wedge d \bar{z}  \biggr) \\
+ & \sum_{j,k=1}^2\biggl(\sqrt{\frac{2}{3a}} \cdot (1+|z |^2)\delta_{jk}
-\frac{2}{9a^2} \cdot (1+|z |^2)^2\bar{w}_jw_k\biggr)
dw_j  \wedge d \bar{w}_k \biggr\}.
\end{split}
\ee
Clearly it can be extended smoothly over the zero section 
given by $w=0$.

It is actually possible in this case to find explicit expression
of $\phi$ in terms of $y$.
By \eqref{eqn:Resolved-Con},
\be \label{eqn:u-in-y-con}
u = \biggl(\frac{3a}{2} y^2 + y^3 \biggr)^{1/2}
= y \sqrt{y+\frac{3a}{2}},
\ee
therefore,
\be \label{eqn:phi-in-y-con}
\phi = y \ln (u(y)) - \int \ln (u(y)) dy
= = \frac{3y}{2} - \frac{3a}{4} \ln (y+ \frac{3a}{2}) +C.
\ee
This can be also derived as follows.
We get from \eqref{eqn:u-in-y-con} the following two identities:
\bea
&& \frac{du}{dy} = \sqrt{y+\frac{3a}{2}} + \frac{y}{2\sqrt{y+\frac{3a}{2}}}, \\
&& \frac{d\phi}{du} = \phi'(u) = \frac{y}{u} = \frac{1}{\sqrt{y + \frac{3a}{2}}}.
\eea
It follows that
\be
\frac{d\phi}{dy} = \frac{d\phi}{du} \frac{du}{dy}
= 1+ \frac{y}{2(y+\frac{3a}{2})}.
\ee
Therefore,
\eqref{eqn:phi-in-y-con} can be derived after integration.

By \eqref{eqn:y0} and \eqref{eqn:yj},
the moment map is:
\ben
&& y_0 = -\frac{a}{1+r_0^2} + \frac{r_0^2}{1+r_0^2} y, \\
&& y_j = \frac{r_j^2}{r_1^2+r_2^2} y,
\;\;\; j=1,2. 
\een
We have:
\ben
&& r_0 = \sqrt{\frac{y_0+a}{y_1+y_2-y_0}}, \\
&& r_j = \sqrt{\frac{y_j (y_1+y_2-y_0)\sqrt{y + \frac{3a}{2}}}{(y_1+y_2+a)}}, \;\;\;\; j=1,2.
\een
The image of the moment map is
the convex domain  bounded by four planes:
\be
y_j \geq 0, \;\;\; j=1,2, \;\;\;\; \;\;\; y_0 \geq -a, 
\;\;\;\;\;\;\;
y_1 + y_2 - y_0 \geq 0.
\ee
This convex domain has five edges,
the four rays
$\bR_{\geq 0}(1,1,0)$,
$\bR_{\geq 0}(1,0,1)$, $\bR_{\geq 0}(-a,1,0)$, 
$\bR_{\geq 0}(-a,0,1)$,
and an interval $\{(-at,0,0)\;|\; 1 \leq t \leq 1\}$.
By \eqref{def:Psi-3},
\ben
\psi & = & (y-y_0) (\ln (y-y_0) -1) + (y_0+a) (\ln (y_0+a)-1)
+ \sum_{i=1}^2 y_i (\ln y_i - 1) \\
& - & (y+a) (\ln (y+a)-1) 
+  \half (y+\frac{3a}{2}) (\ln (y+\frac{3a}{2}) - 1).
\een
The dual local coordinates $y_i^\vee$ are given by
\bea
&& y_0^\vee = \ln(y_0+a)-\ln(y-y_0) = 2\ln r_0, \\
&& y_i^\vee
= \ln(y-y_0)+\ln(y_i)-\ln(y+a)-\half\ln(y+\frac{3a}{2})
= 2 \log r_i,
\eea
for $i=1,2$, and one can check that:
\be
\psi^\vee  = \phi - a \ln \frac{r_0^2}{1+r_0^2}.
\ee

\section{Kepler Metric on 3-Conifold and Resolved 3-Conifold by Calabi Ansatz
on $\cO_{\bP^1 \times \bP^1}(-1,-1)$}

\label{sec:CA-2}

In this Section we repeat the steps in last Section
for another way to understand $K_3$.
It turns out to give us some natural way to understand the flop transformation
of the resolved conifold,
and it also leads to K\"ahler Ricci-flat metrics
on $\cO_{\bP^1\times \bP^1}(-2, -2) = \kappa_{\bP^1 \times \bP^1}$, 
the total space of the canonical line bundle of $\bP^1 \times \bP^1$.

\subsection{Calabi Ansatz on $\cO_{\bP^1 \times \bP^1}(-1,-1)$}

We will work with the local coordinates $(z_1, z_2, w)$
on $\cO_{\bP^1 \times \bP^1}(-1,-1)$.
They are related to other local coordinates,
denoted by $(\tilde{z}_1, \tilde{z}_2, \tilde{w})$,
$(\hat{z}_1, \hat{z}_2, \hat{w})$,
$(\hat{\tilde{z}}_1, \hat{\tilde{z}}_2, \hat{\tilde{w}})$
in the following fashion:
\begin{align*}
\tilde{z}_1 & = \frac{1}{z_1}, &
\tilde{z}_2 & = z_2, & \tilde{w} & = z_1w, \\
\hat{z}_1 & = z_1, &
\hat{z}_2 & = \frac{1}{z_2},
& \hat{w} & = z_2w, \\
\hat{\tilde{z}}_1 & = \frac{1}{z_1}, &
\hat{\tilde{z}}_2 & = \frac{1}{z_2},
& \hat{\tilde{w}} & = z_1z_2w.
\end{align*}
Define an Hermitian metric by
\be
u = r^2=  |w|^2 (1+|z_1|^2)(1+|z_2|^2).
\ee
We consider K\"ahler metrics of the form:
\be \label{eqn:CA-Conifold2}
\begin{split}
\omega_{a_1,a_2} = & a_1 \pi_1^*\omega_{\bP^1}
+ a_2 \pi_2^*\omega_{\bP^1} + \sqrt{-1} \pd \dbar \phi(u) \\
= & \sqrt{-1} a_1 \pd \dbar \log (1+|z_1|^2)
+ \sqrt{-1} a_2 \pd \dbar \log (1+|z_2|^2)
+ \sqrt{-1} \pd \dbar \phi(u).
\end{split}
\ee

\subsection{Symplectic coordinates}
One finds
\ben
\pd u & = & |w|^2(1+|z_2|^2) \bar{z}_1 dz_1
+ |w|^2 (1+|z_1|^2) \bar{z}_2 dz_2
+ (1+|z_1|^2)(1+|z_2|^2)\bar{w} dw, \\
\dbar u & = & |w|^2(1+|z_2|^2) z_1 d\bar{z}_1
+ |w|^2 (1+|z_1|^2) z_2 d\bar{z}_2
+ (1+|z_1|^2)(1+|z_2|^2) w d\bar{w},
\een
\ben
\pd \dbar u & = & |w|^2(1+|z_2|^2)dz_1 \wedge d \bar{z}_1
+ |w|^2(1+|z_1|^2)dz_2 \wedge d \bar{z}_2
+ \frac{u}{|w|^2} d w \wedge d\bar{w} \\
& + & |w|^2 (\bar{z}_1z_2 dz_1 \wedge d \bar{z}_2
+ \bar{z}_2z_1 dz_2 \wedge d \bar{z}_1 ) \\
& + & (1+|z_2|^2) (\bar{w} z_1 dw \wedge d\bar{z}_1
+ w\bar{z}_1 dz_1 \wedge d \bar{w}) \\
& + & (1+|z_1|^2) (\bar{w} z_2 d w \wedge d\bar{z}_2
+ w \bar{z}_2 dz_2 \wedge d \bar{w}),
\een
\ben
\pd u \wedge \dbar u
& = & |w|^4(1+|z_2|^2)^2 |z_1|^2 dz_1 \wedge d\bar{z}_1
+ |w|^4 (1+|z_1|^2)^2 |z_2|^2 dz_2 \wedge d\bar{z}_2 \\
& + & |w|^2 u (\bar{z}_1z_2 dz_1 \wedge d\bar{z}_2
+ \bar{z}_2 z_1 d z_2 \wedge d \bar{z}_1) \\
& + & u(1+|z_2|^2) (\bar{z}_1w dz_1 \wedge d \bar{w}
+ z_1\bar{w} dw \wedge d\bar{z}_1) \\
& + & u(1+|z_1|^2) (\bar{z}_2w dz_2 \wedge d \bar{w}
+ z_2\bar{w} dw \wedge d\bar{z}_2)
+  \frac{u^2}{|w|^2} dw \wedge d \bar{w}.
\een
Therefore,
since we have
\ben
\omega_{a_1,a_2} & = &
\sqrt{-1} \frac{a_1 dz_1  \wedge d\bar{z}_1}{(1+|z_1|^2)^2}
+ \sqrt{-1} \frac{a_2 dz_2  \wedge d\bar{z}_2}{(1+|z_2|^2)^2} \\
& + & \sqrt{-1} \phi'(u) \cdot  \pd\dbar u
+ \sqrt{-1} \phi''(u) \cdot \pd u \wedge \dbar u,
\een

\ben
\omega_{a_1,a_2} & = & \sqrt{-1} \biggl\{
\sum_{j=1}^2 \frac{1}{(1+|z_j|^2)^2}
\biggl(a_j
+ \phi'(u) \cdot u (1+|z_j|^2) + \phi''(u) \cdot u^2 |z_j|^2\biggr)
dz_j  \wedge d\bar{z}_j  \\
& + & |w|^2 (\phi'(u) + u \phi''(u)) \cdot (\bar{z}_1z_2 dz_1 \wedge d \bar{z}_2
+ \bar{z}_2z_1 dz_2 \wedge d \bar{z}_1 ) \\
& + & (1+|z_2|^2) (\phi'(u) + u \phi''(u)) (\bar{w} z_1 dw \wedge d\bar{z}_1
+ w\bar{z}_1 dz_1 \wedge d \bar{w}) \\
& + & (1+|z_1|^2) (\phi'(u) + u \phi''(u)) (\bar{w} z_2 d w \wedge d\bar{z}_2
+ w \bar{z}_2 dz_2 \wedge d \bar{w}) \\
& + & \biggl(\phi'(u) \cdot \frac{u}{|w|^2}
 + \phi''(u) \cdot \frac{u^2}{|w|^2} \biggr)
dw \wedge d \bar{w} \biggr\}.
\een
Write $z_j = r_j e^{i\theta_j}$, $j=1,2$, $w = r_3e^{i\theta_3}$.
\ben
\omega_{a_1,a_2} & = & 2 \biggl\{
\biggl( \frac{a_1}{(1+r_1^2)^2}
   + r_3^2(1+r_2^2) \phi'(u) +r_1^2(1+r_2^2)^2r_3^4\phi''(u)\biggr)
r_1dr_1  \wedge d\theta_1  \\
& + & \biggl( \frac{a_2}{(1+r_2^2)^2}
   + r_3^2(1+r_1^2) \phi'(u) +r_2^2(1+r_1^2)^2r_3^4\phi''(u)\biggr)
r_2dr_2  \wedge d\theta_2  \\
& + & r_3^2 (\phi'(u) + u \phi''(u)) \cdot r_1r_2(dr_1 \wedge r_2d \theta_2
+ dr_2 \wedge r_1d \theta_1 ) \\
& + & (1+r_2^2) (\phi'(u) + u \phi''(u)) \cdot r_1r_3 (dr_1 \wedge r_3 d\theta_3
+ dr_3 \wedge r_1 d \theta_1) \\
& + & (1+r_1^2) (\phi'(u) + u \phi''(u)) \cdot r_2r_3 (dr_2 \wedge r_3 d\theta_3
+ dr_3 \wedge r_2 d \theta_2) \\
& + & (\phi'(u) + u \phi''(u)) \cdot (1+r_1^2)(1+r_2^2) \cdot r_3 dr_3 \wedge d \theta_3 \biggr\}.
\een
From this one can easily check that
\be
\omega = dy_1 \wedge d\theta_1 + dy_2 \wedge d\theta_2 + dy_3 \wedge d \theta_3,
\ee
where $y_1, y_2, y_3$ are defined by:
\bea
&& y_1 = -\frac{a_1}{1+r_1^2} + r_1^2r_3^2 (1+r_2^2) \cdot \phi'((1+r_1^2)(1+r_2^2)r_3^2),
\label{eqn:y1} \\
&& y_2 = -\frac{a_2}{1+r_2^2} + r_2^2r_3^2 (1+r_1^2) \cdot \phi'((1+r_1^2)(1+r_2^2)r_3^2),
\label{eqn:y2} \\
&& y_3 = (1+r_1^2)(1+r_2^2)r_3^2\cdot \phi'((1+r_1^2)(1+r_2^2)r_3^2). \label{eqn:y3}
\eea
These functions generate $T^3$-action on $\cO_{\bP^1\times \bP^1}(-1,-1)$ defined
in local coordinates $(z_1, z_2, z_3)$ by:
\be
(e^{i\alpha_1}, e^{i\alpha_2}, e^{i\alpha_3})\cdot
(z_1, z_2, w)
= (e^{i\alpha_1}z_1, e^{i\alpha_2}z_2, e^{i\alpha_3}w).
\ee

Let $y:= y_3$.
Then
\be
y = u \cdot \phi'(u)
\ee
and one can express $\phi$ as a function in $y$.

\subsection{Riemannian metric in symplectic coordinates and Hessian geometry}

The Riemannian metric $g_a$ associated with $\omega_a$
can be explicitly written down as follows:
\ben
g_{a_1,a_2} & = & 2 \biggl\{
\biggl( \frac{a_1}{(1+r_1^2)^2}
   + r_3^2(1+r_2^2) \phi'(u) +r_1^2(1+r_2^2)^2r_3^4\phi''(u)\biggr)
(dr_1^2 +r_1^2 d\theta_1^2)  \\
& + & \biggl( \frac{a_2}{(1+r_2^2)^2}
   + r_3^2(1+r_1^2) \phi'(u) +r_2^2(1+r_1^2)^2r_3^4\phi''(u)\biggr)
 (dr_2^2+ r_2^2 d\theta_2^2)  \\
& + & 2 r_3^2 (\phi'(u) + u \phi''(u)) \cdot r_1r_2(dr_1 dr_2 +
+  r_1d \theta_1 r_2d \theta_2) \\
& + & 2(1+r_2^2) (\phi'(u) + u \phi''(u)) \cdot r_1r_3 (dr_1dr_3  +
 r_1 d \theta_1 r_3 d\theta_3 ) \\
& + & 2 (1+r_1^2) (\phi'(u) + u \phi''(u)) \cdot r_2r_3
(dr_2 dr_3 + r_2 d \theta_2 r_3 d\theta_3) \\
& + & (\phi'(u) + u \phi''(u)) \cdot (1+r_1^2)(1+r_2^2)
\cdot (dr_3^2 + r_3 ^2d \theta_3^2) \biggr\}.
\een

\begin{theorem}
In symplectic coordinates the Riemannian metric $g$ takes the following form:
\be
g_{a_1,a_2} = \sum_{i,j=1}^n (\frac{1}{2} G_{ij} dy_i dy_j +2 G^{ij}d\theta_id\theta_j),
\ee
where the coefficients $G_{ij}$ and $G^{ij}$ will be given in the proof.
Furthermore, the matrices $(G_{ij})$ and $(G^{ij})$ are inverse to each other.
\end{theorem}

\begin{proof}
We separate the terms with $dr_i$'s and those with $d\theta_j$'s
in the above expression for $g_a$ to get:
\be
g_{a_1,a_2} = g_{a_1,a_2}^r + g_{a_1,a_2}^\theta,
\ee
where $g_a^r$ and $g_a^\theta$ are given by:
\ben
g_{a_1,a_2}^r & = & 2 \biggl\{
\biggl( \frac{a_1}{(1+r_1^2)^2}
   + r_3^2(1+r_2^2) \phi'(u) +r_1^2(1+r_2^2)^2r_3^4\phi''(u)\biggr) dr_1^2   \\
& + & \biggl( \frac{a_2}{(1+r_2^2)^2}
   + r_3^2(1+r_1^2) \phi'(u) +r_2^2(1+r_1^2)^2r_3^4\phi''(u)\biggr) dr_2^2 \\
& + & 2 r_3^2 (\phi'(u) + u \phi''(u)) \cdot r_1r_2dr_1 dr_2 \\
& + & 2(1+r_2^2) (\phi'(u) + u \phi''(u)) \cdot r_1r_3 dr_1dr_3  \\
& + & 2 (1+r_1^2) (\phi'(u) + u \phi''(u)) \cdot r_2r_3 dr_2 dr_3  \\
& + & (\phi'(u) + u \phi''(u)) \cdot (1+r_1^2)(1+r_2^2)\cdot dr_3^2 \biggr\}.
\een
\ben
g_{a_1,a_2}^\theta & = & 2 \biggl\{
\biggl( \frac{a_1}{(1+r_1^2)^2}
   + r_3^2(1+r_2^2) \phi'(u) +r_1^2(1+r_2^2)^2r_3^4\phi''(u)\biggr) r_1^2 d\theta_1^2  \\
& + & \biggl( \frac{a_2}{(1+r_2^2)^2}
   + r_3^2(1+r_1^2) \phi'(u) +r_2^2(1+r_1^2)^2r_3^4\phi''(u)\biggr) r_2^2 d\theta_2^2 \\
& + & 2 r_3^2 (\phi'(u) + u \phi''(u)) \cdot r_1^2r_2^2  d \theta_1 d \theta_2 \\
& + & 2 (1+r_2^2) (\phi'(u) + u \phi''(u)) \cdot r_1^2r_3^2 d \theta_1 d\theta_3  \\
& + & 2 (1+r_1^2) (\phi'(u) + u \phi''(u)) \cdot r_2^2r_3^2 d \theta_2d\theta_3 \\
& + & (\phi'(u) + u \phi''(u)) \cdot (1+r_1^2)(1+r_2^2) r_3 ^2d \theta_3^2 \biggr\}.
\een
From \eqref{eqn:y1} to \eqref{eqn:y3} one can get:
\bea
&& r_j = \sqrt{\frac{y_j+a_j}{y-y_j}}, \;\;\;\; j=1,2, \\
&& r_3 = \sqrt{\frac{y (y-y_1)(y-y_2)}{\phi'(u)\cdot (y+a_1)(y+a_2)}}.
\eea
With these identities one can compute $g_a^r$ and $g_a^\theta$.
For $g_a^\theta$ we get:
\ben
g_{a_1,a_2}^\theta & = & 2\sum_{j=1}^2 \frac{(y-y_j)(y_j+a_j)}{y+a_j} d \theta_j^2
+2 \frac{u}{\frac{du}{dy}} \biggl( \sum_{j=1}^3 \frac{y_j+a_j}{y+a_j}
d \theta_j \biggr)^2.
\een

Denote by $(G_{jk})$ the inverse matrix of $(G^{jk})$.
By a computation we have found the following
remarkably simple expressions for $G_{jk}$:
\ben
&& G_{11} = \frac{1}{y_1+a_1} +\frac{1}{y-y_1}, \;\;\;\;\;
 G_{22} = \frac{1}{y_2+a_2} +\frac{1}{y-y_2}, \\
&& G_{12} = G_{21} = 0, \\
&& G_{13} = G_{31}  = - \frac{1}{y-y_1}, \;\;\;\;\;
 G_{23} = G_{32}  = - \frac{1}{y-y_2},  \\
&& G_{33} = \frac{\frac{du}{dy}}{u}
+ \frac{1}{y-y_1}-\frac{1}{y+a_1} +\frac{1}{y-y_2}-\frac{1}{y+a_2}.
\een
It is clear that
\ben
g_{a_1,a_2}^r = 2 \sum_{i,j=1}^3 G^{ij} \frac{dr_i}{r_i} \frac{dr_j}{r_j}.
\een
So we have:
\ben
g_{a_1,a_2}^r & = & 2\sum_{j=1}^2 \frac{(y-y_j)(y_j+a_j)}{y+a_j} (d \log r_j)^2
+ 2 \frac{u}{\frac{du}{dy}} \biggl( \sum_{j=1}^3 \frac{y_j+a_j}{y+a_j}
d \log r_j \biggr)^2.
\een
By a straightforward computation we
check that:
\be
g_{a_1,a_2}^r = \half \sum_{j,k=0}^2 G_{jk} dy_j dy_k.
\ee
This completes the proof.
\end{proof}

Again as in \S \ref{sec:Hessian}
we are in the regime of Hessian geometry.
It is easy to see that
\be
G_{ij} = \frac{\pd^2\psi}{\pd y_i\pd y_j}
\ee
for the function $\psi$ defined by:
\be \label{def:Psi-4}
\begin{split}
\psi = & \sum_{j=1}^2 \biggl( (y-y_j) (\ln (y-y_j) -1) + (y_j+a_j) (\ln (y_j+a_j)-1) \\
- & (y+a_j) (\ln (y+a_j)-1) \biggr) +  \int \log u(y) dy.
\end{split}
\ee
One can check that for $i=1,2,3$,
\ben
\frac{dz_i}{z_i}
& = & \half \sum_{j=1}^3 \frac{\pd^2\psi}{\pd y_i \pd y_j} dy_j
+ \sqrt{-1} d\theta_i,
\een
where $z_3 = w$.

The dual local coordinates $y_i^\vee$ are given by
\bea
&& y_j^\vee = \ln(y_j+a_j)-\ln(y-y_j) = 2\ln r_j, \;\;\;\; j=1,2, \\
&& y_3^\vee =
\sum_{j=1}^2 \biggl( \ln(y-y_j)-\ln(y+a_j) \biggr)+\ln(u(y))
= 2 \log r_3,
\eea
and dual potential function:
\be
\psi^\vee = y \ln u(y)  -  \int \ln u(y) dy  - \sum_{j=1}^2 a_j \ln \frac{y_j+a_j}{y+a_j}
= \phi - \sum_{j=1}^2 a_j \ln \frac{r_j^2}{1+r_j^2}.
\ee

\subsection{Riemannian metric in polar coordinates}
As in \S \ref{sec:Polar} we introduce polar coordinates,
one for each of $z_j$ as follows.
For convenience, we will repeat some of the computations there.
Let
\begin{align}
x_j^0 & = \frac{1-|z_j|^2}{|z_j|^2+1}, & x_j^1+\sqrt{-1} x_j^2 = \frac{2z}{|z_j|^2+1},
\end{align}
Then one has
\be
z_j= \frac{x_j^1+\sqrt{-1}x_j^2}{x_j^0+1}.
\ee
Let $\theta_j$ and $\phi_j$ be Euler angles such that
\begin{align}
x_j^0 & = \cos \theta_j, & x_j^1 & = \sin \theta_j \cos \phi_j, & 
x_j^2 & = \sin \theta_j \sin \phi_j.
\end{align}
Then one has
\be
z_j = e^{i \phi_j} \tan (\frac{\theta_j}{2}),
\ee
and so we have
\be
1+|z_j|^2= \frac{1}{\cos^2(\frac{\theta_j}{2})}
\ee
and
\be
dz_j = i e^{i\phi_j} \tan(\frac{\theta_j}{2}) d \phi_j
+\half \frac{e^{i\phi_j}}{\cos^2 (\frac{\theta_j}{2})} d\theta_j.
\ee
Therefore,
\be
\frac{|dz_j|^2}{(|z_j|^2+1)^2} = \frac{1}{4} (d\theta_j^2 + \sin^2\theta_j d\phi^2).
\ee
Now let
\be
w = r \cos (\frac{\theta_1}{2}) \cos (\frac{\theta_2}{2}) \cdot e^{i\beta}.
\ee
Its differential is given by:
\be
dw = r \cos (\frac{\theta_1}{2}) \cos (\frac{\theta_2}{2}) \cdot e^{i\beta}
(\frac{dr}{r} - \frac{1}{2} \tan(\frac{\theta_1}{2})d\theta_1
- \frac{1}{2} \tan(\frac{\theta_2}{2}) d\theta_2 + i d\beta).
\ee
It is easy to see that $u = r^2$.

The Riemannian metric is given by:
\ben
g_{a_1,a_2} & = &
\sum_{j=1}^2 \frac{1}{(1+|z_j|^2)^2}
\biggl(a_j + y(u) +|z_j|^2uy'(u)\biggr)
|dz_j|^2   \\
& + & 2|w|^2 y'(u) \cdot \Re (\bar{z}_1z_2 dz_1 d \bar{z}_2) \\
& + & 2(1+|z_2|^2) y'(u) \cdot \Re (\bar{w} z_1 d w  d\bar{z}_1) \\
& + & 2(1+|z_1|^2) y'(u) \cdot \Re (\bar{w} z_2 d w d\bar{z}_2) \\
& + & y'(u) \cdot \frac{u}{|w|^2} \cdot |dw|^2.
\een
It can be rewritten as:
\be
g_{a_1,a_2} =
\sum_{j=1}^2 \frac{a_j + y(u)}{(1+|z_j|^2)^2} |dz_j|^2
+ u y'(u) \cdot |\gamma|^2,
\ee
where
\be
\gamma= \frac{\bar{z}_1}{1+|z_1|^2} dz_1+ \frac{\bar{z}_2}{1+|z_2|^2} dz_2
+ \frac{\bar{w}}{|w|^2} d w.
\ee
It is easy to find
\be
\gamma = \frac{dr}{r} + i (d\beta + \sin^2(\frac{\theta_1}{2}) d\phi_1
+ \sin^2(\frac{\theta_2}{2}) d \phi_2),
\ee
so we get:
\ben
g_{a_1,a_2} & = &  y'(u) dr^2 +
\sum_{j=1}^2 \frac{1}{4} \biggl(a_j + y(u)\biggr)
(d\theta_j^2 + \sin^2\theta_j d\phi_j^2)  \\
& + & uy'(u) \cdot (d\beta + \sin^2(\frac{\theta_1}{2}) d\phi_1
+ \sin^2(\frac{\theta_2}{2}) d \phi_2)^2.
\een
Let $\psi = \theta_1+\theta_2-2\beta$,
\be
\begin{split}
g_{a_1,a_2}  = &  y'(u) dr^2 +
\sum_{j=1}^2 \frac{1}{4} \biggl(a_j + y(u)\biggr)
(d\theta_j^2 + \sin^2\theta_j d\phi_j^2)  \\
+ & \frac{1}{4}uy'(u) \cdot (d\psi + \cos(\theta_1) d\phi_1
+ \cos(\theta_2) d \phi_2)^2.
\end{split}
\ee

\subsection{Kepler metric on $K_3$ lifted to $\cO_{\bP^1 \times \bP^1}(-1,-1)$}

By \eqref{eqn:Rawnsley} and Proposition \ref{prop:Norm},
the Kepler metric on $K_3$ when lifted to $\cO_{\bP^1 \times \bP^1}(-1,-1)$
is of the type \eqref{eqn:CA-Conifold2} with
$a_1=a_2=0$ and $\phi(u) = u^{1/2}$.
Again we have $y = \half u^{1/2}$ and $u = \frac{1}{4}y^2$.
In this case the moment map is given by:
\ben
&& y_1 = \frac{r_1^2r_3}{2}\sqrt{\frac{1+r_2^2}{1+r_1^2}}, \;\;\;\;\;\;
y_2 = \frac{r_2^2r_3}{2}\sqrt{\frac{1+r_1^2}{1+r_2^2}}, \\
&& y_3 = \frac{r_3}{2} \sqrt{(1+r_1^2)(1+r_2^2)}.
\een
It has the following inverse map:
\ben
&& r_j = \sqrt{\frac{y_j}{y-y_j}}, \;\;\;\; j=1,2, \\
&& r_3 = 2\sqrt{(y-y_1) (y-y_2)}.
\een
The image of the moment map is the convex cone
 bounded by four planes:
\be
y_j \geq 0,  \;\;\;\; \;\;\; y \geq y_j,  \;\;\; j=1,2,
\ee
generated by $(1,1,0)$, $(1,0,1)$, $(1,1,1)$, $(0,0,1)$.

One can check that
\ben 
\psi = \sum_{j=1}^2 \biggl( (y-y_j) (\ln (y-y_j) -1) + y_j(\ln (y_j)-1)  \biggr).
\een
The dual local coordinates $y_i^\vee$ are given by
\ben
&& y_j^\vee = \ln(y_j+a_j)-\ln(y-y_j), \;\;\;\; j=1,2, \\
&& y_3^\vee = \sum_{j=1}^2 \ln(y-y_j),
\een
and dual potential function:
\ben
\psi^\vee = 2y = \phi.
\een

\subsection{K\"ahler Ricci-flat metrics}
Rewrite $\omega$ in the following form:
\ben
\omega_{a_1,a_2} & = & \sqrt{-1} \biggl\{
\sum_{j=1}^2 \frac{1}{(1+|z_j|^2)^2}
\biggl(a_j + y(u) +|z_j|^2uy'(u)\biggr)
dz_j  \wedge d\bar{z}_j  \\
& + & |w|^2 y'(u) \cdot (\bar{z}_1z_2 dz_1 \wedge d \bar{z}_2
+ \bar{z}_2z_1 dz_2 \wedge d \bar{z}_1 ) \\
& + & (1+|z_2|^2) y'(u) \cdot (\bar{w} z_1 dw \wedge d\bar{z}_1
+ w\bar{z}_1 dz_1 \wedge d \bar{w}) \\
& + & (1+|z_1|^2) y'(u) \cdot (\bar{w} z_2 d w \wedge d\bar{z}_2
+ w \bar{z}_2 dz_2 \wedge d \bar{w}) \\
& + & y'(u) \cdot \frac{u}{|w|^2} \cdot dw \wedge d \bar{w} \biggr\}.
\een
Write $\omega = \sqrt{-1} h_{a\bar{b}}dz_a d\bar{z}_b$, where $z_3 =w$,
and let $\Phi = \det (h_{a\bar{b}})$.
It is straightforward to find:
\ben
\Phi & = & \frac{(a_1+y) (a_2+y) y'}{(1+|z_1|^2)(1+|z_2|^2)}.
\een
The Ricci form is $\rho = - \sqrt{-1} \pd \dbar \Phi$, and since
\ben
\pd\dbar \log \Phi & = & -\frac{dz_1 \wedge d\bar{z}_1}{(1+|z_1|^2)^2}
-\frac{dz_2 \wedge d\bar{z}_2}{(1+|z_2|^2)^2}  \\
& + & (\log((a_1+y)(a_2+y)y'))'\cdot \pd \dbar u
+ (\log((a_1+y)(a_2+y)y'))''\cdot \pd u \wedge \dbar u,
\een
for all the coefficients to vanish we need:
\ben
&& -1 + \alpha'(u)  \cdot  u (1+|z_j|^2) + \alpha''(u) \cdot u^2 |z_j|^2  = 0,
\;\;\;\; j = 1, 2, \\
&& \alpha'(u)   + u\alpha''(u) =0,
\een
where $\alpha(u) = \log ((a_1+y)(a_2+y)y')$.
One easily finds that
\be
u \alpha'(u) = 1,
\ee
It follows that
\be
(a_1+y)(a_2+y) y'  = c u.
\ee
After integration:
\be
\frac{y^3}{3} + \frac{a_1+a_2}{2} y^2 + a_1a_2y -\frac{c}{2} u^2 - c_1 = 0.
\ee
Take $c=\frac{2}{3}$ and $c_1 = 0$,
one gets a cubic equation
\be \label{eqn:a1a2}
\frac{y^3}{3} + \frac{a_1+a_2}{2} y^2 + a_1a_2y -\frac{1}{3} u^2 = 0.
\ee

\subsubsection{The case of $a_1=a_2=0$}

In this case $y^3 = u^2$, $y = u^{2/3}$,
$\phi'(u) = u^{-1/3}= y^{-1/2}$.
The moment map is given by:
\ben
&& y_1 =  r_1^2r_3^2 (1+r_2^2) \cdot ((1+r_1^2)(1+r_2^2)r_3^2)^{-1/3},  \\
&& y_2 =  r_2^2r_3^2 (1+r_1^2) \cdot ((1+r_1^2)(1+r_2^2)r_3^2)^{-1/3}, \\
&& y_3 =  ((1+r_1^2)(1+r_2^2)r_3^2)^{2/3}.  
\een
We have
\ben
&& r_j = \sqrt{\frac{y_j}{y-y_j}}, \;\;\;\; j=1,2, \\
&& r_3 = \sqrt{\frac{(y-y_1)(y-y_2)}{y^{1/2} }}.
\een
The image of the moment map is the same as in the case of 
Kepler metric on $K_3$.
The complex potential function $\psi$ is defined by:
\ben 
\psi & = & \sum_{j=1}^2 \biggl( (y-y_j) (\ln (y-y_j) -1) 
+ y_j (\ln (y_j)-1) \biggr) - \half y (\ln (y)-1).
\een
The dual local coordinates $y_i^\vee$ are given by
\ben
&& y_j^\vee = \ln(y_j)-\ln(y-y_j) = 2\ln r_j, \;\;\;\; j=1,2, \\
&& y_3^\vee =
\sum_{j=1}^2 \ln(y-y_j)- \half \ln(y) 
= 2 \log r_3,
\een
and dual potential function:
\ben
&& \psi^\vee = \frac{3}{2} y    
= \phi .
\een

\subsubsection{The case of one of $a_1, a_2$ is nonzero}

Suppose that $a_1 = a$, $a_2 = 0$.
We blow down $\cO_{\bP^1\times \bP^1}(-1,-1)$ along
the copy of $\bP^1$ parameterized by $z_1$.
In local coordinates $(z_1, z_2, w)$,
we make the following change of variables:
\begin{align}
z & = z_1, & w_1 & = w, & w_2 = wz.
\end{align}
Since 
$$(|z|^2+1)(|w_1|^2+|w_2|^2) = (|z_1|^2+1)(|z_2|^2+1)|w|^2,$$
after blowing down we are in the situation of \S \ref{sec:Resolved}.

\subsubsection{The case of $a_1 >0$ and $a_2 > 0$}
By \eqref{eqn:a1a2},
\be 
u = \biggl(y^3 + \frac{3(a_1+a_2)}{2} y^2 + 3a_1a_2y\biggr)^{1/2}.
\ee
We will rewrite as
\be
u =(y(y-\beta_+)(y-\beta_-))^{1/2},
\ee
where $\beta_\pm$  are given by the root formula:
\be
\beta_\pm = - \frac{3(a_1+a_2)}{4} + \frac{1}{2}\sqrt{(3a_1-a_2)(a_1-3a_2)}.
\ee
Hence one finds the  potential functions $\phi$ and $\psi$ explicitly as
functions in $y$:
\be
\phi = \frac{3y}{2} + \half \beta_+ \ln(y-\beta_+) + \half \beta_- \ln(y-\beta_-),
\ee
\be 
\begin{split}
\psi = & \sum_{j=1}^2 \biggl( (y-y_j) (\ln (y-y_j) -1) + (y_j+a_j) (\ln (y_j+a_j)-1) \\
- & (y+a_j) (\ln (y+a_j)-1) \biggr) +  \half y (\ln(y)-1) \\
+ & \half (y-\beta_+) (\ln(y-\beta_+) -1)
+ \half (y-\beta_-) (\ln(y-\beta_-) -1).
\end{split}
\ee

By \eqref{eqn:y1}-\eqref{eqn:y3},
the moment map is given by:
\ben
&& y_1 = -\frac{a_1}{1+r_1^2} + \frac{r_1^2}{1+r_1^2} y,\\
&& y_2 = -\frac{a_2}{1+r_2^2} + \frac{r_2^2}{1+r_2^2} y, \\
&& y_3 = y,
\een
and we have
\ben
&& r_j = \sqrt{\frac{y_j+a_j}{y-y_j}}, \;\;\;\; j=1,2, \\
&& r_3 = \sqrt{\frac{\biggl(y^3 + \frac{3(a_1+a_2)}{2} y^2 + 3a_1a_2y\biggr)^{1/2} (y-y_1)(y-y_2)}{(y+a_1)(y+a_2)}}.
\een
The image of the moment map is the convex domain in $\bR^3$ bounded by:
\begin{align}
y_j & \geq -a_j, & y & \geq y_j, \;\; j=1,2, & y & \geq 0.
\end{align}
Examining the above formula for $\phi$ and $\psi$,
we see that they are not directly given by the boundary
of the convex domain.

By the Inverse Function Theorem,
when $y > 0$,
$y$ is a smooth function in $u$.
So far we have obtained smooth K\"ahler Ricci-flat metrics
on $\cO_{\bP^1\times \bP^1}(-1,-1)$ away from the zero section.
To examine the behavior of the metric along the zero section,
as in \S \ref{sec:Resolved},
we  analyse the behavior of $y$ and $y'(u)$ as $u \to 0$.
From \eqref{eqn:a1a2},
\be
y(u) \sim \frac{u^2}{3a_1a_2} 
-\frac{1}{18}\frac{a_1+a_2}{a_1^3a_2^3}u^4,
\ee
and so from \eqref{eqn:Resolved-Con1},
\be
y'(u) \sim \frac{2u}{3a_1a_2}
-\frac{2}{9}\frac{a_1+a_2}{a_1^3a_2^3}u^3.
\ee
Therefore, as $u \to 0$,
\ben
\omega_{a_1,a_2} & \sim & \sqrt{-1} \biggl\{
\sum_{j=1}^2 \frac{1}{(1+|z_j|^2)^2}
\biggl(a_j +  \frac{u^2}{3a_1a_2}  +|z_j|^2u\cdot \frac{2u}{3a_1a_2}\biggr)
dz_j  \wedge d\bar{z}_j  \\
& + & |w|^2 \frac{2u}{3a_1a_2} \cdot (\bar{z}_1z_2 dz_1 \wedge d \bar{z}_2
+ \bar{z}_2z_1 dz_2 \wedge d \bar{z}_1 ) \\
& + & (1+|z_2|^2)\frac{2u}{3a_1a_2} \cdot (\bar{w} z_1 dw \wedge d\bar{z}_1
+ w\bar{z}_1 dz_1 \wedge d \bar{w}) \\
& + & (1+|z_1|^2) \frac{2u}{3a_1a_2} \cdot (\bar{w} z_2 d w \wedge d\bar{z}_2
+ w \bar{z}_2 dz_2 \wedge d \bar{w}) \\
& + & \frac{2u}{3a_1a_2} \cdot \frac{u}{|w|^2} \cdot dw \wedge d \bar{w} \biggr\}.
\een

Along the zero section of $\cO_{\bP^1\times \bP^1}(-1,-1)$
given in local coordinates by $w = 0$,
$$\omega_{a_1,a_2} 
\sim  \sqrt{-1} \biggl\{
\sum_{j=1}^2 \frac{1}{(1+|z_j|^2)^2}
\biggl(a_j +  \frac{u^2}{3a_1a_2}  +|z_j|^2u\cdot \frac{2u}{3a_1a_2}\biggr)
dz_j  \wedge d\bar{z}_j  \biggr\}
$$ 
is {\em degenerate}.
To fix this problem,
recall there is a $2:1$ map from $\cO_{\bP^1\times \bP^1}(-1,-1)$
to $\cO_{\bP^1\times \bP^1}(-2,-2)$.
We will work with local coordinates $(z_1, z_2, v)$,
$(\tilde{z}_1, \tilde{z}_2, \tilde{v})$,
$(\hat{z}_1, \hat{z}_2, \hat{v})$,
$(\hat{\tilde{z}}_1, \hat{\tilde{z}}_2, \hat{\tilde{v}})$
on $\cO_{\bP^1 \times \bP^1}(-2,-2)$. 
They are related in the following way:
\begin{align*}
\tilde{z}_1 & = \frac{1}{z_1}, &
\tilde{z}_2 & = z_2, & \tilde{v} & = z_1^2v, \\
\hat{z}_1 & = z_1, &
\hat{z}_2 & = \frac{1}{z_2},
& \hat{v} & = z_2^2v, \\
\hat{\tilde{z}}_1 & = \frac{1}{z_1}, &
\hat{\tilde{z}}_2 & = \frac{1}{z_2},
& \hat{\tilde{v}} & = z_1^2z_2^2v.
\end{align*}
The $2:1$ map is given locally by:
\be
(z_1, z_2, w) \mapsto (z_1, z_2, v= w^2).
\ee
Since $w = \pm v^{1/2}$,
$\omega_{a_1, a_2}$ is the pullback of a 
K\"ahler form $\hat{\omega}_{a_1,a_2}$
on $\cO_{\bP^1\times \bP^1}(-2,-2)-\bP^1\times \bP^1$,
and as $v \to 0$,
we have
\ben
\hat{\omega}_{a_1,a_2} & \sim & \sqrt{-1} \biggl\{
\sum_{j=1}^2 \frac{1}{(1+|z_j|^2)^2}
\biggl(a_j +  \frac{|v|^2\prod_{k=1}^2(1+|z_k|^2)^2}{3a_1a_2}  \\
& + & |z_j|^2\cdot \frac{2|v|^2\prod_{k=1}^2(1+|z_k|^2)^2}{3a_1a_2}\biggr)
dz_j  \wedge d\bar{z}_j  \\
& + & \frac{2|v|^2\prod_{k=1}^2(1+|z_k|^2)}{3a_1a_2} \cdot (\bar{z}_1z_2 dz_1 \wedge d \bar{z}_2
+ \bar{z}_2z_1 dz_2 \wedge d \bar{z}_1 ) \\
& + & (1+|z_2|^2)\frac{\prod_{k=1}^2(1+|z_k|^2)}{3a_1a_2} \cdot 
(\bar{v} z_1 dv \wedge d\bar{z}_1
+ v\bar{z}_1 dz_1 \wedge d \bar{v}) \\
& + & (1+|z_1|^2) \frac{\prod_{k=1}^2(1+|z_k|^2)}{3a_1a_2} 
\cdot (\bar{v} z_2 d v \wedge d\bar{z}_2
+ v \bar{z}_2 dz_2 \wedge d \bar{v}) \\
& + & \frac{\prod_{k=1}^2(1+|z_k|^2)}{6a_1a_2} 
\cdot dv \wedge d \bar{v} \biggr\}.
\een
Clearly it can be extended smoothly over the zero section
given by $v=0$.
Therefore,
$\hat{\omega}_{a_1,a_2}$ gives a family of K\"ahler Ricci-flat
metric on the total space of $\cO_{\bP^1\times \bP^1}(-2,-2)$.

Now that we are working on $\cO_{\bP^1\times \bP^1}(-2,-2)$,
we need to modify the symplectic coordinates.
Write $z_j = \hat{r}_j e^{i\hat{\theta}_j}$, $j=1,2$, $v = \hat{r}_3e^{i\hat{\theta}_3}$.
Obviously, $\hat{\theta}_1 = \theta_1$, $\hat{\theta}-2 = \theta_2$, $\hat{\theta}_3 = 2\theta_3$,
Then
\be
\hat{\omega}_{a_1,a_2} = d\hat{y}_1 \wedge d\hat{\theta}_1 
+ d\hat{y}_2 \wedge d\hat{\theta}_2 
+ d\hat{y}_3 \wedge d\hat{\theta}_3,
\ee
where $\hat{y}_1, \hat{y}_2, \hat{y}_3$ are defined by:
\begin{align}
\hat{y}_1 & = y_1, & \hat{y}_2 & = y_2, & \hat{y}_3 & = \half y_3.
\end{align} 
Therefore,
the image of the moment map on $\cO_{\bP^1\times \bP^1}(-2,-2)$
is the convex domain in $\bR^3$ bounded by:
\begin{align}
\hat{y}_j & \geq -a_j, & \hat{y}_3 & \geq \half \hat{y}_j, \;\; j=1,2, & \hat{y}_3 & \geq 0.
\end{align}

\section{Summary}

\label{sec:Summary}

In this paper we have considered the Kepler manifolds $K_n$.
The following facts are well-known in the lieterature. 
As symplectic manifolds,
$K_n \cong T^*S^n - S^n$, with their natural symplectic structures;
as complex manifolds,
$K_n \cong \{ z_1^2 + \cdots + z_n^2 =0\} - \{0\}$,
and when $n=2$, $K_2 \cong (\bC^2 - \{0\})/\bZ_2 \cong \cO_{\bP^1}(-2) - \bP^1$ is a hypercomplex manifold,
and when $n=3$, $K_3 \cong (\cO_{\bP^1}\oplus \cO_{\bP^1}(-1))- \bP^1 \cong
\cO_{\bP^1 \times \bP^1}(-1,-1) -\bP^1 \times \bP^1$;
the complex structures and the symplectic structures on $K_n$ are compatible,
giving rise to K\"ahler structures called Kepler metrics on $K_n$.
The Kepler metric has K\"ahler potential $\half (|z_1|^2+ \cdots + |z_n|^2)^{1/2}$,
and this can be also used to define K\"ahler metrics on deformed conifolds.

Based on these facts,
in this paper we have presented the following results.
Kepler metrics are Sasaki metrics on the conormal bundles of the round spheres.
The Kepler metrics are related to Sasaki metrics
on  the unit  conormal bundles of the round spheres,
which are Sasakian.
When $n=2$ and $3$,
the Kepler metrics can be studied using the Hessian geometry,
arising in the setting of K\"ahler metrics with toric symmetries.
When $n=3$ we have the following diagram of related spaces:
$$
\xymatrix{
&  \cO_{\bP^1\times \bP^1}(-2,-2) \ar[d] & \\
                & \cO_{\bP^1\times \bP^1}(-1,-1) \ar[dl]_{\pi_1} \ar[dr]^{\pi_2}     \\
 \cO_{\bP^1}(-1) \oplus  \cO_{\bP^1}(-1)
 \ar[dr]  & &    \cO_{\bP^1}(-1) \oplus  \cO_{\bP^1}(-1)  \ar[dl]     \\
& \{z_1^2+\cdots + z_4^2=0\} \ar[d] & \\ 
& \{z_1^2+\cdots + z_4^2=a\} &
}
$$
and for $n=2$ we have a similar diagram.
We have explicit constructions of the pullback of Kepler metrics
or some K\"ahler Ricci-flat metrics.
We conjecture that these metrics are related to each other 
by special K\"ahler Ricci flow in a fashion similar to the one described 
in \S \ref{sec:KRF-Hessian}.
Surprisingly,
in all of our examples,
the complex potential functions take the following form:
\be
\psi = \sum_{i=1}^k l_i \ln(l_i),
\ee
where $l_i$ are some linear functions,
some of which are defining functions of the boundary
hyperplanes of the convex bodies arising as the image of the moment maps. 
This generalizes the formula in Guillemin \cite{Gui}, Abreu \cite{Abreu}
and Donaldson \cite{Donaldson}.
We will verify this formula for more examples in subsequent work.

The motivation for our investigations is to establish a link 
between the Kepler problem in classical gravity
and the modern geometric studies of string theory.
The basic principle is to apply the intrinsic symmetry of the problem.
In this paper we have related the geometry of Kepler problem
to Sasakian geometry and Hessian geometry
which have play important roles in the study of ADS/CFT 
correspondence (see e.g. Martelli-Sparks-Yau \cite{Mar-Spa-Yau}).
A big difference from  string theory is that we are using Hessian geometry 
to directly 
produce the phase space of the dynamical system of the Kepler problem,
while in string theory
one first obtains some spaces in extra dimensions for compactifications.
We will continue this line of research in later parts of this series
of papers.

\vspace{.2in}
{\em Acknowledgements}.
The research in this work is partially supported by NSFC grant 11661131005
for Russia-China Collaborations on Integrable Systems in Mathematical Physics
and Differential Geometry.
The author's interest in the geometric aspects of classical dynamical systems 
was stimulated at Dynamics in Siberia 2017.
He thanks the organizers and participants for the hospitality enjoyed 
at this conference.
Some of the materials in this paper were lectured on in a graduate course 
in Tsinghua University in the spring semester of 2017.
The author thanks the attendants of this course for sharing their enthusiasms.   
Finally, special thanks are due to Guowu Meng for introducing the 
author to the Kepler problem many years ago.


\begin{thebibliography}{999}


\bibitem{Abreu}
M. Abreu, 
{\em K\"ahler geometry of toric varieties and extremal metrics}, 
International J. Math. 9 (1998) 641-651.

\bibitem{Auroux}
D. Auroux,
{\em Special Lagrangian fibrations, wall-crossing, and mirror symmetry}, 
Surv. Differ. Geom., vol. 13, Int. Press, Somerville, MA, 2009, pp. 1-47.

\bibitem{Bac}
H. Bacry, {\em The de Sitter group L4,1 and the bound states of hydrogen atom}.
Nuovo Cimento, 41A (1966), pp. 222-234.

\bibitem{Bac-Rue-Sou}
H. Bacry, H. Ruegg, J.-M. Souriau,
{\em Dynamical groups and spherical potentials
in classical mechanics},
Comm. Math. Phys. 3 (1966), pp. 323--333.

\bibitem{Ban-Itz}
M. Bander, C. Itzkson, 
{\em Group theory and the hydrogen atom (I)}.
Rev. Mod. Phys. 38 (1966), pp. 330-345.

\bibitem{Bel}
E. A. Belbruno, 
{\em Two body motion under the inverse square central force and equivalent geodesic flows},
Celest. Mech., 15 (1977) 467-476.

\bibitem{Boy-Gal}
C. P. Boyer, K. Galicki, 
{\em On Sasakian-Einstein Geometry}, 
Internat. J. Math. 11 (2000), no. 7, 873-909.

\bibitem{Calabi}
E. Calabi, 
{\em  M\'etriques k\"ahl\'eriennes et fibr\'es holomorphes},
Annales Scientifiques de l'Ecole Normale Superieure 12 (1979),
pp. 268-294\emph{}.


\bibitem{Car}
M. Cariglia, 
{\em Conformal triality of the Kepler problem}. 
J. Geom. Phys.  106  (2016), 205-209. 


\bibitem{Chan-Lau-Leung}
K. Chan, S.-C. Lau, N.C. Leung, 
{\em SYZ mirror symmetry for toric Calabi-Yau manifolds}. 
J. Differential Geom.  90  (2012),  no. 2, 177-250.

\bibitem{Cor}
B. Cordani, 
{\em The Kepler problem. Group theoretical aspects, regularization and quantization,
 with application to the study of perturbations}. 
 Progress in Mathematical Physics, 29. Birkh\"auser Verlag, Basel, 2003.

\bibitem{Donaldson}
S.K. Donaldson,
{\em Scalar curvature and stability of toric varieties},
J. Diff. Geom. 62 (2002), 289-349.


\bibitem{Duan-Zhou1}
X. Duan, J. Zhou,
{\em Rotationally symmetric pseudo-K\"ahler-Einstein metrics}. 
Front. Math. China  6  (2011),  no. 3, 391-410.

\bibitem{Duan-Zhou2}
X. Duan, J. Zhou, 
{\em Rotationally symmetric pseudo-K\"ahler metrics 
of constant scalar curvatures}. 
Sci. China Math.  54  (2011),  no. 5, 925-938. 

\bibitem{Egu-Han}
T. Eguchi, A.J. Hanson, 
{\em Self-dual solutions to Euclidean gravity}. 
Ann. Physics  120  (1979), no. 1, 82-106. 

\bibitem{Eng}
M. J. Englefield,
{\em Group theory and the Coulomb problem}. Wiley-Interscience, New York-London-Sydney, 1972.

\bibitem{Fock}
V.A. Fock,
{\em Zur theory des wasserstoffatoms}. Z. Phys., 98 (1935), pp. 145-154.


\bibitem{Gui}
V. Guillemin, 
{\em Kaehler structures on toric varieties}, J. Differential Geom. 40
(1994) 285-309.

\bibitem{Gui-Ste77}
V. Guillemin, S. Sternberg,
{\em Geometric Asymptotics}. AMS, Providence, RI, (1977).

\bibitem{Gui-Ste90}
V. Guillemin, S. Sternberg,
{\em Variations on a Theme by Kepler}. 
American Mathematical Society, Providence, RI, (1990).

\bibitem{Gyo}
G. Gy\"orgyi,
{\em Kepler's equation, Fock variables, Bacry's generators and Dirac
brackets}. Nuovo Cimento, 53A (1968), pp. 717-735.

\bibitem{Ham} W. R. Hamilton, 
{\em The hodograph or a new method of expressing in symbolic language the Newtonian law of
attraction}, Proc. Roy. Irish Acad., 3 (1846) 344-353 (Math. Papers v. 2, 287-294, Camb. Univ.
Press. 1942).

\bibitem{Hec-deL}
G. Heckman, T. de Laat,
{\em On the regularization of the Kepler Problem},
 J. Symp. Geom. 10 (2012), 463-473.

\bibitem{Kea}
A.J. Keane, 
{\em Spectrum generating algebras for the classical Kepler problem}. 
J. Phys. A  35  (2002),  no. 38, 8083-8108.

\bibitem{Kea-Bar}
A.J. Keane, R.K. Barrett, 
{\em The conformal group SO(4,2)  and Robertson-Walker spacetimes}. 
Classical Quantum Gravity  17  (2000),  no. 1, 201-218. 

\bibitem{Kea-Bar-Sim}
A.J. Keane, R.K. Barrett, J.F. Simmons, 
{\em The classical Kepler problem and geodesic motion on spaces of constant curvature}. 
J. Math. Phys.  41  (2000),  no. 12, 8108-8116.


\bibitem{KS}
P. Kustaanheimo, E. Stiefel, 
{\em Perturbation theory of Kepler motion based on spinor regutarization}. 
J. Reine Angew. Math. 218 (1965), 204-219.

\bibitem{LeBrun}
C. LeBrun, {\em Counter-Examples to the
Generalized Positive Action Conjecture}, 
Commun.Math.Phys.118 (1988),
pp. 591-596.

\bibitem{LC}
T. Levi-Civita,
{\em Sur la resolution qualitative du probleme restreint des trois corps}. 
Acta Math. 30 (1906), pp. 305-327.

\bibitem{Lig-Sch}
T. Ligon, M. Schaaf,
{\em On the global symmetry of the classical Kepler problem}.
Rep. Math. Phys., 9 (1976), pp. 281-300.

\bibitem{Mar-Spa-Yau}
D. Martelli, J. Sparks, S.-T. Yau, 
{\em The geometric dual of a-maximisation for toric
Sasaki-Einstein manifolds}, Commun. Math. Phys. 268 (2006), 39-65.

\bibitem{Meng}
G. Meng, 
{\em Generalized Kepler problems and euclidean Jordan algebras}.  
Geometry, integrability and quantization XVII,  72¨C94, 
Bulgar. Acad. Sci., Sofia, 2016.

\bibitem{MIC}
H. McIntosh A. Cisneros, 
{\em Degeneracy in the presence of a magnetic monopole},
J. Math. Phys. 11 (1970) 896-916.

\bibitem{Mil}
J.Milnor,
{\em On the Geometry of the Kepler Problem}, Amer. Math.
Monthly 90:6 (1983), 353-365.

\bibitem{Moser}
J. Moser, 
{\em Regularization of Kepler's problem and the averaging method on a manifold}, 
Comm. Pure App.
Math., 23 (1970) 609-636.

\bibitem{Osi}
Yu. S. Osipov, 
{\em The Kepler problem and geodesic flows in spaces of constant curvature},
Celestial Mechanics 16 (1977) pp. 191-208.


\bibitem{Pauli}
W. Pauli,
{\em \"Uber das wasserstoffspektrum vom standpunkt der neuen Quantummechanic}. 
Z. Phys., 36 (1926), pp. 336-363. 

\bibitem{Rawnsley}
J. Rawnsley, 
{\em Coherent states and K¡§ahler manifolds}, Quart. J. Math. Oxford (2) 28 (1977),
403-415.

\bibitem{Sie}
Siegel, C. L.; Moser, J. K. 
{\em Lectures on celestial mechanics}. 
Translated from the German by C. I. Kalme. 
Reprint of the 1971 translation. Classics in Mathematics. Springer-Verlag, Berlin, 1995. 

\bibitem{Souriau74}
J.-M. Souriau, 
{\em Sur la vari\'et\'e de Kepler}. Symp. Math., 14, pp. 343-360, (1974).


\bibitem{Souriau83}
J.-M. Souriau,
{\em G\'eometrie globale du probl\'eme \'a deux corps}. 
Proc. IUTAMISSIM Symp. on Mod. Devl. Anal. Mech., Atti Acad.Sci. Torino, Suppl. 117,
pp. 369-418, (1983).

\bibitem{Souriau97}
J.-M. Souriau,
{\em Structure ofDynamical Systems}. Birkhauser, Boston, (1997).


\bibitem{Sparks}
J. Sparks, {\em Sasaki-Einstein manifolds}.  
Surveys in differential geometry. Volume XVI. Geometry of special holonomy and related topics,  265¨C324, 
Surv. Differ. Geom., 16, Int. Press, Somerville, MA, 2011.


\bibitem{Stenzel}
M.B. Stenzel, 
{\em Ricci-flat metrics on the complexification of 
a compact rank one symmetric space}. 
Manuscripta Math.  80  (1993),  no. 2, 151-163.

\bibitem{SYZ}
A. Strominger, S.-T. Yau, E. Zaslow, 
{\em Mirror symmetry is T-duality}, Nuclear
Phys. B 479 (1996), no. 1-2, 243-259.

 
\bibitem{Sun}
K.F. Sundman, 
{\em M\'{e}moire sur le probl\`{e}me des trois corps}. 
Acta Math.  36  (1913),  no. 1, 105-179. 

\bibitem{Z}
D. Zwanziger, 
{\em Exactly soluble nonrelativistic model of particles with both electric
and magnetic charges}, Phys. Rev. 176 (1968) 1480-1488.

\end{thebibliography}
\end{document}